\documentclass[11pt]{article}
\usepackage{amsmath,amsfonts,amsthm,amssymb,graphicx,bm}
\usepackage{subfigure}
\usepackage[font=small,labelfont=bf]{caption}
\usepackage{algorithm}
\usepackage{algorithmic}
\usepackage{authblk}
\usepackage{multirow}
\usepackage{natbib}

\theoremstyle{plain}
\newtheorem{theorem}{Theorem}
\newtheorem{proposition}{Proposition}

\newtheorem{assumption}{Assumption}

\begin{document}

\title{A Cheap Bootstrap Method for Fast Inference}
\author{Henry Lam\thanks{Email: \texttt{henry.lam@columbia.edu}}}
\affil{Department of Industrial Engineering and Operations Research, Columbia University}
\date{}
\maketitle

\begin{abstract}


The bootstrap is a versatile inference method that has proven powerful in many statistical problems. However, when applied to modern large-scale models, it could face substantial computation demand from repeated data resampling and model fitting. We present a bootstrap methodology that uses minimal computation, namely with a resample effort as low as \emph{one} Monte Carlo replication, while maintaining desirable statistical guarantees. We present the theory of this method that uses a twisted perspective from the standard bootstrap principle. We also present generalizations of this method to nested sampling problems and to a range of subsampling variants, and illustrate how it can be used for fast inference across different estimation problems.
\end{abstract}

\section{Introduction}
The bootstrap is a versatile method for statistical inference that has proven powerful in many problems. It is advantageously data-driven and automatic: Instead of using analytic calculation based on detailed model knowledge such as the delta method, the bootstrap uses resampling directly from the data as a mechanism to approximate the sampling distribution. The flexibility and ease of the bootstrap have made it widely popular; 
see, e.g., the monographs \cite{efron1994introduction,davison1997bootstrap,shao2012jackknife,hall1988bootstrap} for comprehensive reviews on the subject.


Despite its popularity, an arguable challenge in using the bootstrap is its computational demand. To execute the bootstrap, one needs to run enough Monte Carlo replications to generate resamples and refit models. While this is a non-issue in classical problems, for modern large-scale models and data size, even one model refitting could incur substantial costs. This is the case if the estimation or model fitting involves big optimization or numerical root-finding problems arising from, for instance, empirical risk minimization in machine learning, where the computing routine can require running time-consuming stochastic gradient descent or advanced mathematical programming procedures. Moreover, in these problems, what constitutes an adequate Monte Carlo replication size can be case-by-case and necessitate trial-and-error, which adds to the sophistication. Such computation issues in applying resampling techniques to massive learning models have been discussed and motivate some recent works, e.g., \cite{10.5555/3042573.3042801,lu2020uncertainty,giordano2019swiss,schulam2019can,alaa2020discriminative}.

Moreover, an additional challenge faced by the bootstrap is that its implementation can sometimes involve a \emph{nested} amount of Monte Carlo runs. This issue arises for predictive models that are corrupted by noises from both the data and the computation procedure. An example is predictive inference using high-fidelity stochastic simulation modeling (\cite{nelson2013foundations,law1991simulation}). These models are used to estimate output performance measures in sophisticated systems, and comprises a popular approach for causal decision-making in operations management and scientific disciplines. In this context, obtaining a point estimate of the performance measure itself requires a large number of Monte Carlo runs on the model to average out the aleatory uncertainty. When the model contains parameters that are calibrated from exogenous data, conducting bootstrap inference on the prediction would involve first resampling the exogenous data, then given each resampled input parameter value, running and averaging a large number of model runs. This consequently leads to a multiplicative total computation effort between the resample size and the model run size per resample, an issue known as the input uncertainty problem in simulation (see, e.g., \cite{henderson2003input,song2014advanced,barton2012tutorial,lam2016advanced}). Here, to guarantee consistency, the sizes in the two sampling layers not only need to be sufficiently large, but also depend on the input data size which adds more complication to their choices (\cite{lam2021subsampling}). Other than simulation modeling, similar phenomenon arises in some machine learning models. For example, a bagging predictor (\cite{breiman1996bagging,wager2014confidence,mentch2016quantifying}) averages a large number of base learning models each trained from a resample, and deep ensemble (\cite{lakshminarayanan2016simple,lee2015m}) averages several neural networks each trained with a different randomization seed (e.g., in initializing a stochastic gradient descent). Running bootstraps for the prediction values of these models again requires a multiplicative effort of data resampling and, given each resample, building a new predictor which requires multiple training runs.

Motivated by these challenges, our goal in this paper is to study a statistically valid bootstrap method that allows using very few resamples. More precisely, our method reduces the number of bootstrap resamples to the minimum possible, namely as low as \emph{one} Monte Carlo replication, while still delivering asymptotically exact coverage as the data size grows. This latter guarantee holds under essentially the same regularity conditions for conventional bootstrap methods. As a result, our method can give valid inference under constrained budget when existing bootstrap approaches fail. Moreover, the validity of our inference under any number of resamples endows more robustness, and hence alleviates the tuning effort, in choosing the number of resamples, which as mentioned before could be a trial-and-error task in practice. For convenience, we call our method the \emph{Cheap Bootstrap}, where ``Cheap'' refers to our low computation cost, and also to the correspondingly lower-quality half-width performance when using the low cost as we will see.

In terms of theory, the underpinning mechanism of our Cheap Bootstrap relies on a simple twist to the basic principle used in conventional bootstraps. The latter dictates that the sampling distribution of a statistic can be approximated by the resample counterpart, conditional on the data. To utilize this principle for inference, one then typically generates many resample estimates to obtain their approximate distribution. Our main idea takes this principle a step further by leveraging the following: The resemblance of the resampling distribution to the original sampling distribution, universally regardless of the data realization, implies the asymptotic independence between the original estimate and \emph{any} resample estimate. Thus, under asymptotic normality, we can construct pivotal statistics using the original estimate and any number of resample estimates that cancel out the standard error as a nuisance parameter, which subsequently allows us to conduct inference using as low as one resample.


In addition to the basic asymptotic exact coverage guarantee, we study Cheap Bootstrap in several other aspects. First regards the half-width of the Cheap Bootstrap confidence interval, a statistical efficiency measure in addition to coverage performance. We show how the Cheap Bootstrap interval follows the behavior of a $t$-interval, which is wide when we use only one resample but shrinks rapidly as the number of resamples increases. Second, we investigate the higher-order coverage of Cheap Bootstrap via Edgeworth expansions. Different from conventional bootstrap procedures, the Cheap Bootstrap pivotal statistic has a limiting $t$-distribution, not a normal distribution. Edgeworth expansion on $t$-distribution is, to our best knowledge, open in the literature (except the recent working paper \cite{he2021higher}). Here, we build the Edgeworth expansions for Cheap Bootstrap coverage probabilities and show their errors match the order $O(n^{-1})$ (where $n$ is the sample size) incurred in conventional bootstraps in the two-sided case. Moreover, we explicitly identify the coefficient in this error term, which is expressible as a high-dimensional integral that, although cannot be evaluated in closed-form, is amenable to Monte Carlo integration.

We also generalize the Cheap Bootstrap in several directions. First, we devise the Cheap Bootstrap confidence intervals for models that encounter both data and computation noises and hence the aforementioned nested sampling issue. When applying the Cheap Bootstrap to these problems, the number of outer resamples can be driven down to one or a very small number, and thus the total model run size is no longer multiplicatively big. On the other hand, because of the presence of multiple sources of noises, the joint distribution between the original estimate and the resample estimate is no longer asymptotically independent, but exhibits a dependent structure that nonetheless could be tractably exploited. Second, our Cheap Bootstrap can also be used in conjunction with ``subsampling'' variants. By subsampling here we broadly refer to schemes that resample data sets of size smaller than the original size, such as the $m$-out-of-$n$ Bootstrap (\cite{politis1999subsampling,bickel2012resampling}), and more recently the Bags of Little Bootstraps (\cite{kleiner2014scalable}) and Subsampled Double Bootstrap (\cite{sengupta2016subsampled}). We show how these procedures can all be implemented ``cheaply". Finally, we also extend the Cheap Bootstrap to multivariate problems.




We close this introduction by discussing other related works. Regarding motivation, our study appears orthogonal to most of the classical bootstrap literature, as its focus is on higher-order coverage accuracy, using techniques such as studentization (\cite{davison1997bootstrap} \S2.4), calibration or iterated bootstrap (\cite{hall2013bootstrap} \S3.11; \cite{hall1988bootstrap,hall1986bootstrap,beran1987prepivoting}), and bias-correction and acceleration (\cite{efron1982jackknife} \S10.4; \cite{efron1987better,diciccio1996bootstrap}). The literature also studies the relaxation of assumptions to handle non-smooth functions and dependent data (\cite{politis1999subsampling,bickel2012resampling}). In contrast to these studies, we focus on the capability to output confidence intervals under extremely few resamples, and the produced intervals are admittedly less accurate than the refined intervals in the literature constructed under idealized infinite resamples.


Our study is related to \cite{hall1986number}, which shows a uniform coverage error of $O(n^{-1})$ for one-sided intervals regardless of the number of resamples $B$, as long as the nominal coverage probability is a multiple of $(B+1)^{-1}$. From this, \cite{hall1986number} suggests, for $95\%$ interval for instance, that it suffices to use $B=19$ to obtain a reasonable coverage accuracy. Our suggestion in this paper drives this number down to $B=1$, which is made possible via our use of asymptotic independence and normality instead of order-statistic analyses on the quantile-based bootstrap in \cite{hall1986number}. Our result on a small $B$ resulting in a large interval width matches the insight in \cite{hall1986number}, though we make this tradeoff more explicit by looking at the width's moments. Other related works on bootstrap computation cost include \cite{booth1994monte} and \cite{lee1995asymptotic} that analyze or reduce Monte Carlo sizes by replacing sampling with analytical calculation when applying the iterated bootstrap, and \cite{booth1993simple} and \cite{efron1994introduction} \S23 on variance reduction. These works, however, do not consider the exceedingly low number of resamples that we advocate in this paper.

The remainder of this paper is as follows. Section \ref{sec:basic} introduces the Cheap Bootstrap. Section \ref{sec:theory} presents the main theoretical results on asymptotic exactness and higher-order coverage error, and discusses half-width behaviors. Section \ref{sec:double} generalizes the Cheap Bootstrap to nested sampling problems, and Section \ref{sec:subsampling} to subsampling. Section \ref{sec:numerics} shows numerical results to demonstrate our method and compare with benchmarks. The Appendix details proofs that are not included in the main body, presents additional theoretical and numerical results, and documents some technical background materials.

\section{Cheap Bootstrap Confidence Intervals}\label{sec:basic}
We aim to construct a confidence interval for a target parameter $\psi:=\psi(P)$, where $P$ is the data distribution, and $\psi:\mathcal P\to\mathbb R$ is a functional with $\mathcal P$ as the set of all distributions on the data domain. This $\psi(P)$ can range from simple statistical summaries such as correlation coefficient, quantile, conditional value-at-risk, to model parameters such as regression coefficient and prediction error measurement.

Suppose we are given independent and identically distributed (i.i.d.) data of size $n$, say $X_1,\ldots,X_n$. A natural point estimate of $\psi(P)$ is $\hat\psi_n:=\psi(\hat P_n)$, where $\hat P_n(\cdot):=(1/n)\sum_{i=1}^nI(X_i\in\cdot)$ is the empirical distribution constructed from the data, and $I(\cdot)$ denotes the indicator function.

Our approach to construct a confidence interval for $\psi$ proceeds as follows. For each replication $b=1,\ldots,B$, we resample the data set, namely independently and uniformly sample with replacement from $\{X_1,\ldots,X_n\}$ $n$ times, to obtain $\{X_1^{*b},\ldots,X_n^{*b}\}$, and evaluate the resample estimate $\psi_n^{*b}:=\psi(P_n^{*b})$, where $P_n^{*b}(\cdot)=(1/n)\sum_{i=1}^nI(X_i^{*b}\in\cdot)$ is the resample empirical distribution. Our confidence interval is
\begin{equation}
\mathcal I=\left[\hat\psi_n-t_{B,1-\alpha/2}S,\ \hat\psi_n+t_{B,1-\alpha/2}S\right]\label{CI}
\end{equation}
where
\begin{equation}
S^2=\frac{1}{B}\sum_{b=1}^B\left(\psi_n^{*b}-\hat\psi_n\right)^2\label{var}
\end{equation}
Here, $S^2$ resembles the sample variance of the resample estimates, but ``centered'' at the original point estimate $\hat\psi_n$ instead of the resample mean, and using $B$ in the denominator instead of $B-1$ as in ``textbook'' sample variance. The critical value $t_{B,1-\alpha/2}$ is the $(1-\alpha/2)$-quantile of $t_B$, the student $t$-distribution with degree of freedom $B$. That is, the degree of freedom of this $t$-distribution is precisely the resampling computation effort.

The interval $\mathcal I$ in \eqref{CI} is defined for any positive integer $B\geq1$. In particular, when $B=1$, it becomes
\begin{equation}
\left[\hat\psi_n-t_{1,1-\alpha/2}\left|\psi_n^{*}-\hat\psi_n\right|,\ \hat\psi_n+t_{1,1-\alpha/2}\left|\psi_n^{*}-\hat\psi_n\right|\right]\label{CI single}
\end{equation}
where $\psi_n^*$ is the single resample estimate.


The form of \eqref{CI} looks similar to the so-called standard error bootstrap as we explain momentarily, especially when $B$ is large. However, the point here is that $B$ does not need to be large. In fact, we have the following basic coverage guarantee for \eqref{CI} and \eqref{CI single}. First, consider the following condition that is standard in the bootstrap literature:
\begin{assumption}
[Standard condition for bootstrap validity]
We have $\sqrt n(\hat\psi_n-\psi)\Rightarrow N(0,\sigma^2)$ where $\sigma^2>0$. Moreover, a resample estimate $\psi_n^*$ satisfies $\sqrt n(\psi_n^*-\hat\psi_n)\Rightarrow N(0,\sigma^2)$ conditional on the data $X_1,X_2,\ldots$ in probability as $n\to\infty$.\label{bp}
\end{assumption}
In Assumption \ref{bp},  ``$\Rightarrow$'' denotes convergence in distribution, and the conditional ``$\Rightarrow$''-convergence in probability means $P(\sqrt n(\psi_n^*-\hat\psi_n)\leq x|\hat P_n)\stackrel{p}{\to}P(N(0,\sigma^2)\leq x)$ for any $x\in\mathbb R$, where ``$\stackrel{p}{\to}$'' denotes convergence in probability. Assumption \ref{bp} is a standard condition to justify bootstrap validity, and is ensured when $\psi(\cdot)$ is Hadamard differentiable (see Proposition \ref{HD} in the sequel which follows from \cite{van2000asymptotic} \S23). This assumption implies that, conditional on the data, the asymptotic distributions of the centered resample estimate $\sqrt n(\psi_n^*-\hat\psi_n)$ and the centered original estimate $\sqrt n(\hat\psi_n-\psi)$ are the same. Thus, one can use the former distribution, which is computable via Monte Carlo, to approximate the latter unknown distribution. Simply put, we can use a ``plug-in" of $\hat P_n$ in place of $P$, namely $\xi(\hat\psi_n,\psi_n^*)$, to approximate $\xi(\psi,\hat\psi_n)$ for any suitable data-dependent quantity $\xi(\cdot,\cdot)$.

Under Assumption \ref{bp}, \eqref{CI} is an asymptotically exact $(1-\alpha)$-level confidence interval:
\begin{theorem}[Asymptotic exactness of Cheap Bootstrap]
Under Assumption \ref{bp}, for any $B\geq1$, the interval $\mathcal I$ in \eqref{CI} is an asymptotically exact $(1-\alpha)$-level confidence interval for $\psi$, i.e.,
$$\mathbb P_n(\psi\in\mathcal I)\to1-\alpha$$ as $n\to\infty$,  where $\mathbb P_n$ denotes the probability with respect to the data $X_1,\ldots,X_n$ and all randomness from the resampling.\label{main}
\end{theorem}

Theorem \ref{main} states that, under the same condition to justify the  validity of conventional bootstraps, the Cheap Bootstrap interval $\mathcal I$ has asymptotically exact coverage, for \emph{resample effort as low as 1}.
Before we explain how Theorem \ref{main} is derived, we first compare the Cheap Bootstrap to conventional bootstraps.

\subsection{Comparisons with Established Bootstrap Methods}\label{sec:comparisons}

The commonest bootstrap approach for interval construction uses the quantiles of resample estimates to calibrate the interval limits. This includes the basic bootstrap and (Efron's) percentile bootstrap (\cite{efron1994introduction,davison1997bootstrap}). In the basic bootstrap, we first run Monte Carlo to approximate the $\alpha/2$ and $(1-\alpha/2)$-quantiles of $\psi_n^{*}-\hat\psi_n$, say $q_{\alpha/2}$ and $q_{1-\alpha/2}$. Assumption \ref{bp} guarantees that $q_{\alpha/2}$ and $q_{1-\alpha/2}$ approximate the corresponding quantiles of $\hat\psi_n-\psi$, thus giving $[\hat\psi_n-q_{1-\alpha/2},\hat\psi_n-q_{\alpha/2}]$ as a $(1-\alpha)$ confidence interval for $\psi$ (or equivalently using the $\alpha/2$ and $(1-\alpha/2)$ quantiles of $2\hat\psi_n-\psi_n^{*}$). The percentile bootstrap directly uses the $\alpha/2$ and $(1-\alpha/2)$-quantiles of $\psi_n^{*}$ as the interval limits, i.e., $[q_{\alpha/2},q_{1-\alpha/2}]$, justified by additionally the symmetry of the asymptotic distribution. Alternately, one can bootstrap the standard error and plug into a normality confidence interval (\cite{efron1981nonparametric}). That is, we compute $Var_*(\psi^*_n)$ to approximate $Var(\hat\psi_n)$ (where $Var$ and $Var_*$ denote respectively the variance with respect to the data and the conditional variance with respect to a resample drawn from $\hat P_n$), and output $[\hat\psi_n\pm z_{1-\alpha/2}\sqrt{Var_*(\psi^*)}]$ as the confidence interval, with $z_{1-\alpha/2}$ being the $(1-\alpha/2)$-quantile of standard normal.



All approaches above require generating enough resamples. When $B$ is large, the $S^2$ defined in \eqref{var} satisfies $S^2\approx Var_*(\psi^*_n)$, and moreover $t_{B,1-\alpha/2}\approx z_{1-\alpha/2}$. Thus in this case the Cheap Bootstrap interval $\mathcal I$ becomes $\hat\psi\pm z_{1-\alpha/2}\sqrt{Var_*(\psi^*_n)}$, which is nothing but the standard error bootstrap. In other words, the Cheap Bootstrap can be viewed as a generalization of the standard error bootstrap to using any $B$.


We also contrast our approach with the studentized bootstrap (e.g., \cite{davison1997bootstrap} \S2.4), which resamples pivotal quantities in the form $(\hat\psi_n-\psi)/\hat\sigma_n$ where $\hat\sigma_n$ is a sample standard error (computed from $\hat P_n$). As widely known, while this approach bears the term ``studentized'', it does not use the $t$-distribution, and is motivated from better higher-order coverage accuracy. Moreover, attaining such an interval could require additional computation to resample the standard error. Our approach is orthogonal to this studentization in that we aim to minimize computation instead of expanding it for the sake of attaining higher-order accuracy.

Lastly, \eqref{CI} and \eqref{CI single} have natural analogs for one-sided intervals, where we use
\begin{equation}
\mathcal I_{lower}=\left(-\infty,\ \hat\psi_n+t_{B,1-\alpha}S\right]\text{\ \ or\ \ }
\mathcal I_{upper}=\left[\hat\psi_n-t_{B,1-\alpha}S,\ \infty\right).\label{CI upper}
\end{equation}
Theorem \ref{main} applies to \eqref{CI upper} under the same assumption. In the one-sided case, \cite{hall1986number} has shown an error of $O(n^{-1})$, uniformly in $B$, for the studentized bootstrap when the nominal coverage probability $(1-\alpha)$ is a multiple of $(B+1)^{-1}$. This translates, in the case of $95\%$ interval, a minimum of $B=19$. Our Theorem \ref{main} drives down this suggestion of \cite{hall1986number} for the studentized bootstrap to $B=1$ for the Cheap Bootstrap.


\subsection{A Basic Numerical Demonstration}\label{sec:basic numerics}
We numerically compare our Cheap Bootstrap with the conventional approaches, namely the basic bootstrap, percentile bootstrap and standard error bootstrap described above. We use a basic example on estimating a $95\%$ confidence interval of the $0.6$-quantile of, say, an exponential distribution with unit rate from i.i.d. data. This example can be handled readily by other means, but it demonstrates how Cheap Bootstrap can outperform baselines under limited replication budget.




We use a data size $n=100$, and vary $B=1,2,5,10,50$. For each $B$, we generate synthetic data $1000$ times, each time running all the competing methods, and outputting the empirical coverage and interval width statistics from the $1000$ experimental repetitions. Table \ref{table:comparison} shows that when $B=1$, Cheap Bootstrap already gives confidence intervals with a reasonable coverage of $92\%$, while all other bootstrap methods fail because they simply cannot operate with only one resample (basic and percentile bootstraps cannot output two different finite numbers as the upper and lower interval limits, and standard error bootstrap has a zero denominator in the formula). When $B=2$, all baseline methods still have poor performance, as $B=2$ is clearly still too small a size to possess any statistical guarantees. In contrast, Cheap Bootstrap gives a similar coverage of $93\%$ as in $B=1$, and continues to have stable coverages when $B$ increases. As $B$ increases through 5 to 50, the baseline methods gradually catch up on the coverage level.

\begin{table}[ht]
\caption{Interval performances with different bootstrap methods, at nominal confidence level $95\%$ and with sample size $n=100$.}
\centering
\begin{tabular}{cccc}
 & Replication & Empirical coverage & Width mean \\
 & size $B$ & (margin of error) & (standard deviation)\\
   \hline
  Cheap Bootstrap & 1 & 0.92\ \ (0.02) & 2.42\ \ (2.06) \\
  Basic Bootstrap & 1 & NA & NA \\
  Percentile Bootstrap & 1 & NA & NA \\
  Standard Error Bootstrap & 1 & NA & NA \\
   \hline
  Cheap Bootstrap & 2 & 0.93\ \ (0.02) & 0.95\ \ (0.60) \\
  Basic Bootstrap & 2 & 0.32\ \ (0.03) & 0.14\ \ (0.12) \\
  Percentile Bootstrap & 2 & 0.32\ \ (0.03) & 0.14\ \ (0.12) \\
  Standard Error Bootstrap & 2 & 0.67\ \ (0.03) & 0.38\ \ (0.33) \\
  \hline
  Cheap Bootstrap & 5 & 0.92\ \ (0.02) & 0.63\ \ (0.28) \\
  Basic Bootstrap & 5 & 0.62\ \ (0.03) & 0.29\ \ (0.14) \\
  Percentile Bootstrap & 5 & 0.69\ \ (0.03) & 0.29\ \ (0.14) \\
  Standard Error Bootstrap & 5 & 0.86\ \ (0.02) & 0.47\ \ (0.22) \\
  \hline
  Cheap Bootstrap & 10 & 0.92\ \ (0.02) & 0.53\ \ (0.20) \\
  Basic Bootstrap & 10 & 0.73\ \ (0.03) & 0.37\ \ (0.14) \\
  Percentile Bootstrap & 10 & 0.80\ \ (0.02) & 0.37\ \ (0.14) \\
  Standard Error Bootstrap & 10 & 0.88\ \ (0.02) & 0.46\ \ (0.17) \\
  \hline
  Cheap Bootstrap & 50 & 0.94\ \ (0.02) & 0.50\ \ (0.13) \\
  Basic Bootstrap & 50 & 0.86\ \ (0.02) & 0.44\ \ (0.12) \\
  Percentile Bootstrap & 50 & 0.92\ \ (0.02) & 0.44\ \ (0.12) \\
  Standard Error Bootstrap & 50 & 0.93\ \ (0.02) & 0.48\ \ (0.13) \\
\end{tabular}
\label{table:comparison}
\end{table}

Regarding interval length, we see that the Cheap Bootstrap interval shrinks as $B$ increases, sharply for small $B$ and then stabilizes. In particular, the length decreases by $1.47$ when $B$ increases from 1 to 2, compared to a much smaller $0.03$ when $B$ increases from 10 to 50. Though the length of Cheap Bootstrap is always larger than the baselines, it closes in at $B=10$ and even more so at $B=50$. Both the good coverage starting from $B=1$ for Cheap Bootstrap, and the sharp decrease of its interval length at very small $B$ and then level-off will be explained by our theory next.

\section{Theory of Cheap Bootstrap}\label{sec:theory}
We describe the theory of Cheap Bootstrap, including its asymptotic exactness (Section \ref{sec:exact}), higher-order coverage error (Section \ref{sec:higher}) and interval width behavior (Section \ref{sec:tightness}).
\subsection{Asymptotic Exactness}\label{sec:exact}
We present the proof of Theorem \ref{main} on the asymptotic exactness of Cheap Bootstrap for any $B\geq1$.
The first key ingredient is the following joint asymptotic characterization among the original estimate and the resample estimates:
\begin{proposition}[Asymptotic independence and normality among original and resample estimates]
Under Assumption \ref{bp}, we have, for the original estimate $\hat\psi_n$ and resample estimates $\psi_n^{*b},b=1,\ldots,B$,
\begin{equation}
\sqrt n(\hat\psi_n-\psi,\psi^{*1}_n-\hat\psi_n,\ldots,\psi^{*B}_n-\hat\psi_n)\Rightarrow(\sigma Z_0,\sigma Z_1,\ldots,\sigma Z_B)\label{joint}
\end{equation}
where $Z_b,b=0,\ldots,B$ are i.i.d. $N(0,1)$ variables.\label{joint prop}
\end{proposition}
The convergence \eqref{joint} entails that, under $\sqrt n$-scaling, the centered resample estimates and the original point estimate are all asymptotically independent, and moreover are distributed according to the same normal, with the only unknown being $\sigma$ that captures the asymptotic standard error. The asymptotic independence is thanks to the universal limit of a resample estimate $\sqrt n(\psi_n^{*b}-\hat\psi_n)$ conditional on any data sequence, as detailed below.

\begin{proof}[Proof of Proposition \ref{joint prop}]
Denote $\Phi_\sigma(\cdot)$ as the distribution function of $N(0,\sigma^2)$. To show the joint weak convergence of $\sqrt n(\hat\psi_n-\psi)$ and $\sqrt n(\psi_n^{*b}-\hat\psi_n)$, $b=1,\ldots,B$, to i.i.d. $N(0,\sigma^2)$ variables, we will show, for any $x_b\in\mathbb R,b=0,\ldots,B$,
$$P\left(\sqrt n(\hat\psi_n-\psi)\leq x_0,\sqrt n(\psi_n^{*1}-\hat\psi_n)\leq x_1,\ldots,\sqrt n(\psi_n^{*B}-\hat\psi_n)\leq x_B\right)\to\prod_{b=0}^B\Phi_\sigma(x_b)$$
as $n\to\infty$. To this end, we have
\begin{eqnarray}
&&\left|P\left(\sqrt n(\hat\psi_n-\psi)\leq x_0,\sqrt n(\psi_n^{*1}-\hat\psi_n)\leq x_1,\ldots,\sqrt n(\psi_n^{*B}-\hat\psi_n)\leq x_B\right)-\prod_{b=0}^B\Phi_\sigma(x_b)\right|\notag\\
&=&\left|E\left[I(\sqrt n(\hat\psi_n-\psi)\leq x_0)P(\sqrt n(\psi_n^{*1}-\hat\psi_n)\leq x_1|\hat P_n)\cdots P(\sqrt n(\psi_n^{*B}-\hat\psi_n)\leq x_B|\hat P_n)\right]-\prod_{b=0}^B\Phi_\sigma(x_b)\right|{}\notag\\
&&\text{\ \ \ \ \ \ \  \ \ \ \ \ \ since $\sqrt n(\psi_n^{*b}-\hat\psi_n)$ for $b=1,\ldots,B$ are independent conditional on the data $X_1,\ldots,X_n$}\notag\\
&=&\Bigg|E\left[I(\sqrt n(\hat\psi_n-\psi)\leq x_0)P(\sqrt n(\psi_n^{*1}-\hat\psi_n)\leq x_1|\hat P_n)\cdots P(\sqrt n(\psi_n^{*B}-\hat\psi_n)\leq x_B|\hat P_n)\right]{}\notag\\
&&{}-P(\sqrt n(\hat\psi_n-\psi)\leq x_0)\prod_{b=1}^B\Phi_\sigma(x_b)+P(\sqrt n(\hat\psi_n-\psi)\leq x_0)\prod_{b=1}^B\Phi_\sigma(x_b)-\prod_{b=0}^B\Phi_\sigma(x_b)\Bigg|\notag\\
&\leq&E\Bigg[I(\sqrt n(\hat\psi_n-\psi)\leq x_0)\Bigg|P(\sqrt n(\psi_n^{*1}-\hat\psi_n)\leq x_1|\hat P_n)\cdots P(\sqrt n(\psi_n^{*B}-\hat\psi_n)\leq x_B|\hat P_n){}\notag\\
&&{}-\prod_{b=1}^B\Phi_\sigma(x_b)\Bigg|\Bigg]+\left|\left(P(\sqrt n(\hat\psi_n-\psi)\leq x_0)-\Phi_\sigma(x_0)\right)\prod_{b=1}^B\Phi_\sigma(x_b)\right|{}\notag\\
&&{}\text{\ \ \ \ \ \ \  \ \ \ \ \ \ by the triangle inequality}\notag\\
&\leq&E\left|P(\sqrt n(\psi_n^{*1}-\hat\psi_n)\leq x_1|\hat P_n)\cdots P(\sqrt n(\psi_n^{*B}-\hat\psi_n)\leq x_B|\hat P_n)-\prod_{b=1}^B\Phi_\sigma(x_b)\right|{}\notag\\
&&{}+\left|P(\sqrt n(\hat\psi_n-\psi)\leq x_0)-\Phi_\sigma(x_0)\right|\prod_{b=1}^B\Phi_\sigma(x_b)\label{new detail}
\end{eqnarray}
Now, by Assumption \ref{bp}, we have $P(\sqrt n(\psi_n^{*b}-\hat\psi_n)\leq x_b|\hat P_n)\stackrel{p}{\to}\Phi_\sigma(x_b)$ for $b=1,\ldots,B$, and thus $\prod_{b=1}^BP(\sqrt n(\psi_n^{*b}-\hat\psi_n)\leq x_b|\hat P_n)\stackrel{p}{\to}\prod_{b=1}^B\Phi_\sigma(x_b)$. Hence, since $P(\sqrt n(\psi_n^{*b}-\hat\psi_n)\leq x_b|\hat P_n)$ and $\Phi_\sigma(x_b)$ are bounded by 1, the first term in \eqref{new detail} converges to 0 by the bounded convergence theorem. The second term in \eqref{new detail} also goes to 0 because $P(\sqrt n(\hat\psi_n-\psi)\leq x_0)\to\Phi_\sigma(x_0)$ by Assumption \ref{bp} again. Therefore \eqref{new detail} converges to 0 as $n\to\infty$.
\end{proof}

Given Proposition \ref{joint prop}, to infer $\psi$ we can now leverage classical normality inference tools to ``cancel out'' the nuisance parameter $\sigma$. In particular, we can use the pivotal statistic
\begin{equation}
T:=\frac{\hat\psi_n-\psi}{S}\label{pivotal basic}
\end{equation}
where $S^2$ is as defined in \eqref{var}, which converges to a student $t$-distribution with degree of freedom $B$. More concretely:


\begin{proof}[Proof of Theorem \ref{main}]
The pivotal statistic \eqref{pivotal basic} satisfies
\begin{equation}
T=\frac{\hat\psi_n-\psi}{S}=\frac{\sqrt n(\hat\psi_n-\psi)}{\sqrt{\frac{1}{B}\sum_{b=1}^B(\sqrt n(\psi_n^{*b}-\hat\psi_n))^2}}\Rightarrow\frac{Z_0}{\sqrt{\frac{1}{B}\sum_{b=1}^BZ_{b}^2}}\label{proof weak convergence}
\end{equation}
for i.i.d. $N(0,1)$ variables $Z_0,Z_{1},\ldots,Z_{B}$, where we use Proposition \ref{joint prop} and the continuous mapping theorem to deduce the weak convergence. Note that
$$\frac{Z_0}{\sqrt{\frac{1}{B}\sum_{b=1}^BZ_{b}^2}}\stackrel{d}{=}\frac{N(0,1)}{\sqrt{\chi^2_B/B}}\stackrel{d}{=}t_B$$
where $N(0,1)$ and $\chi^2_B$ here represent standard normal and $\chi^2_B$ random variables. Here the two equalities in distribution (denoted ``$\stackrel{d}{=}$") follow from the elementary constructions of $\chi^2$- and $t$-distributions respectively. Thus  $$T=\frac{\hat\psi_n-\psi}{S}\Rightarrow t_B$$
Hence we have
$$\mathbb P_n\left(-t_{B,1-\alpha/2}\leq\frac{\hat\psi_n-\psi}{S}\leq t_{B,1-\alpha/2}\right)\to1-\alpha$$
from which we conclude
$$\mathbb P_n\left(\hat\psi_n-t_{B,1-\alpha/2}S\leq\psi\leq\hat\psi_n+t_{B,1-\alpha/2}S\right)\to1-\alpha$$
and the theorem.
\end{proof}



Note that instead of using the $t$-statistic approach, it is also possible to produce intervals from the normal variables in other ways (e.g., \cite{wall2001effective}) which could have potential benefits for very small $B$. However, the $t$-interval form of \eqref{CI} is intuitive, matches the standard error bootstrap as $B$ grows, and its width is easy to quantify.

In Theorem \ref{main} we have used Assumption \ref{bp}. This assumption is ensured when the functional $\psi(\cdot)$ is Hadamard differentiable with a non-degenerate derivative. More precisely, for a class of functions $\mathcal F$ from $\mathcal X\to\mathbb R$, define $\ell^\infty(\mathcal F):=\left\{z:\|z\|_{\mathcal F}:=\sup_{f\in\mathcal F}|z(f)|<\infty\right\}$
where $z$ is a map from $\mathcal F$ to $\mathbb R$. We have the following:
\begin{proposition}[Sufficient conditions for bootstrap validity]
Consider $\hat P_n$ and $P_n^*$ as random elements that take values in $\ell^\infty(\mathcal F)$, where $\mathcal F$ is a Donsker class with finite envelope. Suppose $\psi:\ell^\infty(\mathcal F)\to\mathbb R$ is Hadamard differentiable at $P$ (tangential to $\ell^\infty(\mathcal F)$) where the derivative $\psi_P'$ satisfies that $\psi_P'(\mathbb G_P)$ is a non-degenerate random variable (i.e., with positive variance), for a tight Gaussian process $\mathbb G_P$ on $\ell^\infty(\mathcal F)$ with mean 0 and covariance $Cov(\mathbb G_P(f_1),\mathbb G_P(f_2))=Cov_P(f_1(X),f_2(X))$ (where $Cov_P$ denotes the covariance taken with respect to $P$). Then Assumption \ref{bp} holds under i.i.d. data.\label{HD}
\end{proposition}

Proposition \ref{HD} follows immediately from the functional delta method (see \cite{van2000asymptotic} \S23 or Appendix \ref{sec:sufficient proof}). Note that we may as well use the conditions in Proposition \ref{HD} as the assumption for Theorem \ref{main}, but Assumption \ref{bp} helps highlight our point that Cheap Bootstrap is valid whenever conventional bootstraps are.

\subsection{Higher-Order Coverage Errors}\label{sec:higher}
We analyze the higher-order coverage errors of Cheap Bootstrap. A common approach to analyze coverage errors in conventional bootstraps is to use Edgeworth expansions, which we will also utilize. However, unlike these existing methods, the pivotal statistic used in Cheap Bootstrap has a limiting $t$-distribution, not a normal distribution. Edgeworth expansions on limiting $t$-distributions appear open to our best knowledge (except a recent working paper \cite{he2021higher}). Here, we derive our expansions for Cheap Bootstrap by integrating the expansions of the original estimate and resample estimates that follow (conditional) asymptotic normal distributions. The resulting coefficients can be explicitly identified which, even though cannot be evaluated in closed-form, are amenable to Monte Carlo integration.

As is customary in the bootstrap literature, we consider the function-of-mean model, namely $\psi=g(\bm\mu)$, where $\bm\mu=E\mathbf X$ for a $d$-dimensional random vector $\mathbf X$ and $g:\mathbb R^d\to\mathbb R$ is a function. Denote $\overline{\mathbf X}=(1/n)\sum_{i=1}^n\mathbf X_i$ as the sample mean of i.i.d. data $\{\mathbf X_1,\ldots,\mathbf X_n\}$. To facilitate our discussion, define
\begin{equation}
A(\mathbf x)=\frac{g(\mathbf x)-g(\bm\mu)}{h(\bm\mu)}\label{A def}
\end{equation}
for function $h:\mathbb R^d\to\mathbb R$, where $h(\bm\mu)^2$ is the asymptotic variance of $\sqrt ng(\overline{\mathbf X})$ that can be written in terms of $\bm\mu$ by the delta method (under regularity conditions that will be listed explicitly in our following theorem). To be more precise on the latter point, note that given any $\tilde{\bm\mu}=E\tilde{\mathbf X}$, we can augment it to $\bm\mu=E\left[\begin{array}{c}\tilde{\mathbf X}\\\tilde{\mathbf X}^2\end{array}\right]$, with the operation $\cdot^2$ defined component-wise (viewing $\tilde{\mathbf X}$ as column vector). Thus, for a given $\tilde g(\tilde{\bm\mu})$, we can define $g(\bm\mu)=\tilde g(\tilde{\bm\mu})$ and $h(\bm\mu)^2=\nabla{\tilde g}(\tilde{\bm\mu})^\top\Sigma\nabla{\tilde g}(\tilde{\bm\mu})$, where $\Sigma$ is the covariance matrix $Cov(\tilde{\mathbf X})$ which is a function of $\bm\mu$, and $\top$ denotes transpose.

We also define the ``studentized" version of $A$ given by
\begin{equation}
A_s(\mathbf x)=\frac{g(\mathbf x)-g(\bm\mu)}{h(\mathbf x)}\label{A def2}
\end{equation}

We have the following:
\begin{theorem}[Higher-order coverage errors for Cheap Bootstrap]
Consider the function-of-mean model where $\psi=g(\bm\mu)$ for some function $g:\mathbb R^d\to\mathbb R$ and $\bm\mu=E\mathbf X$ for a $d$-dimensional random vector $\mathbf X$. Consider also the function $h:\mathbb R^d\to\mathbb R$ that appears in \eqref{A def}. Assume that $g$ and $h$ each has $\nu+3$ bounded derivatives in a neighborhood of $\bm\mu$, that $E\|\mathbf X\|^l<\infty$ for a sufficiently large positive number $l$, and that the characteristic function $\chi$ of $\mathbf X$ satisfies Cramer's condition $\limsup_{\|\mathbf t\|\to\infty}|\chi(\mathbf t)|<1$. Then
\begin{enumerate}
\item When $\nu\geq2$, the two-sided Cheap Bootstrap confidence interval $\mathcal I$ satisfies
$$\mathbb P_n(g(\bm\mu)\in\mathcal I)=(1-\alpha)+\frac{\zeta}{n}+o\left(\frac{1}{n}\right)$$
where the coefficient
\begin{eqnarray*}
\zeta&=&B\int\cdots\int_{\left|\frac{z_0}{\sqrt{\frac{1}{B}\sum_{b=1}^Bz_b^2}}\right|\leq t_{B,1-\alpha/2}}d(p_2(z_B)\phi(z_B))d\Phi(z_{B-1})\cdots d\Phi(z_1)d\Phi(z_0){}\\
&&{}+\int\cdots\int_{\left|\frac{z_0}{\sqrt{\frac{1}{B}\sum_{b=1}^Bz_b^2}}\right|\leq t_{B,1-\alpha/2}}d\Phi(z_B)d\Phi(z_{B-1})\cdots d\Phi(z_1)d(q_2(z_0)\phi(z_0))
\end{eqnarray*}
\item When $\nu\geq1$, the one-sided upper Cheap Bootstrap confidence interval $\mathcal I_{upper}$ satisfies
$$\mathbb P_n(g(\bm\mu)\in\mathcal I_{upper})=(1-\alpha)+\frac{\zeta_{upper}}{\sqrt n}+o\left(\frac{1}{\sqrt n}\right)$$
where the coefficient
\begin{eqnarray*}
\zeta_{upper}&=&B\int\cdots\int_{\frac{z_0}{\sqrt{\frac{1}{B}\sum_{b=1}^Bz_b^2}}\leq t_{B,1-\alpha}}d(p_1(z_B)\phi(z_B))d\Phi(z_{B-1})\cdots d\Phi(z_1)d\Phi(z_0){}\\
&&{}+\int\cdots\int_{\frac{z_0}{\sqrt{\frac{1}{B}\sum_{b=1}^Bz_b^2}}\leq t_{B,1-\alpha}}d\Phi(z_B)d\Phi(z_{B-1})\cdots d\Phi(z_1)d(q_1(z_0)\phi(z_0))
\end{eqnarray*}
and the one-sided lower Cheap Bootstrap confidence interval $\mathcal I_{lower}$ satisfies
$$\mathbb P_n(g(\bm\mu)\in\mathcal I_{lower})=(1-\alpha)+\frac{\zeta_{lower}}{\sqrt n}+o\left(\frac{1}{\sqrt n}\right)$$
where the coefficient
\begin{eqnarray*}
\zeta_{lower}&=&B\int\cdots\int_{\frac{z_0}{\sqrt{\frac{1}{B}\sum_{b=1}^Bz_b^2}}\geq t_{B,1-\alpha}}d(p_1(z_B)\phi(z_B))d\Phi(z_{B-1})\cdots d\Phi(z_1)d\Phi(z_0){}\\
&&{}+\int\cdots\int_{\frac{z_0}{\sqrt{\frac{1}{B}\sum_{b=1}^Bz_b^2}}\geq t_{B,1-\alpha}}d\Phi(z_B)d\Phi(z_{B-1})\cdots d\Phi(z_1)d(q_1(z_0)\phi(z_0))
\end{eqnarray*}
In the above, $\Phi$ and $\phi$ are the standard normal distribution and density functions respectively, $p_j$ and $q_j$ are polynomials of degree $3j-1$, odd for even $j$ and even for odd $j$, with coefficients depending on moments of $\mathbf X$ up to order $j+2$ polynomially and also $g(\cdot)$ and $h(\cdot)$.
\end{enumerate}\label{main Edgeworth}
\end{theorem}

In Theorem \ref{main Edgeworth}, the polynomials $p_j$ and $q_j$ are related to $A$ and $A_s$ defined in \eqref{A def} and \eqref{A def2} as follows. Under the assumptions in the theorem, the $j$-th cumulant of $\sqrt nA(\overline{\mathbf X})$ and $\sqrt nA_s(\overline{\mathbf X})$ can be expanded as
$$\kappa_{j,n}=n^{-(j-2)/2}\left(k_{j,1}+\frac{k_{j,2}}{n}+\frac{k_{j,3}}{n^2}+\cdots\right)$$
for coefficients $k_{j,l}$'s depending on whether we are considering $A$ or $A_s$.
Then $p_1$ or $q_1$ is equal to
\begin{equation*}
-\left(k_{1,2}+\frac{1}{6}k_{3,1}H_2(x)\right)=-\left(k_{1,2}+\frac{1}{6}k_{3,1}(x^2-1)\right)
\end{equation*}
while $p_2$ or $q_2$ is equal to
\begin{eqnarray*}
&&-\left(\frac{1}{2}(k_{2,2}+k_{1,2}^2)H_1(x)+\frac{1}{24}(k_{4,1}+4k_{1,2}k_{3,1})H_3(x)+\frac{1}{72}k_{3,1}^2H_5(x)\right)\\
&=&-x\left(\frac{1}{2}(k_{2,2}+k_{1,2}^2)+\frac{1}{24}(k_{4,1}+4k_{1,2}k_{3,1})(x^2-3)+\frac{1}{72}k_{3,1}^2(x^4-10x^2+15)\right)
\end{eqnarray*}
Here $H_j(\cdot)$ is the $j$-th order Hermite polynomial, and $k_{j,l}$'s are determined from $A$ for $p_1,p_2$ and determined from $A_s$ for $q_1,q_2$.

The coverage error $O(n^{-1})$ of the Cheap Bootstrap in the two-sided case and $O(n^{-1/2})$ in the one-sided case in Theorem \ref{main Edgeworth} match the conventional basic and percentile bootstraps described in Section \ref{sec:comparisons}. Nonetheless, these errors are inferior to more refined approaches, including the studentized bootstrap also mentioned earlier which attains $O(n^{-1})$ in the one-sided case (\cite{davison1997bootstrap}).

Theorem \ref{main Edgeworth} is proved by expressing the distribution of the pivotal statistic in \eqref{pivotal basic} as a multi-dimensional integral, with respect to measures that are approximated by the Edgeworth expansions of the original estimate $\sqrt n(\hat\psi_n-\psi)$ and the conditional expansions of the resample estimates $\sqrt n(\psi_n^{*b}-\hat\psi)$ which have limiting normal distributions. From these expansions, we could also identify the polynomials $p_j$ and $q_j$ in our discussion above by using equations (2.20), (2.24) and (2.25) in \cite{hall2013bootstrap}. Lastly, we note that the remainder term in the coverage of the two-sided confidence interval in Theorem \ref{main Edgeworth} can be refined to $O(1/n^{3/2})$ when $\nu\geq3$, and the one-sided intervals can be refined to $O(1/n)$ when $\nu\geq2$. These can be seen by tracing our proof (in Appendix \ref{sec:main proof}).

The coefficients $\zeta$, $\zeta_{upper}$ and $\zeta_{lower}$ are computable via Monte Carlo integration, because integral in the form
$$\int\cdots\int_{(z_0,\cdots,z_B)\in\mathcal S}d(\pi(z_B)\phi(z_B))d\Phi(z_{B-1})\cdots d\Phi(z_1)d\Phi(z_0)$$
for some set $\mathcal S$ and polynomial $\pi$ can be written as
\begin{eqnarray*}
&&\int\cdots\int_{(z_0,\cdots,z_B)\in\mathcal S}(\pi'(z_B)\phi(z_B)-z_B\pi(z_B)\phi(z_B)) dz_Bd\Phi(z_{B-1})\cdots d\Phi(z_1)d\Phi(z_0)\\
&=&\int\cdots\int_{(z_0,\cdots,z_B)\in\mathcal S}(\pi'(z_B)-z_B\pi(z_B))d\Phi(z_B)d\Phi(z_{B-1})\cdots d\Phi(z_1)d\Phi(z_0)
\end{eqnarray*}
which is expressible as an expectation
$$E[\pi'(Z_B)-Z_B\pi(Z_B);(Z_0,\cdots,Z_B)\in\mathcal S]$$
taken with respect to independent standard normal variables $Z_0,\ldots,Z_B$.

\subsection{Interval Tightness}\label{sec:tightness}
Besides coverage, another important efficiency criterion is the interval width. From Section \ref{sec:exact}, we know that the Cheap Bootstrap interval \eqref{CI} arises from a $t$-interval construction, using which we can readily extract its width behavior. More specifically, for a fixed number of resamples $B$, $S^2$ satisfies
$$\sqrt nS\Rightarrow\sqrt{\frac{1}{B}\sum_{b=1}^B Z_b^2}\stackrel{d}{=}\sigma\sqrt{\frac{\chi^2_B}{B}}$$
as $n\to\infty$ so that the half-width of $\mathcal I$ is $\text{HW}\approx t_{B,1-\alpha/2}\sigma\sqrt{\chi_B^2/(nB)}$. Plugging in the moments of $\chi_B^2$, we see that the half-width for large sample size $n$ has mean and variance given by
\begin{align}
E[\text{HW}]&\approx t_{B,1-\alpha/2}\sqrt{\frac{2}{B}}\frac{\Gamma((B+1)/2)}{\Gamma(B/2)}\frac{\sigma}{\sqrt n}\label{mean HW}\\
Var(\text{HW})&\approx t_{B,1-\alpha/2}^2\left(B-\frac{2\Gamma((B+1)/2)^2}{\Gamma(B/2)^2}\right)\frac{\sigma^2}{nB}\label{var HW}
\end{align}
respectively, where $\Gamma(\cdot)$ is the Gamma function. As $B$ increases, both the mean and variance decrease, which signifies a natural gain in statistical efficiency, until in the limit $B=\infty$ we get a mean  $z_{1-\alpha/2}\sigma/\sqrt n$ and a variance 0, which correspond to the normality interval with a known $\sigma$.





The expressions \eqref{mean HW} and \eqref{var HW} reveal that the half-width of Cheap Bootstrap is large when $B=1$, but falls and stabilizes quickly as $B$ increases. Table \ref{table:HW} shows the approximate half-width mean and standard deviation shown in \eqref{mean HW} and \eqref{var HW} at $\alpha=5\%$ (ignoring the $\sigma/\sqrt n$ factor), and the relative inflation in mean half-width compared to the case $B=\infty$ (i.e., $z_{1-\alpha/2}=1.96)$. We see that, as $B$ increases from 1 to 2, the half-width mean drops drastically from $10.14$ ($417.3\%$ inflation relative to $B=\infty$) to $3.81$ ($94.6\%$ inflation), and half-width standard deviation from $7.66$ to $1.99$. As $B$ increases from 2 to 3, the mean continues to drop notably to $2.93$ ($49.6\%$ inflation) and standard deviation to $1.24$. The drop rate slows down as $B$ increases further. For instance, at $B=10$ the mean is $2.17$ ($10.9\%$ inflation) and standard deviation is $0.49$, while at $B=20$ the mean is $2.06$ ($5.1\%$ inflation) and standard deviation is $0.33$. Though what constitutes an acceptable inflation level compared to $B=\infty$ is context-dependent, generally the inflation appears reasonably low even when $B$ is a small number, except perhaps when $B$ is $1$ or $2$.

\begin{table}[ht]
\caption{Approximate half-width performance of Cheap Bootstrap against $B$ at $95\%$ confidence level.}
\centering
\begin{tabular}{cccc}
 Replication&Mean (ignoring&Mean inflation relative&Standard deviation (ignoring\\
 size $B$&the $\sigma/\sqrt n$ factor)&to $B=\infty$ ($z_{0.975}$)&the $\sigma/\sqrt n$ factor)\\
   \hline
  1&10.14&417.3\%&7.66\\
  2&3.81&94.6\%&1.99\\
  3&2.93&49.6\%&1.24\\
  4&2.61&33.2\%&0.95\\
  5&2.45&24.8\%&0.79\\
  6&2.35&19.8\%&0.69\\
  7&2.28&16.4\%&0.62\\
  8&2.24&14.0\%&0.57\\
  9&2.20&12.3\%&0.53\\
  10&2.17&10.9\%&0.49\\
  11&2.15&9.8\%&0.46\\
  12&2.13&8.9\%&0.44\\
  13&2.12&8.1\%&0.42\\
  14&2.11&7.5\%&0.40\\
  15&2.10&7.0\%&0.39\\
  16&2.09&6.5\%&0.37\\
  17&2.08&6.1\%&0.36\\
  18&2.07&5.7\%&0.35\\
  19&2.07&5.4\%&0.34\\
  20&2.06&5.1\%&0.33\\

\end{tabular}
\label{table:HW}
\end{table}

In Appendices \ref{sec:SE} and \ref{sec:multi}, we discuss inference on standard error and a multivariate version of the Cheap Bootstrap, which are immediate extensions of our developments in this section. In the next two sections, we discuss two generalizations of our approach to models corrupted by both data and computation noises (Section \ref{sec:double}) and to subsampling (Section \ref{sec:subsampling}).

\section{Applying Cheap Bootstrap to Nested Sampling Problems}\label{sec:double}
Our Cheap Bootstrap can reduce computational cost for problems which, when applying the conventional bootstraps, require nested sampling. This phenomenon arises when the estimate involves noises coming from both the data and computation procedure. To facilitate discussion, as in Section \ref{sec:basic}, denote $\psi:=\psi(P)$ as a target parameter or quantity of interest. However, here the computation of $\psi(Q)$, for any given $Q$, could be noisy. More precisely, suppose that, given any $Q$, we could only generate $\hat\psi_r(Q)$ where $\hat\psi_r(Q)$ is an unbiased output for $\psi(Q)$. Then, to compute a point estimate of $\psi$, we use the data $X_1,\ldots,X_n$ to construct the empirical distribution $\hat P_n$, and with this $\hat P_n$, we output
\begin{equation}
\hat{\hat\psi}_{n,R}=\frac{1}{R}\sum_{r=1}^R\hat\psi_r(\hat P_n)\label{pt estimate nested}
\end{equation}
where $\hat\psi_r(\hat P_n)$, $r=1,\ldots,R$ is a sequence of unbiased runs for $\psi(\hat P_n)$, independent given $\hat P_n$. The ``double hat" above $\psi$ in the left hand side of \eqref{pt estimate nested} signifies the two sources of noises, one from the estimation of $P$ and one from the estimation of $\psi(\cdot)$, in the resulting overall point estimate.

The above construction that requires generating and averaging multiple noisy outputs of $\hat\psi_r$ arises in the following examples:
\begin{enumerate}
\item\emph{Input uncertainty quantification in simulation modeling:} Here $\psi(P)$ denotes the output performance measure of a stochastic simulation model, where $P$ is the distribution of input random variates fed into the simulation logic. For instance, $\psi(P)$ could be the expected workload of a queueing network, and $P$ denotes the inter-arrival time distribution. Estimating the performance measure $\psi(P)$, even with a known input distribution $P$, would require running the simulation model many times (i.e., $R$ times) and taking their average. When $P$ is observed from exogenous data, we use a plug-in estimate $\hat P_n$ to drive the input variates in the simulation model (\cite{henderson2003input,song2014advanced,barton2012tutorial,lam2016advanced}). This amounts to the point estimate \eqref{pt estimate nested}.
\item\emph{Bagging:} Bagging predictors are constructed by averaging a large number of base predictors, where each base predictor is obtained by re-training the model on a new resample of the data (\cite{breiman1996bagging,buhlmann2002analyzing}). Here, $\psi(P)$ denotes the target predicted value (at some tested point) of the bagging predictor, and $P$ is the data distribution. The quantity $\hat\psi_r(\hat P_n)$ denotes the predicted value of a model trained on a resample from $\hat P_n$.
\item\emph{Deep ensemble:} Deep ensembles are predictors constructed by averaging several neural networks, each trained from the same data but using a different randomization seed (\cite{lakshminarayanan2016simple,lee2015m}). The seed controls, for instance, the initialization of a stochastic gradient descent. Like bagging, $\psi(P)$ here denotes the target predicted value (at some tested point) where $P$ is the data distribution. However, the quantity $\hat\psi_r(\hat P_n)$ refers to the output of a trained individual neural network. 
\end{enumerate}

Applying the bootstrap to assess the uncertainty of \eqref{pt estimate nested} requires a nested sampling that involves two layers: First, we resample the data $B$ times to obtain $P_n^{*b},b=1,\ldots,B$. Then, given each resample $P_n^{*b}$, we run our computation procedure $R$ times to obtain $\hat\psi_r(P_n^{*b}),r=1,\ldots,R$. The quantities
\begin{equation}
\psi_{n,R}^{**b}=\frac{1}{R}\sum_{r=1}^R\hat\psi_r(P_n^{*b}),\ \ b=1,\ldots,B\label{re estimate nested}
\end{equation}
then serve as the resample estimates whose statistic (e.g., quantiles or standard deviation) can be used to obtain a confidence interval for the target quantity $\psi$. The computation effort in generating these intervals is the product between the outer number of resamples and the inner number of computation runs, namely $BR$. The number $B$ needs to be sufficiently large as required by conventional bootstraps. The number $R$, depending on the context, may also need to be large. For instance, the simulation modeling and bagging examples above require a large enough $R$. Furthermore, not only does $R$ need to be large, but it could also depend (linearly) on the data size $n$ in order to achieve consistency, as shown in the context of simulation modeling (\cite{lam2021subsampling,lam2018subsampling}). In other words, the difficulty with nested sampling when running the bootstrap not only lies in the multiplicative computation effort between $B$ and $R$, but also that their choices can depend intricately on the data size which adds to the complication. 

Cheap Bootstrap reduces the outer number of resamples $B$ to a low number, a strength that we investigate in this section. To build the theoretical framework, we first derive an asymptotic on problems involving noises from both data and computation.
Denote $\tau^2(Q)=Var(\hat\psi_r(Q))$ and $\kappa_3(Q)=E|\hat\psi_r(Q)-\psi(Q)|^3$ as the variance and third-order absolute central moment of a computation run driven by a given input distribution $Q$. We assume the following.
\begin{assumption}[Moment consistency of computation runs]
We have $\tau^2(\hat P_n)\stackrel{p}{\to}\tau^2(P)=:\tau^2$ and $\kappa_3(\hat P_n)\stackrel{p}{\to}\kappa_3(P)=:\kappa_3$ as $n\to\infty$, where $0<\tau^2<\infty$ and $\kappa_3<\infty$. Similarly, $\tau^2(P_n^*)\stackrel{p}{\to}\tau^2(P)=\tau^2$ and $\kappa_3(P_n^*)\stackrel{p}{\to}\kappa_3(P)=\kappa_3$. \label{assumption:sim}
\end{assumption}

Assumption \ref{assumption:sim} can be justified if $\tau^2(\cdot)$ and $\kappa_3(\cdot)$ are Hadamard differentiable with non-degenerate derivatives like in Assumption \ref{bp}. Specifically, we have the following:
\begin{proposition}
If $\tau^2(\cdot)$ and $\kappa_3(\cdot)$ satisfy the assumptions on $\psi(\cdot)$ in Proposition \ref{HD}, then Assumption \ref{assumption:sim} holds.\label{technical moment}
\end{proposition}

In stating our asymptotic result below, we consider a slightly more general version of the bootstrap where the computation sizes in the original estimate and the resample estimates are allowed to be different, which gives some extra flexibility to our procedures. That is, we use $\hat{\hat\psi}_{n,R_0}$ as defined in \eqref{pt estimate nested} where $R_0$ denotes the computation size in constructing the point estimate. Then we resample $B$ times and obtain $\psi_{n,R}^{**b}$ as defined in \eqref{re estimate nested} where $R$ denotes the computation size in constructing each resample estimate. Moreover, we denote $\hat\psi_n=\psi(\hat P_n)$ as the hypothetical original estimate constructed from $\hat P_n$ without any computation noise, and $\psi_n^{*b}=\psi(P_n^{*b})$ as the hypothetical $b$-th resample estimate without computation noise.
We have the following:


\begin{theorem}[Joint asymptotic of original and resample estimates under both data and computation noises]
Suppose Assumptions \ref{bp} and \ref{assumption:sim} hold and $n,R_0,R\to\infty$ such that $R_0/n\to p_0$ and $R/n\to p$ for some constants $0<p_0,p<1$. We have
\begin{eqnarray}
&&\sqrt n\left(\hat{\hat\psi}_{n,R_0}-\psi,\ \psi_{n,R}^{**1}-\hat{\hat\psi}_{n,R_0},\ldots,\ \psi_{n,R}^{**B}-\hat{\hat\psi}_{n,R_0}\right)\notag\\
&\Rightarrow&\left(\sigma Z_0+\frac{\tau}{\sqrt{p_0}}W_0,\ \sigma Z_1+\frac{\tau}{\sqrt p}W_1-\frac{\tau}{\sqrt{p_0}}W_0,\ldots,\ \sigma Z_B+\frac{\tau}{\sqrt p}W_B-\frac{\tau}{\sqrt{p_0}}W_0\right)\label{CLT nested}
\end{eqnarray}
where $Z_0,Z_1,\ldots,Z_B,W_0,W_1,\ldots,W_B\sim N(0,1)$ and are all independent.\label{joint1}
\end{theorem}
Note that the scaling $R_0/n\to p_0$ and $R/n\to p$ in Theorem \ref{joint1} are introduced only for mathematical convenience. The convergence \eqref{CLT nested} stipulates that
\begin{eqnarray*}
&&\left(\hat{\hat\psi}_{n,R_0}-\psi,\ \psi_{n,R}^{**1}-\hat{\hat\psi}_{n,R_0},\ldots,\ \psi_{n,R}^{**B}-\hat{\hat\psi}_{n,R_0}\right)\\
&\stackrel{d}{\approx}&\left(\frac{\sigma}{\sqrt n}Z_0+\frac{\tau}{\sqrt{R_0}}W_0,\ \frac{\sigma}{\sqrt n}Z_1+\frac{\tau}{\sqrt R}W_1-\frac{\tau}{\sqrt{R_0}}W_0,\ldots,\ \frac{\sigma}{\sqrt n}Z_B+\frac{\tau}{\sqrt R}W_B-\frac{\tau}{\sqrt{R_0}}W_0\right)
\end{eqnarray*}
It implies, in particular, that the standard error of $\hat{\hat\psi}_{n,R_0}$ is given by $\sqrt{\sigma^2/n+\tau^2/R_0}$ which consists of two parts: $\sigma^2/n$ captures the variability from the data and $\tau^2/R_0$ from the computation.

Theorem \ref{joint1} is a generalization of Proposition \ref{joint prop}. The random variables $Z_b$'s and $W_b$'s in the limit each signifies a source of noise coming from either data or computation. More specifically, $Z_0$ comes from the original data, $Z_b,b=1,\ldots,B$ from each resample, $W_0$ from the computation of the original estimate, and $W_b,b=1,\ldots,B$ from the computation of each resample. Unlike \eqref{joint}, in \eqref{CLT nested} the (centered) resample estimate $\psi_{n,R}^{**b}-\hat{\hat\psi}_{n,R_0}$ and original estimate $\hat{\hat\psi}_{n,R_0}-\psi$ are no longer asymptotically independent, because they share the same source of noise $W_0$. This causes modifications to our Cheap Bootstrap in canceling out the nuisance parameter.



We will describe two approaches to construct Cheap Bootstrap confidence intervals for $\psi$ under the conditions in Theorem \ref{joint1}, both resulting in the form
\begin{equation}
\left[\hat{\hat\psi}_{n,R_0}-q_{\cdot,1-\alpha/2} S_\cdot,\ \hat{\hat\psi}_{n,R_0}+q_{\cdot,1-\alpha/2}S_\cdot\right]\label{CI general nested}
\end{equation}
where $S_\cdot$, like the $S$ in \eqref{CI}, is a standard error estimate of $\hat{\hat\psi}_{n,R_0}$. These two approaches differ by how we pool together the resample estimates in $S_\cdot$, where the $\cdot$ therein denotes the approach. Correspondingly, the critical value $q_{\cdot,1-\alpha/2}$ is obtained from the $(1-\alpha/2)$-quantile of a distribution pertinent to the pivotal statistic using $S_\cdot$. Both approaches we propose achieve asymptotic validity with very few resamples. 


\subsection{Centered at Original Estimate}\label{sec:cen}

In the first approach, we use in \eqref{CI general nested}
\begin{equation}
S_{O}^2=\frac{1}{B}\sum_{b=1}^B(\psi_{n,R}^{**b}-\hat{\hat\psi}_{n,R_0})^2\label{S center}
\end{equation}
That is, ${S_{O}}^2$ acts as the sample variance of the bootstrap estimates, where the center of the squares is set to be $\hat{\hat\psi}_{n,R_0}$. We choose $q_{O,1-\alpha/2}$ to be
\begin{equation}
q_{O,1-\alpha/2}=\min\left\{q:\min_{\theta\geq0}F(q;\theta,\rho)\geq1-\frac{\alpha}{2}\right\}\label{q center}
\end{equation}
where $F(q;\theta,\rho)$ is the distribution function of
\begin{equation}
\frac{\theta V_1+V_2}{\sqrt{\frac{\theta^2+\rho^2}{B}\left(Y+V_3^2\right)-2\sqrt{\frac{\theta^2+\rho^2}{B}}V_3V_2+V_2^2}}\label{limiting expression}
\end{equation}
with $V_1,V_2,V_3\stackrel{i.i.d.}{\sim}N(0,1)$ and  $Y\sim\chi^2_{B-1}$ being all independent, and $\rho=\sqrt{p_0/p}$ is the square-rooted ratio between the computation size used in the original estimate and each resample estimate. When $B=1$, we set $Y=0$ in \eqref{limiting expression} so that the expression is equivalent to
\begin{equation}
\frac{\theta V_1+V_2}{|\sqrt{\theta^2+\rho^2}V_3-V_2|}\label{limiting expression single}
\end{equation}
We call the resulting interval
\begin{equation}
\mathcal I_{O}=\left[\hat{\hat\psi}_{n,R_0}-q_{O,1-\alpha/2} S_{O},\ \hat{\hat\psi}_{n,R_0}+q_{O,1-\alpha/2}S_{O}\right]\label{cen CI}
\end{equation}
the Cheap Bootstrap interval ``centered at original estimate".

Note that the distribution function $F(q;\theta,\rho)$ has no closed-form expression or reduction to any common distribution where a table or built-in calculator is available. However, it can be easily simulated, for instance, by running many independent copies of $(V_1,V_2,V_3,Y)$, say $(V_1^j,V_2^j,V_3^j,Y^j)$ for $j=1,\ldots,N$, and computing (via, e.g., a simple grid search)
\begin{equation}
\min\left\{q:\min_{\theta\geq0}\frac{1}{N}\sum_{j=1}^NI\left(\frac{\theta V_1^j+V_2^j}{\sqrt{\frac{\theta^2+\rho^2}{B}\left(Y^j+{V_3^j}^2\right)-2\sqrt{\frac{\theta^2+\rho^2}{B}}V_3^jV_2^j+{V_2^j}^2}}\leq q\right)\geq1-\frac{\alpha}{2}\right\}\label{grid search}
\end{equation}
where $I(\cdot)$ denotes the indicator function, to approximate $q_{O,1-\alpha/2}$. Note that this Monte Carlo computation only involves standard normal and $\chi^2$ random variables, not the noisy model evaluation $\hat\psi_r$ that could be expensive.

We have the following guarantee:
\begin{theorem}[Asymptotic validity of Cheap Bootstrap ``centered at point estimate"]
Under the same assumptions and setting in Theorem \ref{joint1}, the interval $\mathcal I_{O}$ in \eqref{cen CI}, where $S_{O}^2$ and $q_{O,1-\alpha/2}$ are defined in \eqref{S center} and \eqref{q center}, is an asymptotically valid $(1-\alpha)$-level confidence interval, i.e., it satisfies
$$\liminf_{n\to\infty}\mathbb P_n(\psi\in\mathcal I_{O})\geq1-\alpha$$
as $n\to\infty$, where $\mathbb P_n$ denotes the probability with respect to the data $X_1,\ldots,X_n$ and all randomness from the resampling and computation.\label{thm1}
\end{theorem}

Theorem \ref{thm1} is obtained by looking at the asymptotic distribution of the (approximately) pivotal statistic $(\hat{\hat\psi}_{n,R_0}-\psi)/S_O$. However, because of the dependence between the original and resample estimates in \eqref{CLT nested}, the unknown nuisance parameters $\sigma^2$ and $\tau^2$ are not directly canceled out. To tackle this, we consider a worst-case calculation which leads to some conservativeness in the coverage guarantee, i.e., inequality instead of exact equality in Theorem \ref{thm1}.

\subsection{Centered at Resample Mean}\label{sec:nc}
We consider our second approach. For $B\geq2$, we use in \eqref{CI general nested}
\begin{equation}
S_{M}^2=\frac{1}{B-1}\sum_{b=1}^B(\psi_{n,R}^{**b}-\overline{\psi_{n,R}^{**}})^2\label{S nc}
\end{equation}
where $\overline{\psi_{n,R}^{**}}$ is the sample mean of $\psi_{n,R}^{**b}$'s given by $\overline{\psi_{n,R}^{**}}=(1/B)\sum_{b=1}^B\psi^{**b}_{n,R}$. That is, we use the sample variance of $\psi^{**b}_{n,R}$'s as our standard error estimate. Correspondingly, we set
\begin{equation}
q_{M,1-\alpha/2}=\max\{\rho^{-1},1\}t_{B-1,1-\alpha/2}\label{q nc}
\end{equation}
where, like in Section \ref{sec:cen}, $\rho=\sqrt{p_0/p}$ is the square-rooted ratio between the computation sizes of the original estimate and each resample estimate. Also, recall $t_{B-1,1-\alpha/2}$ is the $(1-\alpha/2)$-quantile of $t_{B-1}$. We call the resulting interval
\begin{equation}
\mathcal I_{M}=\left[\hat{\hat\psi}_{n,R_0}-q_{M,1-\alpha/2}S_{M},\ \hat{\hat\psi}_{n,R_0}+q_{M,1-\alpha/2}S_{M}\right]\label{nc CI}
\end{equation}
the Cheap Bootstrap interval ``centered at resample mean".

We have the following guarantee:
\begin{theorem}[Asymptotic validity of Cheap Bootstrap ``centered at resample mean"]
Under the same assumptions and setting in Theorem \ref{joint1}, for $B\geq2$, the interval $\mathcal I_{M}$ in \eqref{nc CI} where $S_{M}^2$ and $q_{M,1-\alpha/2}$ are defined in \eqref{S nc} and \eqref{q nc}, is an asymptotically valid $(1-\alpha)$-level confidence interval, i.e., it satisfies
$$\liminf_{n\to\infty}\mathbb P_n(\psi\in\mathcal I_{M})\geq1-\alpha$$
as $n\to\infty$, where $\mathbb P_n$ denotes the probability with respect to the data $X_1,\ldots,X_n$ and all randomness from the resampling and computation. Moreover, if the computation sizes for each resample estimate and the original estimate are the same, i.e., $\rho=1$, then $\mathcal I_{M}$ is asymptotically exact, i.e.,
$$\lim_{n\to\infty}\mathbb P_n(\psi\in\mathcal I_{M})=1-\alpha$$
\label{thm2}
\end{theorem}

Compared to $\mathcal I_O$ in Section \ref{sec:cen}, $\mathcal I_{M}$ does not require Monte Carlo computation for $q_{M,1-\alpha/2}$ and is thus easier to use. Moreover, note that when the computation size used in each resample estimate is at most that used in the original estimate, i.e., $\rho^{-1}\leq1$, we have $q_{M,1-\alpha/2}=t_{B-1,1-\alpha/2}$, so the critical value in $\mathcal I_{M}$ reduces to a standard $t$-quantile, just like the interval \eqref{CI} for non-nested problems. Furthermore, when these computation sizes are the same, then we have asymptotically exact coverage in Theorem \ref{thm2}. On the other hand, $\mathcal I_{M}$ is defined for $B\geq2$ instead of $B\geq1$, and thus its requirement on $B$ is slightly more stringent than $\mathcal I_{O}$ and \eqref{CI}.


\section{Cheap Subsampling}\label{sec:subsampling}
We integrate Cheap Bootstrap into, in broad term, subsampling methods where the bootstrap resample has a smaller size than the original full data size. Subsampling is motivated by the lack of consistency when applying standard bootstraps in non-smooth problems (\cite{politis1999subsampling,bickel2012resampling}), but also can be used to address the computational challenge in  repeatedly fitting models over large data sets (\cite{10.5555/3042573.3042801}). The latter arises because standard sampling with replacement on the original data retains around $63\%$ of the observation values, thus for many problems each bootstrap resample estimate requires roughly the same computation order as the original estimate constructed from the raw data.


Here we consider three subsampling variants: $m$-out-of-$n$ Bootstrap (\cite{bickel2012resampling}), and the more recent Bag of Little Bootstraps (\cite{kleiner2014scalable}) and Subsampled Double Bootstrap (\cite{sengupta2016subsampled}). The validity of these methods relies on generalizations of Assumption \ref{bp} (though not always, as explained at the end of this section), in which the distribution of $\sqrt n(\hat\psi_n-\psi)$ can be approximated by that of $\sqrt{N}(\Psi^*-\Psi)$, where $\Psi^*$ and/or $\Psi$ are some estimators obtained using a small resample data size or number of distinct observed values, and $N$ is a suitable scaling parameter. More precisely, suppose we have $\sqrt{N}(\Psi^*-\Psi)\Rightarrow N(0,\sigma^2)$ in probability conditional on the data $X_1,\ldots,X_n$ and possibly independent randomness (the independent randomness is sometimes used to help determine $\Psi$). Then, to construct a $(1-\alpha)$-level confidence interval for $\psi$, we can simulate the $\alpha/2$ and $(1-\alpha/2)$-quantiles of $\Psi^*-\Psi$, say $ Q_{\alpha/2}$ and $Q_{1-\alpha/2}$, and then output $$\left[\hat\psi_n-\sqrt{\frac{N}{n}}Q_{1-\alpha/2},\ \hat\psi_n+\sqrt{\frac{N}{n}}Q_{\alpha/2}\right]$$
Like in the standard bootstrap, such an approach requires many Monte Carlo replications.

We describe how to use subsampling in conjunction with Cheap Bootstrap to devise bootstrap schemes that are small simultaneously in the number of resamples and the resample data size (or number of distinct observed values in use). Following the ideas in Sections \ref{sec:basic} and \ref{sec:theory}, our Cheap Subsampling unifiedly uses the following framework: Generate $\Psi^{*b}$ and $\Psi^b$ given $X_1,\ldots,X_n$ and any required additional randomness, for $b=1,\ldots,B$. Then the $(1-\alpha)$-level confidence interval is given by
\begin{equation}
\left[\hat\psi_n-t_{B,1-\alpha/2}\sqrt{\frac{N}{n}}S,\ \hat\psi_n+t_{B,1-\alpha/2}\sqrt{\frac{N}{n}}S\right]\label{subsampling main}
\end{equation}
where
$$S^2=\frac{1}{B}\sum_{b=1}^B\left(\Psi^{*b}-\Psi^b\right)^2$$
Compared to \eqref{CI}, the main difference in \eqref{subsampling main} is the adjusting factor $\sqrt{N/n}$ in the standard error.

Below we describe the application of the above framework into three subsampling variants. Like in Section \ref{sec:basic}, suppose the target parameter is $\psi=\psi(P)$ and we have obtained the point estimate $\hat\psi_n=\psi(\hat P_n)$.

\subsection{Cheap $m$-out-of-$n$ Bootstrap}
We compute $B$ subsample estimates $\psi_s^{*b}:=\psi(P_s^{*b})$, $b=1,\ldots,B$, where $P_s^{*b}$ is the empirical distribution constructed from a subsample of size $s<n$ drawn from the raw data $\{X_1,\ldots,X_n\}$ via sampling with replacement. Then output the interval
\begin{equation}
\left[\hat\psi_n-t_{B,1-\alpha/2}\sqrt{\frac{s}{n}}S,\ \hat\psi_n+t_{B,1-\alpha/2}\sqrt{\frac{s}{n}}S\right]\label{mn CI}
\end{equation}
where $S^2=\frac{1}{B}\sum_{b=1}^B\left(\psi_s^{*b}-\hat\psi_n\right)^2$.


Cast in our framework above, here we take $\Psi^b$ to be always the original estimate $\hat\psi_n$, $\Psi^{*b}$ to be a subsample estimate $\psi_s^{*b}$, and the scale parameter $N$ to be the subsample size $s$.





\subsection{Cheap Bag of Little Bootstraps}
We first construct $\psi_s^*:=\psi(P_s^*)$, where $P_s^*$ is the empirical distribution of a subsample of size $s<n$, called $\mathbf X_s^*$, drawn from the raw data $\{X_1,\ldots,X_n\}$ via sampling without replacement (i.e., $\mathbf X_s^*$ is an $s$-subset of $\{X_1,\ldots,X_n\}$). The subsample $\mathbf X_s^*$, once drawn, is  fixed throughout. Given $\mathbf X_s^*$, we generate $\psi^{**b}_n:=\psi(P_n^{**b})$, $b=1,\ldots,B$, where $P_n^{**b}$ is the empirical distribution constructed from sampling with replacement from the subsample $\mathbf X_s^*$ for $n$ times, with $n$ being the original full data size (i.e., $P_n^{**b}$ is a weighted empirical distribution over $\mathbf X_s^*$). We output the interval in \eqref{CI}, where now \begin{equation}
S^2=\frac{1}{B}\sum_{b=1}^B\left(\psi_n^{**b}-\psi_s^*\right)^2\label{BLB S}
\end{equation}

Cast in our framework, here we take $\Psi^b$ to be always $\psi_s^*$, $\Psi^{*b}$ to be the double resample estimate $\psi_n^{**b}$, and the scale parameter $N$ to be the original data size $n$. The subsampling used to obtain $\psi_s^*$ introduces an additional randomness that is conditioned upon in the conditional weak convergence $\sqrt{N}(\Psi^*-\Psi)\Rightarrow N(0,\sigma^2)$ described before. Note that when $s$ is small, the double resample estimate $\psi_n^{**b}$, though constructed with a full size data, uses only a small number of distinct data points which, in problems such as $M$-estimation, involves only a weighted estimation on the small-size subsample (\cite{kleiner2014scalable}).

\subsection{Cheap Subsampled Double Bootstrap}
For $b=1,\ldots,B$, we do the following: First, generate a subsample estimate $\psi_s^{*b}:=\psi(P_s^{*b})$ where $P_s^{*b}$ is the empirical distribution of a subsample of size $s<n$, called $\mathbf X_s^{*b}$, drawn from the raw data $\{X_1,\ldots,X_n\}$ via sampling without replacement (i.e., $\mathbf X_s^{*b}$ is an $s$-subset of $\{X_1,\ldots,X_n\}$). Then, given $\mathbf X_s^{*b}$, we construct $\psi_n^{**b}=\psi(P_n^{**b})$ where $P_n^{**b}$ is the empirical distribution of a size-$n$ resample constructed by sampling with replacement from $\mathbf X_s^{*b}$. We output the interval in \eqref{CI}, where now
\begin{equation}
S^2=\frac{1}{B}\sum_{b=1}^B\left(\psi_n^{**b}-\psi_s^{*b}\right)^2\label{SDB S}
\end{equation}


Cheap Subsampled Double Bootstrap is similar to Cheap Bag of Little Bootstraps, but the $\Psi^b$ in our framework is now taken as a subsample estimate $\psi_s^{*b}$ that is newly generated from a new subsample $\mathbf X_s^{*b}$ for each $b$, and $\Psi^{*b}$ is obtained by resampling with replacement from that particular $\mathbf X_s^{*b}$.

\subsection{Theory of Cheap Subsampling}
All of Cheap $m$-out-of-$n$ Bootstrap, Cheap Bag of Little Bootstraps, and Cheap Subsampled Double Bootstrap achieve asymptotic exactness for any $B\geq1$. To state this concretely, we denote $\mathcal F_\delta=\{f-g:f,g\in\mathcal F,\ \rho_P(f-g)<\delta\}$, where $\rho_P(f-g):=(Var_P(f(X)-g(X)))^{1/2}$ is the canonical metric.
\begin{theorem}[Asymptotic exactness of Cheap Subsampling]
Suppose the assumptions in Proposition \ref{HD} hold, and in addition $\mathcal F_\delta$ is measurable for every $\delta>0$. Then the intervals produced by Cheap $m$-out-of-$n$ Bootstrap (i.e., \eqref{mn CI}), Cheap Bag of Little Bootstraps (i.e., \eqref{CI} using \eqref{BLB S}), and Cheap Subsampled Double Bootstrap (i.e., \eqref{CI} using \eqref{SDB S}) are all asymptotically exact for $\psi$ as $n,s\to\infty$ with $s\leq n$.\label{main variant}
\end{theorem}

The proof of Theorem \ref{main variant} uses a similar roadmap as Theorem \ref{main} and the subsampling analogs of Assumption \ref{bp}, which hold under the additional technical measurability condition on $\mathcal F_\delta$ and have been used to justify the original version of these subsampling methods.
We make a couple of remarks. First, the original Bag of Little Bootstraps suggests to use several, or even a growing number of different subsample estimates, instead of a fixed subsample as we have presented earlier. Then from each subsample, an adequate number of resamples are drawn to obtain a quantile that informs an interval limit. These quantile estimates from different subsamples are averaged to obtain the final interval limit. Note that this suggestion requires a multiplicative amount of computation effort, and is motivated from better convergence rates (\cite{kleiner2014scalable} Theorems 2 and 3). Here, we simply use one subsample as we focus on the case of small Monte Carlo replication budget, but the modification to include multiple subsamples is feasible, evidenced by the validity of Cheap Subsampled Double Bootstrap which essentially uses a different subsample in each bootstrap replication. Second, in addition to allowing a smaller data size in model refitting, the $m$-out-of-$n$ Bootstrap, and also the closely related subsampling without replacement or so-called $\binom{n}{m}$ sampling, as well as sample splitting (\cite{politis1999subsampling,bickel2012resampling}), are all motivated as remedies to handle non-smooth functions where standard bootstraps could fail.  Instead of conditional weak convergence which we utilize in this work, these approaches are shown to provide consistent estimates based on symmetric statistics (\cite{politis1994large}).

\section{Numerical Results}\label{sec:numerics}
We test the numerical performances of our Cheap Bootstrap and compare it with baseline bootstrap approaches. Specifically, we consider elementary variance and correlation estimation (Section \ref{sec:elementary}), regression (Section \ref{sec:linear}), simulation input uncertainty quantification (Section \ref{sec:simulation}) and deep ensemble prediction (Section \ref{sec:deep}). Due to space limit, we show other examples in Appendix \ref{sec:add numerics}. Among these examples, Cheap Subsampling is also investigated in Section \ref{sec:linear} and nested sampling issues arise in Sections \ref{sec:simulation} and \ref{sec:deep}.



\subsection{Elementary Examples}\label{sec:elementary}

We consider four setups. In the first two setups, we estimate the variance of a distribution using the sample variance. We consider a folded standard normal (i.e., $|N(0,1)|$) and double exponential with rate 1 (i.e., $Sgn\times Exp(1)$ where $Sgn=+1$ or $-1$ with equal probability and is independent of $Exp(1)$) as the distribution. In the last two setups, we estimate the correlation using the sample correlation. We consider bivariate normal with mean zero, unit variance and correlation $0.5$, and bivariate lognormal (i.e., $(e^{Z_1},e^{Z_2})$ where $(Z_1,Z_2)$ is the bivariate normal just described). We use a sample size $n=1000$ in all cases. We set the confidence level $1-\alpha=95\%$. The ground truths of all the setups are known, namely $1-2/\pi$, $2$, $0.5$ and $(e^{3/2}-e)/(e^2-e)$ respectively (these examples are also used in, e.g., \cite{schenker1985qualms,diciccio1992analytical,lee1995asymptotic}).

We run Cheap Bootstrap using a small number of resamples  $B=1,2,3,4,5,10$, and compare with the Basic Bootstrap, Percentile Bootstrap and Standard Error Bootstrap. For each setting, we repeat the experiments $1000$ times and report the empirical coverage and interval width mean and standard deviation. Table \ref{table:var} shows the performances for variance and correlation estimation that are similar to Table \ref{table:comparison} in Section \ref{sec:basic numerics}. Cheap Bootstrap gives rise to accurate coverage ($91\%-96\%$) in all considered cases, including when $B$ is as low as $1$. On the other hand, all baseline methods fail to generate two-sided intervals when $B=1$, and encounter significant under-coverage when $B=2$ to $5$ (e.g., as low as $28\%$ for $B=2$ and $58\%$ when $B=5$). When $B=10$, these baseline methods start to catch up, with Standard Error Bootstrap rising to a coverage level close to $90\%$. These observations coincide with our theory of Cheap Bootstrap presented in Sections \ref{sec:basic} and \ref{sec:theory}, and also that the baseline methods are all designed to work under a large $B$.

Regarding the interval width, its mean and standard deviation for Cheap Bootstrap are initially large at $B=1$, signifying a price on statistical efficiency when the computation budget is very small. These numbers drop quickly when $B$ increases from $1$ to $2$ and continue to drop at a decreasing rate as $B$ increases further. In contrast, all baseline methods have lower width means and standard deviations than our Cheap Bootstrap, with an increasing trend for the mean as $B$ increases while the standard deviation remains roughly constant for each method. When $B=10$, the Cheap Bootstrap interval has generally comparable width mean and standard deviation with the baselines, though still larger. Note that while the baseline methods generate shorter intervals, they have significant under-coverage in the considered range of $B$.

\begin{table}[!ht]
\caption{Interval performances with different bootstrap methods: Cheap, Basic, Percentile and Standard Error Bootstrap, for variance and correlation estimation, at nominal confidence level $95\%$ and with sample size $n=1000$.}
\centering
{\footnotesize
\begin{tabular}{cc|cc|cc|cc|cc}
&\multirow{7}{*}{$B$}&\multicolumn{2}{c}{Variance of}&\multicolumn{2}{|c}{Variance of}&\multicolumn{2}{|c}{Correlation of}&\multicolumn{2}{|c}{Correlation of}\\
&&\multicolumn{2}{c}{Folded normal}&\multicolumn{2}{|c}{Double exponential}&\multicolumn{2}{|c}{Bivariate normal}&\multicolumn{2}{|c}{Bivariate lognormal}\\
\cline{3-10}
   &  & Coverage & Width & Coverage & Width & Coverage & Width & Coverage & Width \\
  &  & (margin & mean & (margin & mean & (margin & mean & (margin & mean \\
  &  & of error) & (st. dev.) & of error) & (st. dev.) & of error) & (st. dev.) & of error) & (st. dev.) \\
   \hline
  & & & & & & & & & \\
Cheap & 1 & 0.95\ (0.01) & 0.38\ (0.29) & 0.94\ (0.01) & 2.84\ (2.27) & 0.93\ (0.02) & 0.47\ (0.37) & 0.95\ (0.01) & 1.03\ (0.83) \\
  Basic & 1 & NA & NA & NA & NA & NA & NA & NA & NA \\
  Per. & 1 & NA & NA & NA & NA & NA & NA & NA & NA \\
  S.E. & 1 & NA & NA & NA & NA & NA & NA & NA & NA \\
  & & & & & & & & & \\
    \hline
  & & & & & & & & & \\
  Cheap & 2 & 0.95\ (0.01) & 0.15\ (0.08) & 0.94\ (0.01) & 1.10\ (0.60) & 0.95\ (0.01) & 0.18\ (0.10) & 0.94\ (0.01) & 0.38\ (0.25) \\
  Basic & 2 & 0.30\ (0.03) & 0.02\ (0.02) & 0.31\ (0.03) & 0.16\ (0.13) & 0.32\ (0.03) & 0.03\ (0.02) & 0.28\ (0.03) & 0.06\ (0.05) \\
  Per. & 2 & 0.33\ (0.03) & 0.02\ (0.02) & 0.32\ (0.03) & 0.16\ (0.13) & 0.32\ (0.03) & 0.03\ (0.02) & 0.32\ (0.03) & 0.06\ (0.05) \\
  S.E. & 2 & 0.68\ (0.03) & 0.06\ (0.05) & 0.69\ (0.03) & 0.45\ (0.35) & 0.69\ (0.03) & 0.07\ (0.06) & 0.65\ (0.03) & 0.15\ (0.14) \\
  & & & & & & & & & \\
  \hline
   & & & & & & & & & \\
 Cheap & 3 & 0.95\ (0.01) & 0.11\ (0.05) & 0.94\ (0.02) & 0.81\ (0.37) & 0.95\ (0.01) & 0.14\ (0.06) & 0.94\ (0.01) & 0.29\ (0.15) \\
  Basic & 3 & 0.50\ (0.03) & 0.03\ (0.02) & 0.48\ (0.03) & 0.23\ (0.13) & 0.48\ (0.03) & 0.04\ (0.02) & 0.46\ (0.03) & 0.08\ (0.05) \\
  Per. & 3 & 0.50\ (0.03) & 0.03\ (0.02) & 0.48\ (0.03) & 0.23\ (0.13) & 0.52\ (0.03) & 0.04\ (0.02) & 0.46\ (0.03) & 0.08\ (0.05) \\
  S.E. & 3 & 0.81\ (0.02) & 0.07\ (0.04) & 0.78\ (0.03) & 0.48\ (0.27) & 0.81\ (0.02) & 0.08\ (0.04) & 0.80\ (0.02) & 0.17\ (0.11) \\
  & & & & & & & & & \\
  \hline
  & & & & & & & & & \\
  Cheap & 4 & 0.96\ (0.01) & 0.10\ (0.04) & 0.96\ (0.01) & 0.73\ (0.28) & 0.94\ (0.01) & 0.12\ (0.05) & 0.94\ (0.02) & 0.26\ (0.12) \\
  Basic & 4 & 0.62\ (0.03) & 0.04\ (0.02) & 0.62\ (0.03) & 0.29\ (0.13) & 0.61\ (0.03) & 0.05\ (0.02) & 0.54\ (0.03) & 0.10\ (0.05) \\
  Per. & 4 & 0.62\ (0.03) & 0.04\ (0.02) & 0.60\ (0.03) & 0.29\ (0.13) & 0.61\ (0.03) & 0.05\ (0.02) & 0.59\ (0.03) & 0.10\ (0.05) \\
  S.E. & 4 & 0.86\ (0.02) & 0.07\ (0.03) & 0.86\ (0.02) & 0.50\ (0.22) & 0.86\ (0.02) & 0.09\ (0.04) & 0.81\ (0.02) & 0.17\ (0.09) \\
  & & & & & & & & & \\
  \hline
  & & & & & & & & & \\
  Cheap & 5 & 0.95\ (0.01) & 0.10\ (0.03) & 0.95\ (0.01) & 0.68\ (0.24) & 0.94\ (0.01) & 0.12\ (0.04) & 0.91\ (0.02) & 0.25\ (0.12) \\
  Basic & 5 & 0.67\ (0.03) & 0.05\ (0.02) & 0.67\ (0.03) & 0.32\ (0.13) & 0.66\ (0.03) & 0.06\ (0.02) & 0.58\ (0.03) & 0.12\ (0.06) \\
  Per. & 5 & 0.66\ (0.03) & 0.05\ (0.02) & 0.65\ (0.03) & 0.32\ (0.13) & 0.65\ (0.03) & 0.06\ (0.02) & 0.61\ (0.03) & 0.12\ (0.06) \\
 S.E. & 5 & 0.87\ (0.02) & 0.07\ (0.03) & 0.86\ (0.02) & 0.51\ (0.20) & 0.86\ (0.02) & 0.09\ (0.03) & 0.83\ (0.02) & 0.19\ (0.10) \\
  & & & & & & & & & \\
  \hline
  & & & & & & & & & \\
  Cheap & 10 & 0.93\ (0.02) & 0.08\ (0.02) & 0.94\ (0.01) & 0.62\ (0.17) & 0.94\ (0.01) & 0.10\ (0.02) & 0.91\ (0.02) & 0.21\ (0.09) \\
  Basic & 10 & 0.79\ (0.03) & 0.06\ (0.02) & 0.80\ (0.02) & 0.43\ (0.13) & 0.83\ (0.02) & 0.07\ (0.02) & 0.73\ (0.03) & 0.15\ (0.06) \\
  Per. & 10 & 0.79\ (0.03) & 0.06\ (0.02) & 0.80\ (0.02) & 0.43\ (0.13) & 0.80\ (0.02) & 0.07\ (0.02) & 0.78\ (0.03) & 0.15\ (0.06) \\
  S.E. & 10 & 0.90\ (0.02) & 0.07\ (0.02) & 0.91\ (0.02) & 0.54\ (0.15) & 0.92\ (0.02) & 0.09\ (0.02) & 0.86\ (0.02) & 0.19\ (0.08) \\
 \end{tabular}
}
\label{table:var}
\end{table}

\subsection{Regression Problems}\label{sec:linear}
We apply our Cheap Bootstrap and compare with standard approaches on a linear regression. The example is adopted from \cite{sengupta2016subsampled} \S4 and \cite{kleiner2014scalable} \S4.
We fit a model $Y=\beta_1X_{1}+\cdots+\beta_dX_{d}+\epsilon$, where we set dimension $d=100$ and use data $(X_{1,i},\ldots,X_{d,i},Y_i)$ of size $n=10^5$ to fit the model and estimate the coefficients $\beta_j$'s. The ground truth is set as $X_j\sim t_3$ for $j=1,\ldots,d$, $\epsilon\sim N(0,10)$, and $\beta_j=1$ for $j=1,\ldots,d$.


In addition to full-size Cheap Bootstrap, we run the three variants of Cheap Subsampling, namely Cheap $m$-out-of-$n$ Bootstrap, Cheap Bag of Little Bootstraps (BLB), and Cheap Subsampled Double Bootstrap (SDB). Moreover, we compare each cheap bootstrap method with its standard quantile-based counterpart (i.e., the basic bootstrap described in Section \ref{sec:comparisons} and the subsample counterparts at the beginning of Section \ref{sec:subsampling}, with one fixed subsample in the Bag of Little Bootstrap case) to compute $95\%$ confidence intervals for the first coefficient $\beta_1$. For the subsampling methods, we use subsample size $n^{0.6}=1000$, which is also used in \cite{sengupta2016subsampled} and \cite{kleiner2014scalable}. We use in total $50$ resamples for each method.
We repeat the experiments $1000$ times, each time we regenerate a new synthetic data set. For the number of resamples starting from 1 to $50$, we compute the empirical coverage probability and the mean and standard deviation of the confidence interval width for the first coefficient.

\begin{figure*}[tb]
\vskip 0.2in
\begin{center}
\subfigure
{\includegraphics[width=.24\columnwidth,height=.15\textheight]{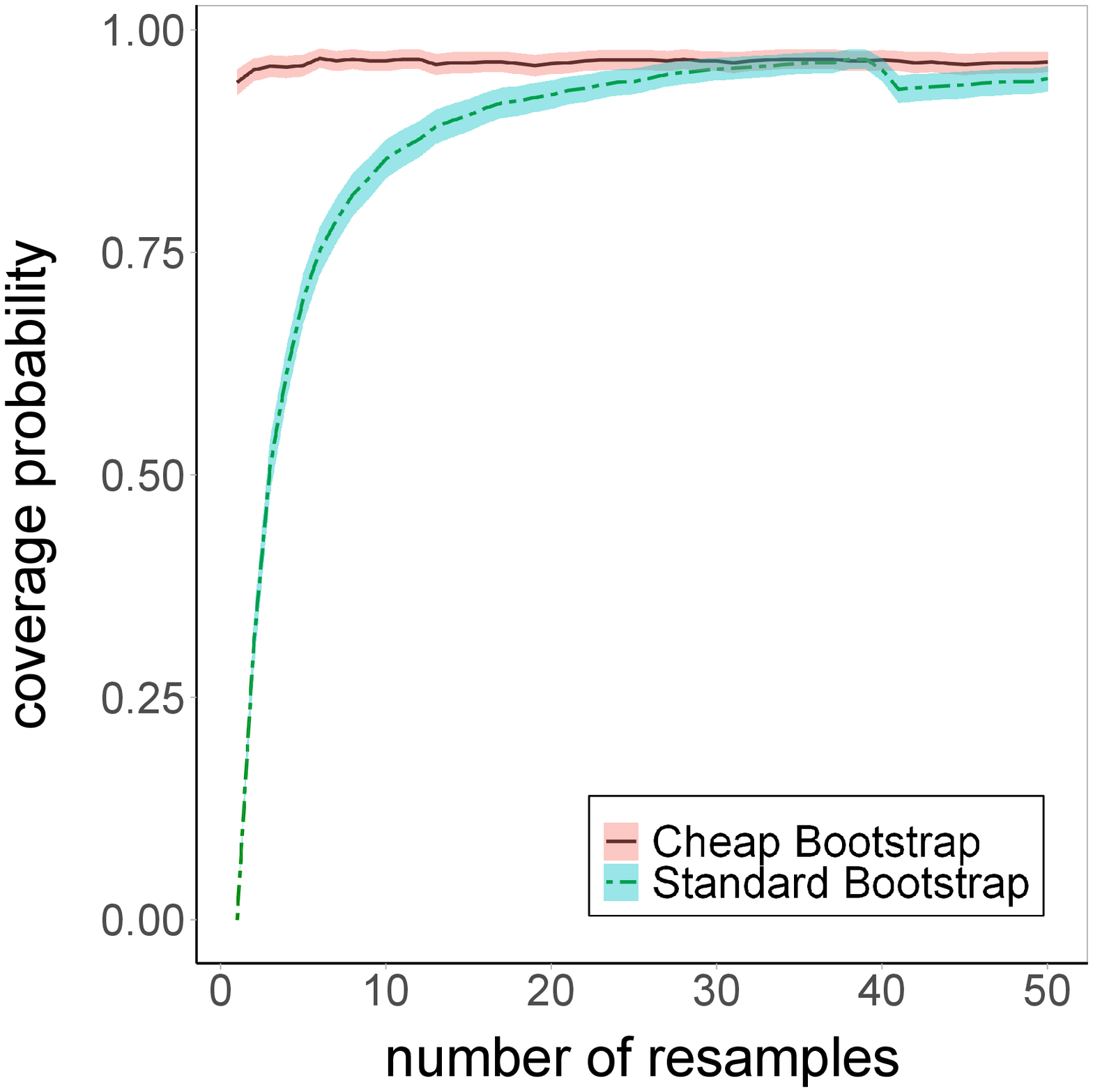}}
\subfigure
{\includegraphics[width=.24\columnwidth,height=.15\textheight]{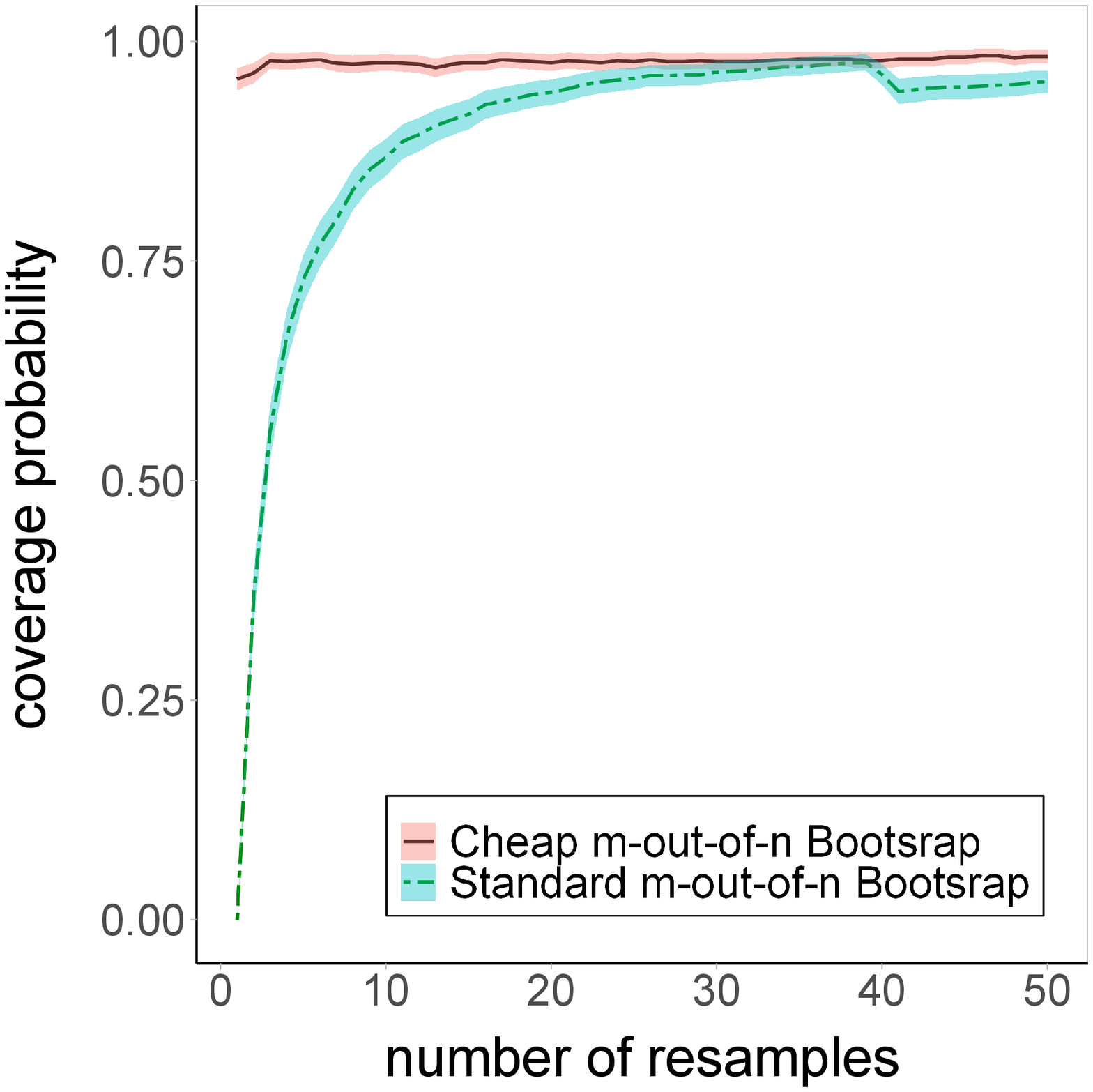}}
\subfigure
{\includegraphics[width=.24\columnwidth,height=.15\textheight]{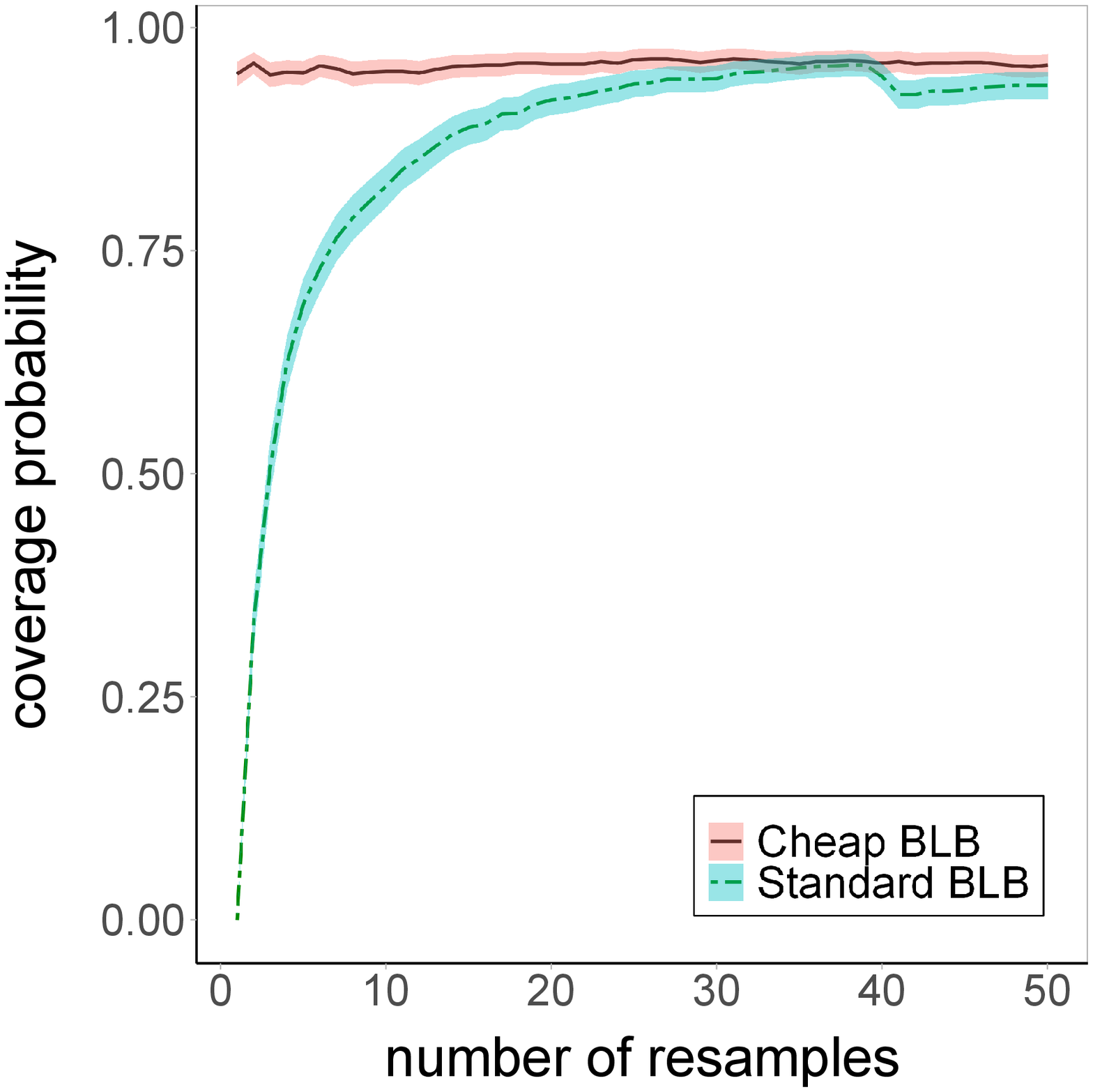}}
\subfigure
{\includegraphics[width=.24\columnwidth,height=.15\textheight]{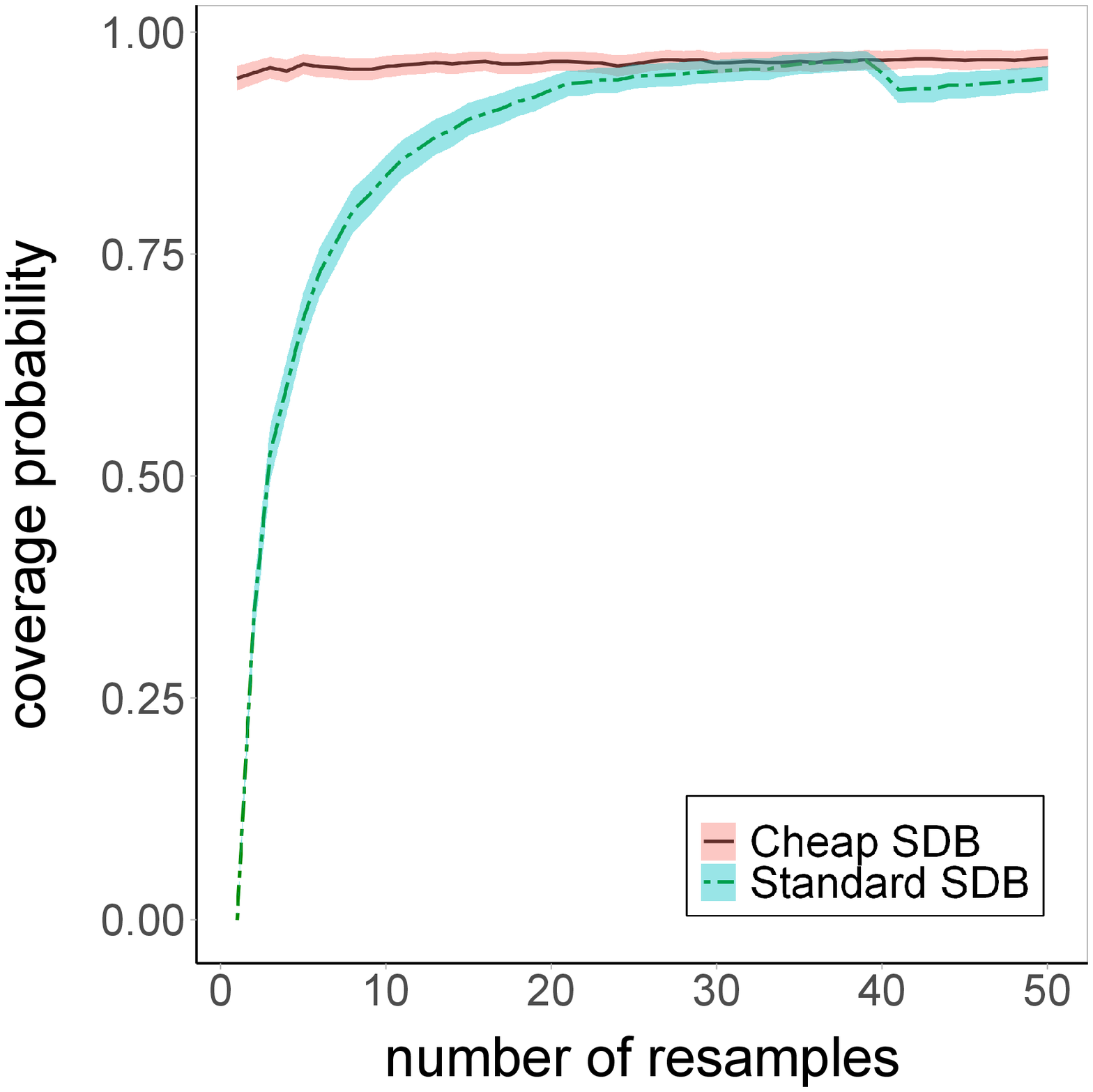} }
\caption{Confidence interval coverage probabilities of Standard versus Cheap Bootstrap methods in linear regression. Nominal confidence level $=95\%$ and sample size $n=10^5$. Shaded areas depict the associated confidence intervals of the coverage probability estimates from 1000 experimental repetitions.}
\label{fig:cov}
\end{center}
\vskip -0.2in
\end{figure*}
\begin{figure*}[tb]
\vskip 0.2in
\begin{center}
\subfigure
{\includegraphics[width=.24\columnwidth,height=.15\textheight]{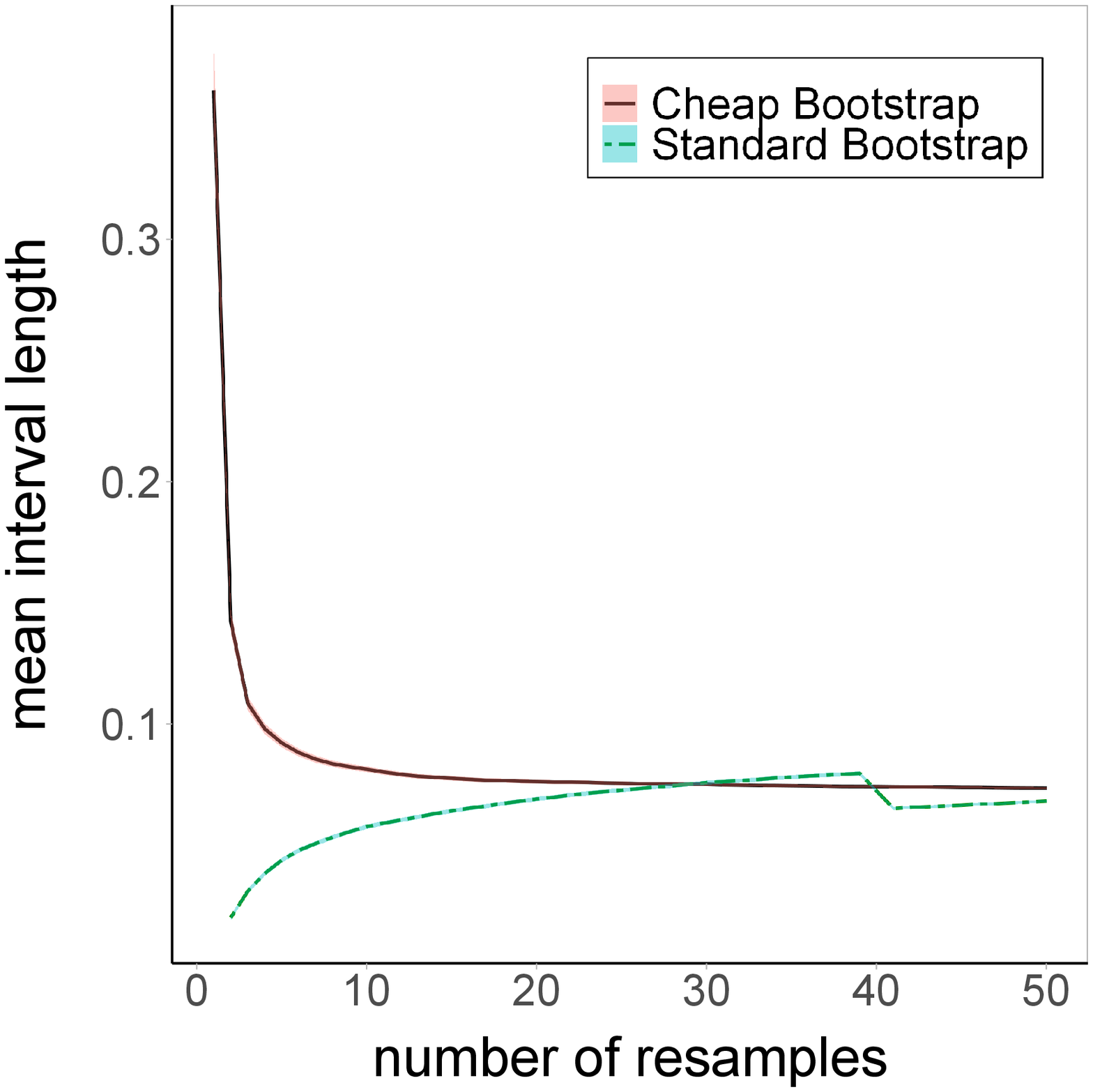} }
\subfigure
{\includegraphics[width=.24\columnwidth,height=.15\textheight]{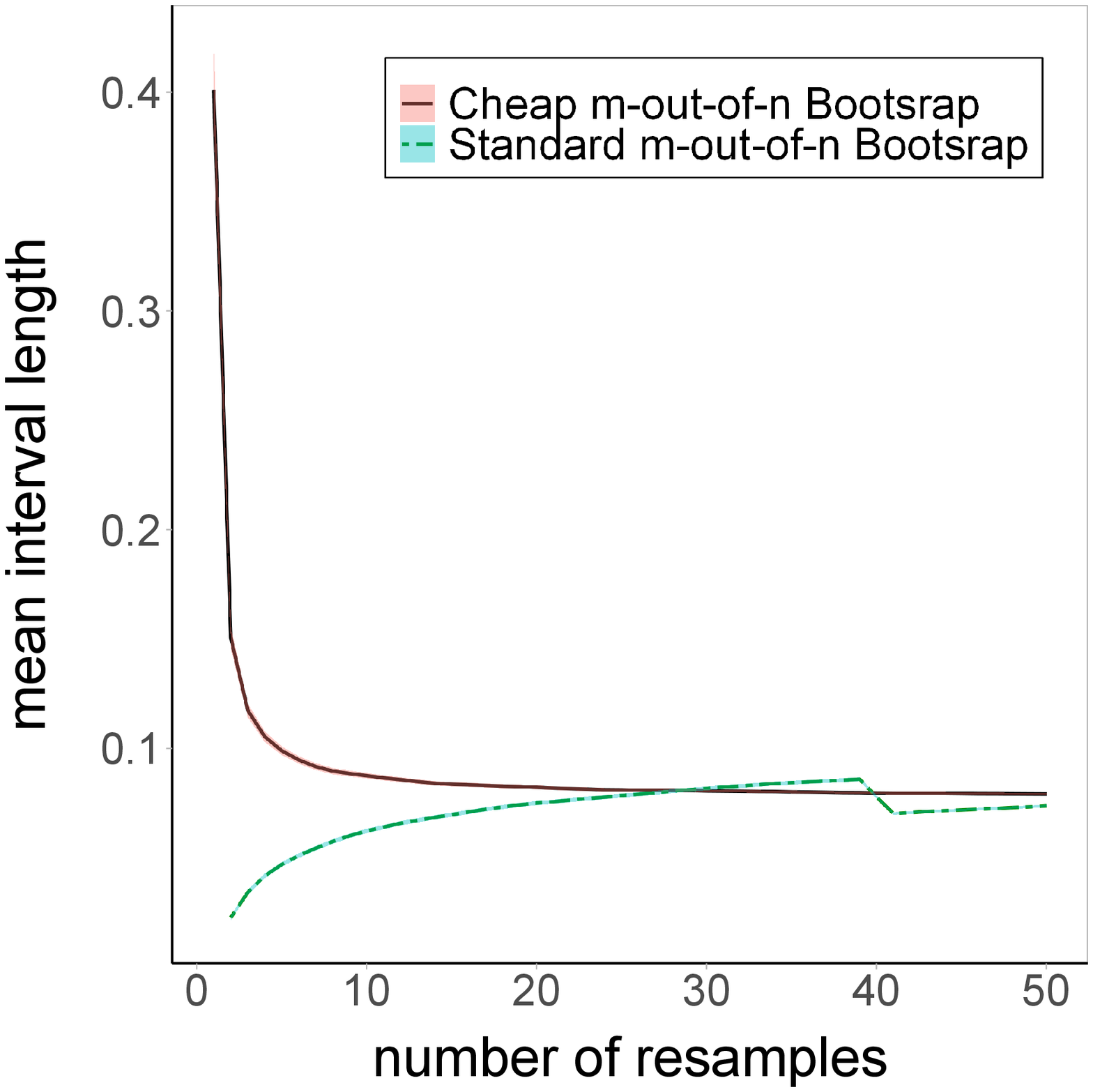}}
\subfigure
{\includegraphics[width=.24\columnwidth,height=.15\textheight]{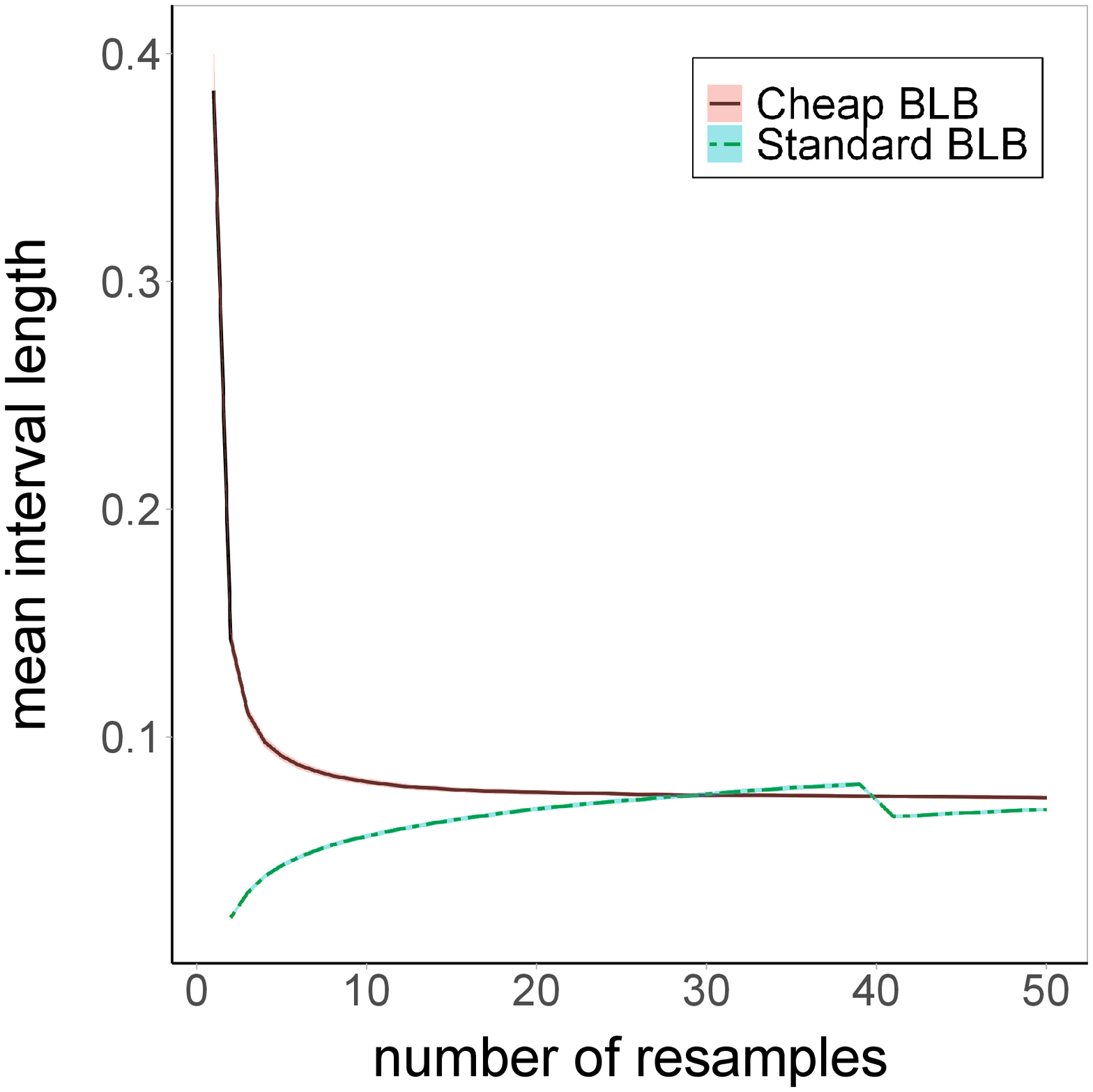}}
\subfigure
{\includegraphics[width=.24\columnwidth,height=.15\textheight]{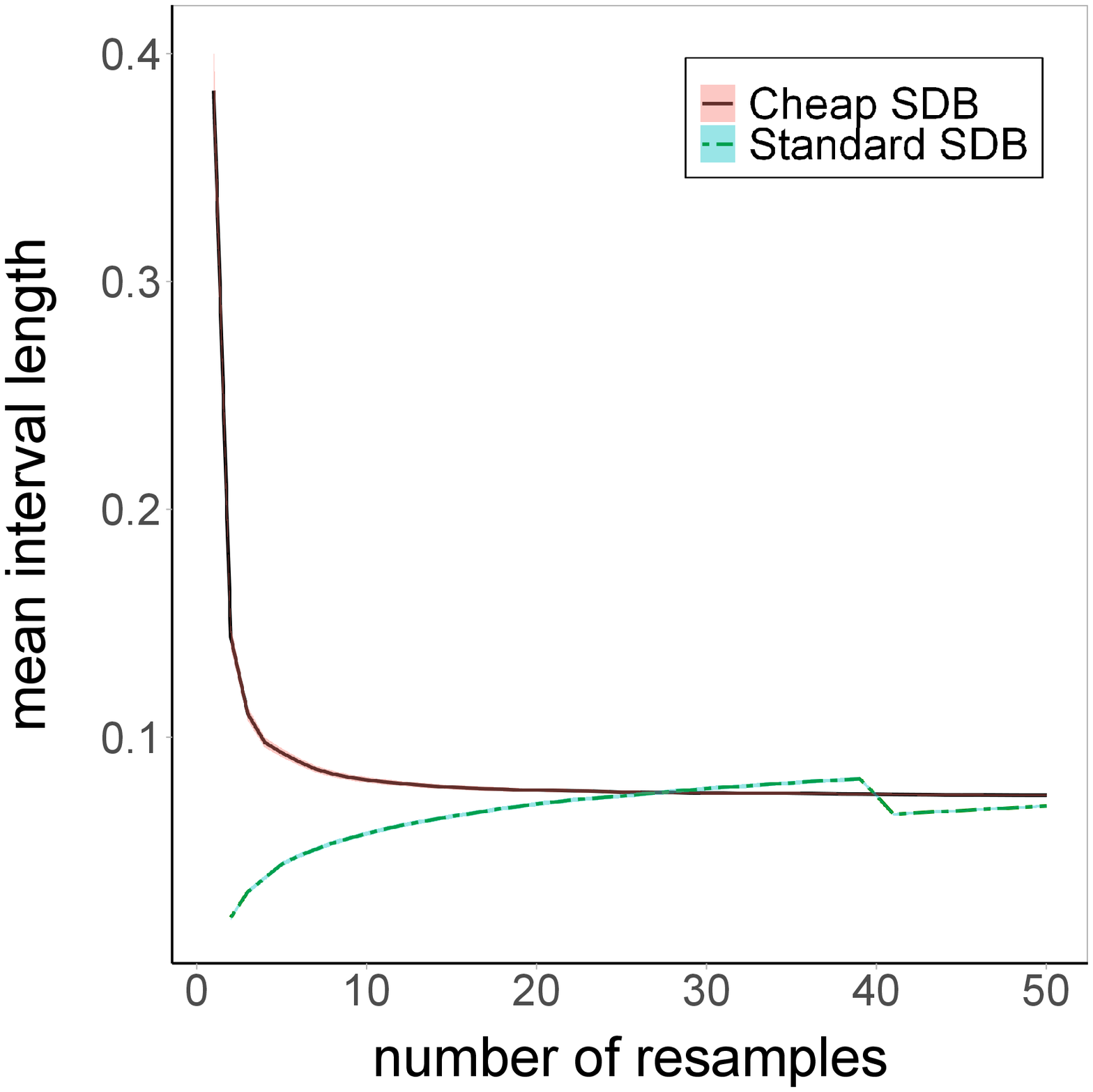} }
\caption{Mean confidence interval widths of Standard versus Cheap Bootstrap methods in linear regression. Nominal confidence level $=95\%$ and sample size $n=10^5$. Shaded areas depict the associated confidence intervals of the mean width estimates from 1000 experimental repetitions.}
\label{fig:len}
\end{center}
\vskip -0.2in
\end{figure*}
\begin{figure*}[tb]
\vskip 0.2in
\begin{center}
\subfigure
{\includegraphics[width=.24\columnwidth,height=.15\textheight]{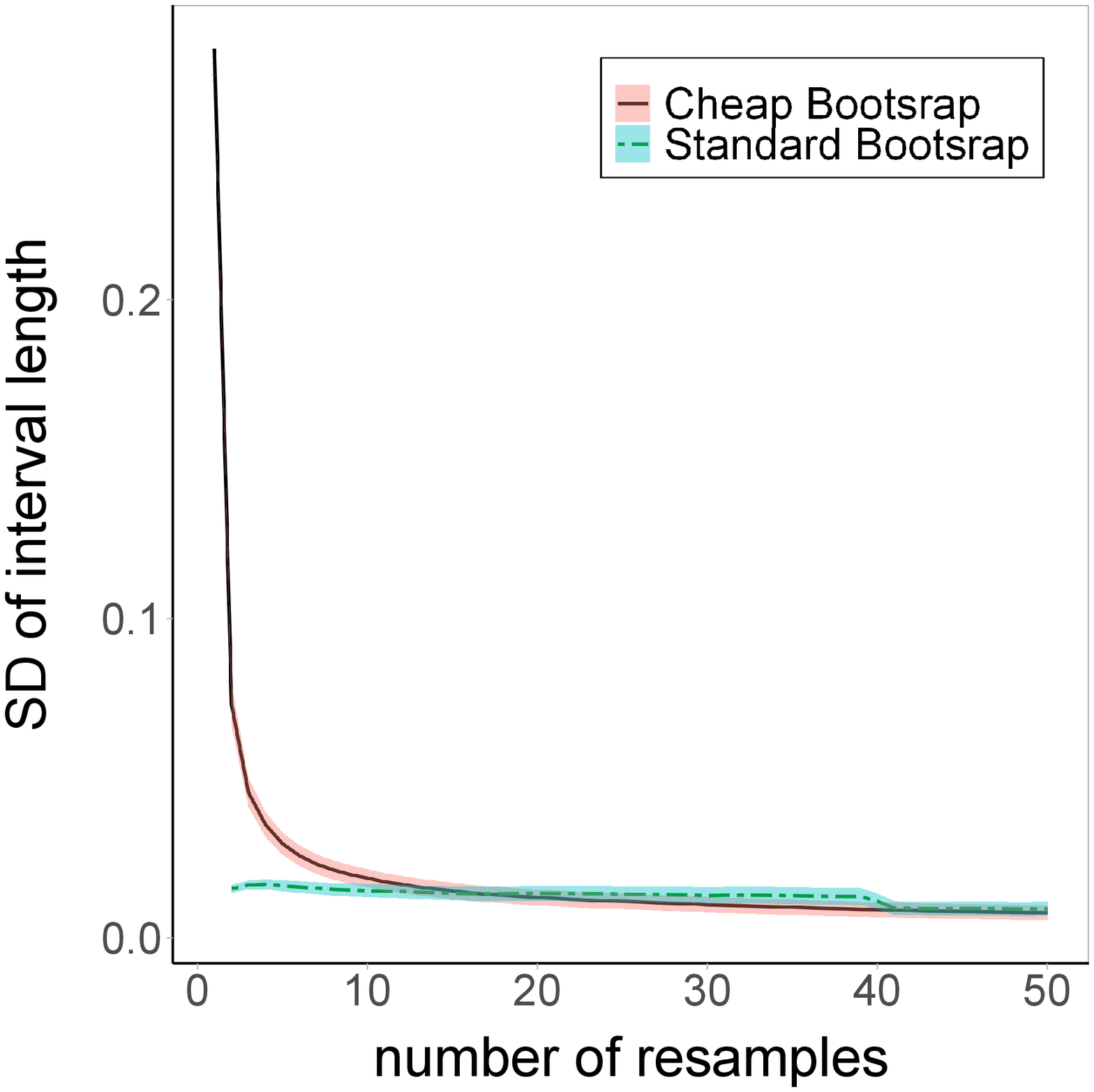}}
\subfigure
{\includegraphics[width=.24\columnwidth,height=.15\textheight]{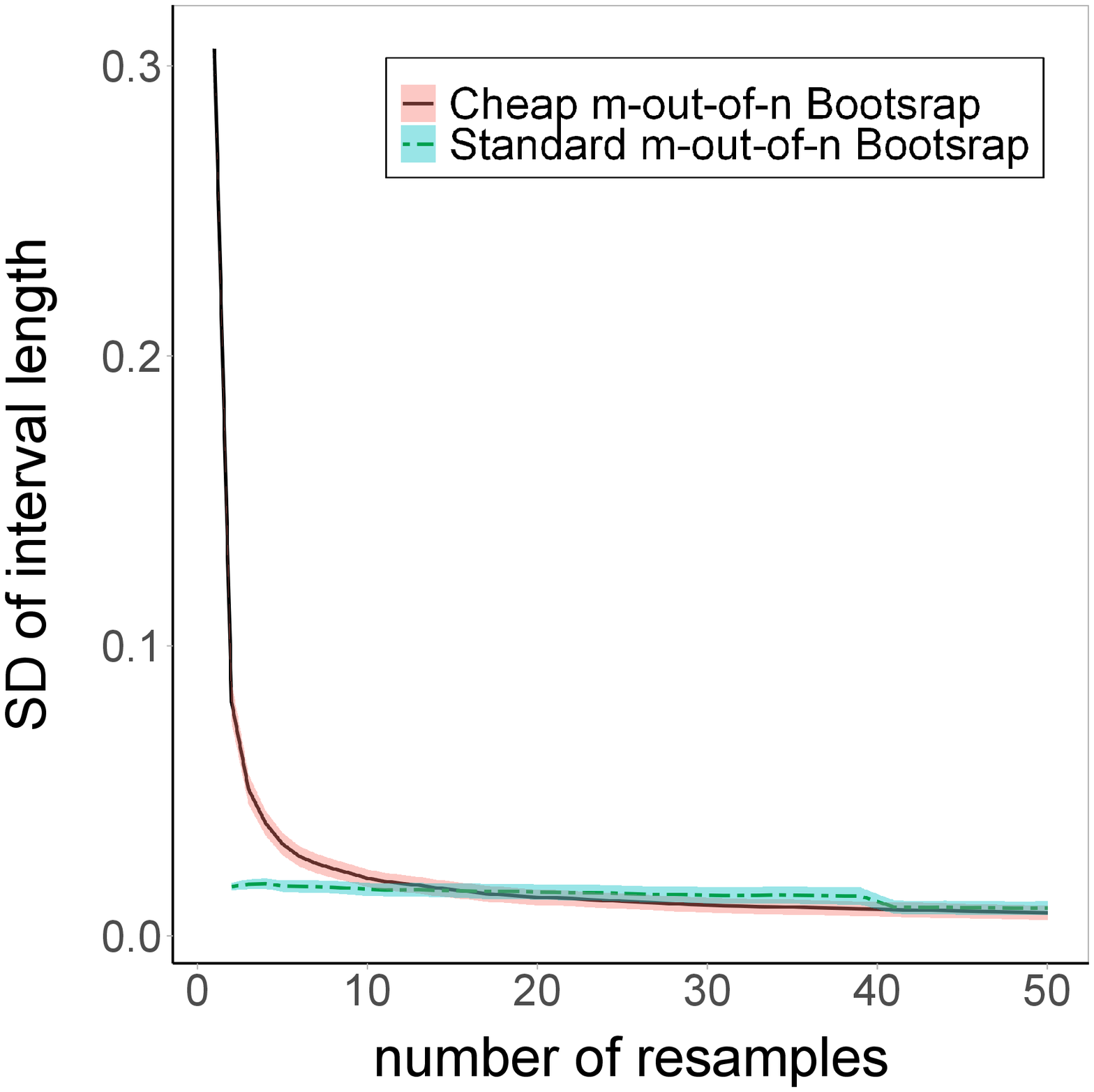}}
\subfigure
{\includegraphics[width=.24\columnwidth,height=.15\textheight]{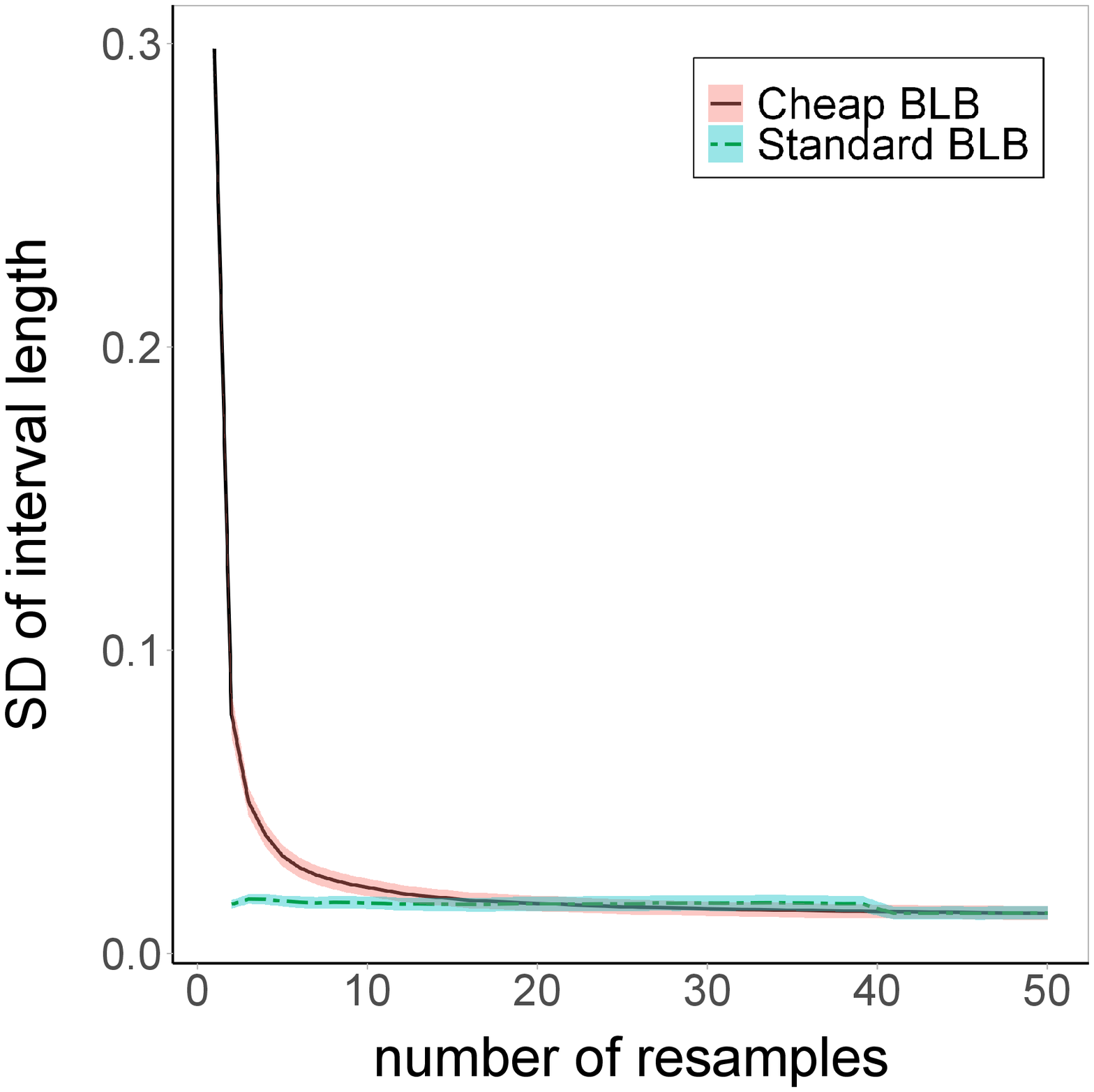}}
\subfigure
{\includegraphics[width=.24\columnwidth,height=.15\textheight]{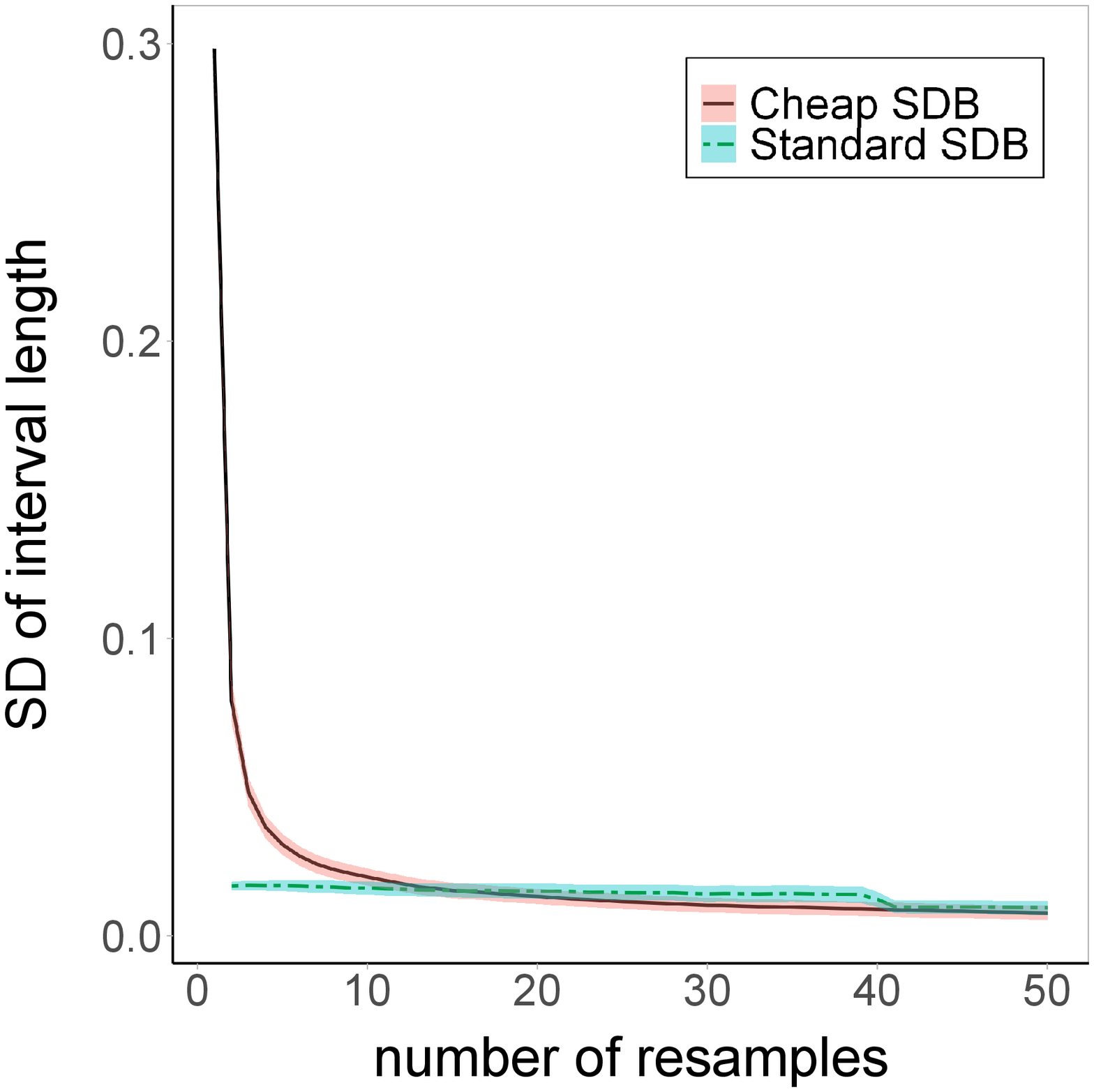} }
\caption{Standard deviations of confidence interval widths of Standard versus Cheap Bootstrap methods in linear regression.  Nominal confidence level $=95\%$ and sample size $n=10^5$. Shaded areas depict the associated confidence intervals of the standard deviation estimates from 1000 experimental repetitions.}
\label{fig:sd}
\end{center}
\vskip -0.2in
\end{figure*}
Figure \ref{fig:cov} shows the empirical coverage probability against the number of resamples from the $1000$ experimental repetitions. In the first graph, we see that Standard Bootstrap falls short of the nominal coverage, severely when $B$ is small and  gradually improves as $B$ increases. For example, the coverage probability at $B=5$ is $70\%$ and at $B=10$ is $86\%$. On the other hand, Cheap Bootstrap gives close to the nominal coverage starting from the first replication ($94\%$), and the coverage stays between $94\%$ and $97\%$ throughout. The comparison is similar for subsampling methods, where with small $B$ Standard $m$-out-of-$n$, Bag of Little Bootstrap and Subsampled Double Bootstrap all under-cover and gradually approach the nominal level as $B$ increases. For a more specific comparison, at $B=5$ for instance, Standard $m$-out-of-$n$, Bags of Little Bootstrap and Subsampled Double Bootstrap have under-coverages of $73\%$, $69\%$ and $68\%$ respectively, while the Cheap counterparts attain $98\%$, $95\%$ and $96\%$. We note that the ``sudden drop'' in the coverage probability for all Standard bootstrap methods at the $41$-st resample in Figure \ref{fig:cov} (as well as all the following figures) is due to the discontinuity in the empirical quantile, which is calculated by inverting the empirical distribution and at this number of resamples the inverse switches from outputting the largest or smallest resample estimate to the second largest or smallest.

Figure \ref{fig:len} shows the mean interval width against the number of resamples (note that the confidence intervals for the mean width estimates are very narrow, meaning that the estimates are accurate, which makes the shaded regions in the figure difficult to visualize). These plots show that the mean interval widths of Cheap bootstrap methods are initially large when only one replication is used, then decrease sharply when the number of resamples increases and then continue to decrease at a slower rate. Figure \ref{fig:sd} shows the standard deviation of the interval width against the number of resamples, which follow a similar trend as the mean in Figure \ref{fig:len}. More specifically, the interval width mean and standard deviation of Cheap Bootstrap are $0.36$ and $0.28$ at $B=1$, drop sharply to $0.11$ and $0.05$ at $B=3$, further to $0.09$ and $0.03$ at $B=5$, and $0.08$ and $0.02$ at $B=10$, with negligible reduction afterwards.
These trends match our theory regarding the interval width pattern in Section \ref{sec:tightness}. On the other hand, the interval width means of Standard methods are initially small and increase with the number of resamples, while the standard deviations are roughly constant. Note that at $B=1$ it is not possible to generate confidence intervals using Standard bootstrap methods, and the interval width mean and standard deviation are undefined. Moreover, even though the interval widths in Standard bootstrap methods appear lower generally, these intervals can substantially under-cover the truth when $B$ is small.

We present another example on logistic regression in Appendix \ref{sec:logistic} which shows similar experimental observations discussed above.

\subsection{Input Uncertainty Quantification in Simulation Modeling}\label{sec:simulation}
We present an example on the input uncertainty quantification problem in simulation modeling discussed in Section \ref{sec:double}. Our target quantity of interest is the expected average waiting time of the first 10 customers in a single-server queueing system, where the interarrival times and service time are i.i.d. with ground truth distributions being exponential with rates $1$ and $1.1$ respectively. The queue starts empty and the first customer immediately arrives. This system is amenable to discrete-event simulation and has been used commonly in existing works in input uncertainty quantification (e.g., \cite{barton2014quantifying,song2015quickly,zouaoui2004accounting}). Suppose we do not know the interarrival time distribution but instead have external data of size $n=100$. To compute a point estimate of the expected waiting time, we use the empirical distribution constructed from the 100 observations, denoted $\hat P_n$, as an approximation to the interarrrival time distribution that is used to drive simulation runs. We conduct $R_0=50$ unbiased simulation runs, each giving output $\hat\psi_r(\hat P_n)$, and then average the outputs of these runs to get $(1/R_0)\sum_{r=1}^{R_0}\hat\psi_r(\hat P_n)$. 

To obtain a $95\%$ confidence interval, we use Cheap Bootstrap centered at original estimate and centered at resample mean presented in Sections \ref{sec:cen} and \ref{sec:nc} respectively. More precisely, the first approach constructs the interval $\mathcal I_O$ using \eqref{CI general nested} with \eqref{S center} and \eqref{q center}, while the second approach constructs the interval $\mathcal I_M$ using \eqref{CI general nested} with \eqref{S nc} and \eqref{q nc}. In both approaches, we set the simulation size for each resample estimate to equal that for the original estimate, i.e., $R=R_0=50$. To construct $\mathcal I_O$, we need to approximate the critical value $q_{O,1-\alpha/2}$ using \eqref{grid search}, where we set $N=100,000$ and discretize $\theta$ over a grid $0.01,0.02,\ldots,100$. Then we increase $q$ from $t_{B,1-\alpha/2}-0.5$ and iteratively check whether 
\begin{equation}
\min_{\theta\geq0}\frac{1}{N}\sum_{j=1}^NI\left(\frac{\theta V_1^j+V_2^j}{\sqrt{\frac{\theta^2+\rho^2}{B}\left(Y^j+{V_3^j}^2\right)-2\sqrt{\frac{\theta^2+\rho^2}{B}}V_3^jV_2^j+{V_2^j}^2}}\leq q\right)\label{grid search1}
\end{equation}
is greater than $1-\alpha/2$; if not, we increase $q$ by a step size $0.01$, until we reach a $q$ such that \eqref{grid search1} passes $1-\alpha/2$. We detail the computed values of $q_{O,1-\alpha/2}$ in Appendix \ref{sec:q compute}.

Using the above, we construct $\mathcal I_O$ and $\mathcal I_M$ at $B$ ranging from $1$ to $10$. For each $B$, we repeat our experiments $1000$ times to record the empirical coverage and interval width mean and standard deviation. To calculate the empirical coverage, we run 1 million simulation runs under the true interarrival and service time distributions to obtain an accurate estimate of the ground truth value. Table \ref{table:queue} shows the performances of $\mathcal I_O$ and $\mathcal I_M$, and also the comparisons with the Basic and Percentile Bootstraps that treat $\psi_{n,R}^{**b},b=1,\ldots,B$ in \eqref{re estimate nested} as the resample estimates. We see that the coverage of $\mathcal I_O$ is very close to the nominal confidence level $95\%$ at all considered $B$, ranging in $95\%-96\%$. The coverage of $\mathcal I_M$ is also close, ranging in $93\%-94\%$. The interval width performances for both $\mathcal I_O$ and $\mathcal I_M$ behave similarly as previous examples, with initially large width at $B=1$ or $B=2$, falling sharply when $B$ increases by 1 or 2, and flattening out afterwards. The Basic and Percentile Bootstraps have coverages substantially deviated from the nominal $95\%$, though they improve as $B$ increases. The inferior coverages of these conventional methods can be caused by both the inadequate $B$ and also the ignorance of the simulation noise in their implementation.


\begin{table}[ht]
\caption{Interval performances using Cheap Bootstrap centered at original estimate $\mathcal I_O$, Cheap Bootstrap centered at resample mean $\mathcal I_M$, Basic Bootstrap and Percentile Bootstrap, at nominal confidence level $95\%$ for input uncertainty quantification in a queueing system simulation.}
\centering
{\footnotesize
\begin{tabular}{c|cc|cc|cc|cc}
\multirow{6}{*}{$B$}&\multicolumn{2}{c}{Cheap bootstrap centered}&\multicolumn{2}{|c}{Cheap bootstrap centered}&\multicolumn{2}{|c}{Basic bootstrap}&\multicolumn{2}{|c}{Percentile bootstrap}\\
&\multicolumn{2}{c}{at original estimate}&\multicolumn{2}{|c}{at resample mean}&\multicolumn{2}{|c}{}&\multicolumn{2}{|c}{}\\
   & Coverage & Width & Coverage & Width & Coverage & Width & Coverage & Width\\
   & (margin & mean & (margin & mean & (margin & mean & (margin & mean\\
    & of error) & (st. dev.) & of error) & (st. dev.) & of error) & (st. dev.) & of error) & (st. dev.)\\
   \hline
   1& 0.96\ (0.01) & 6.73\ (5.41) & NA & NA & NA & NA & NA & NA \\
2& 0.95\ (0.01) & 2.55\ (1.50) & 0.94\ (0.02) & 5.84\ (4.86) & 0.35\ (0.03)     &     0.33\ (0.26) &  0.27\ (0.03)& 0.33\ (0.26)\\
3& 0.95\ (0.01) & 1.97\ (0.99) & 0.94\ (0.01) & 2.26\ (1.30) & 0.52\ (0.03) &  0.50\ (0.29) & 0.41\ (0.03) & 0.50\ (0.29)\\
4& 0.95\ (0.01) & 1.74\ (0.79) & 0.94\ (0.01) & 1.73\ (0.84) & 0.64\ (0.03)& 0.59\ (0.29) & 0.48\ (0.03)& 0.59\ (0.29)\\
5& 0.95\ (0.01) & 1.64\ (0.69) & 0.93\ (0.02) & 1.54\ (0.68) & 0.71\ (0.03)& 0.66\ (0.29) & 0.54\ (0.03) & 0.66\ (0.29)\\
6& 0.95\ (0.01) & 1.58\ (0.63) & 0.93\ (0.02) & 1.45\ (0.59) &0.75\ (0.03)&  0.72\ (0.29)& 0.58\ (0.03) & 0.72\ (0.29)\\
7& 0.95\ (0.01) & 1.54\ (0.58) & 0.93\ (0.02) & 1.40\ (0.53) & 0.77\ (0.03)& 0.76\ (0.29) & 0.61\ (0.03)& 0.76\ (0.29)\\
8& 0.95\ (0.01) & 1.50\ (0.55) & 0.93\ (0.02) & 1.35\ (0.49) & 0.79\ (0.03)&  0.80\ (0.29) & 0.64\ (0.03)& 0.80\ (0.29)\\
9& 0.95\ (0.01) & 1.48\ (0.53) & 0.93\ (0.02) & 1.33\ (0.47) &0.81\ (0.02)& 0.82\ (0.29) & 0.65\ (0.03)& 0.82\ (0.29)\\
10& 0.95\ (0.01) & 1.46\ (0.51) & 0.93\ (0.02) & 1.31\ (0.44) & 0.83\ (0.02)&  0.85\ (0.29)& 0.67\ (0.03)& 0.85\ (0.29)
\end{tabular}
}
\label{table:queue}
\end{table}

Compared with existing works on input uncertainty quantification, we highlight that the computational load in Cheap Bootstrap is significantly lower. According to Table \ref{table:queue}, $B$ can be taken as, for instance, 3 or 4 in $\mathcal I_O$ to obtain a reasonable confidence interval, so that the total computation cost is $BR$ which is 200 or 250 when we use $R=50$ (including the cost to compute the original point estimate). In contrast, bootstrap approaches proposed in the literature suggest to use $B$ at least $50$ in numerical examples (e.g., \cite{cheng2004calculation,song2015quickly}) or linearly dependent on the parameter dimension (e.g., $10$ times the dimension when using the so-called metamodel-assisted bootstrap in the parametric case; \cite{xie2014bayesian}).

\subsection{Deep Ensemble Prediction}\label{sec:deep}
We consider deep ensemble prediction described in Section \ref{sec:double}. We build a prediction model $y=f(\mathbf x)$ from i.i.d. supervised data $(\mathbf X_i,Y_i),i=1,\ldots,n$ of size $n=200$. Here the feature vector $\mathbf x=(x(1),\ldots,x(d))$ is 16-dimensional and $y\in\mathbb R$. Suppose each dimension of $\mathbf X_i$ is generated from the uniform distribution on $[0,1]$, and the ground-truth model is $y=\sum_{j=1}^{16}x(j)+N(0,1)$. We build a deep ensemble by training $R_0=5$ base neural networks each using an independent random initialization, and average these networks to obtain a predictor. More specifically, each base neural network has two fully connected layers with $1024$ hidden neurons, using rectified linear unit as activation. It is trained using the squared loss with $L_2$-regularization on the neuron weights, and Adam for gradient descent. All the weights in each neural network are initialized as independent $N(0,1)$ variables.

Like in Section \ref{sec:simulation}, we use Cheap Bootstrap centered at original estimate $\mathcal I_O$ and centered at resample mean $\mathcal I_M$ to construct $95\%$ confidence intervals for the prediction output. Here, we consider a test point with value $0.5$ in all dimensions. We set $R=R_0$ in both approaches, and use Table \ref{table:nested} in Appendix \ref{sec:q compute} to calibrate the critical value $q_{O,1-\alpha/2}$ in $\mathcal I_O$. We repeat the experiments $100$ times to take down statistics on the generated confidence intervals. Figure \ref{fig:cen} depicts the box plots of the confidence interval widths using $B=1$ to $10$, for $\mathcal I_O$ and $\mathcal I_M$ respectively. We see that the interval widths are relatively large for $B=1$ (in the case of $\mathcal I_O$) and for $B=2$ (in the case of $\mathcal I_M$ which is only well-defined starting from $B=2$), but decrease fast when $B$ increases. When $B=3$ (in the case of $\mathcal I_O$) and $B=4$ (in the case of $\mathcal I_M$), the interval widths appear to more or less stabilize. 

\begin{figure*}[tb]
\vskip 0.2in
\begin{center}
\subfigure
{\includegraphics[width=.49\columnwidth]{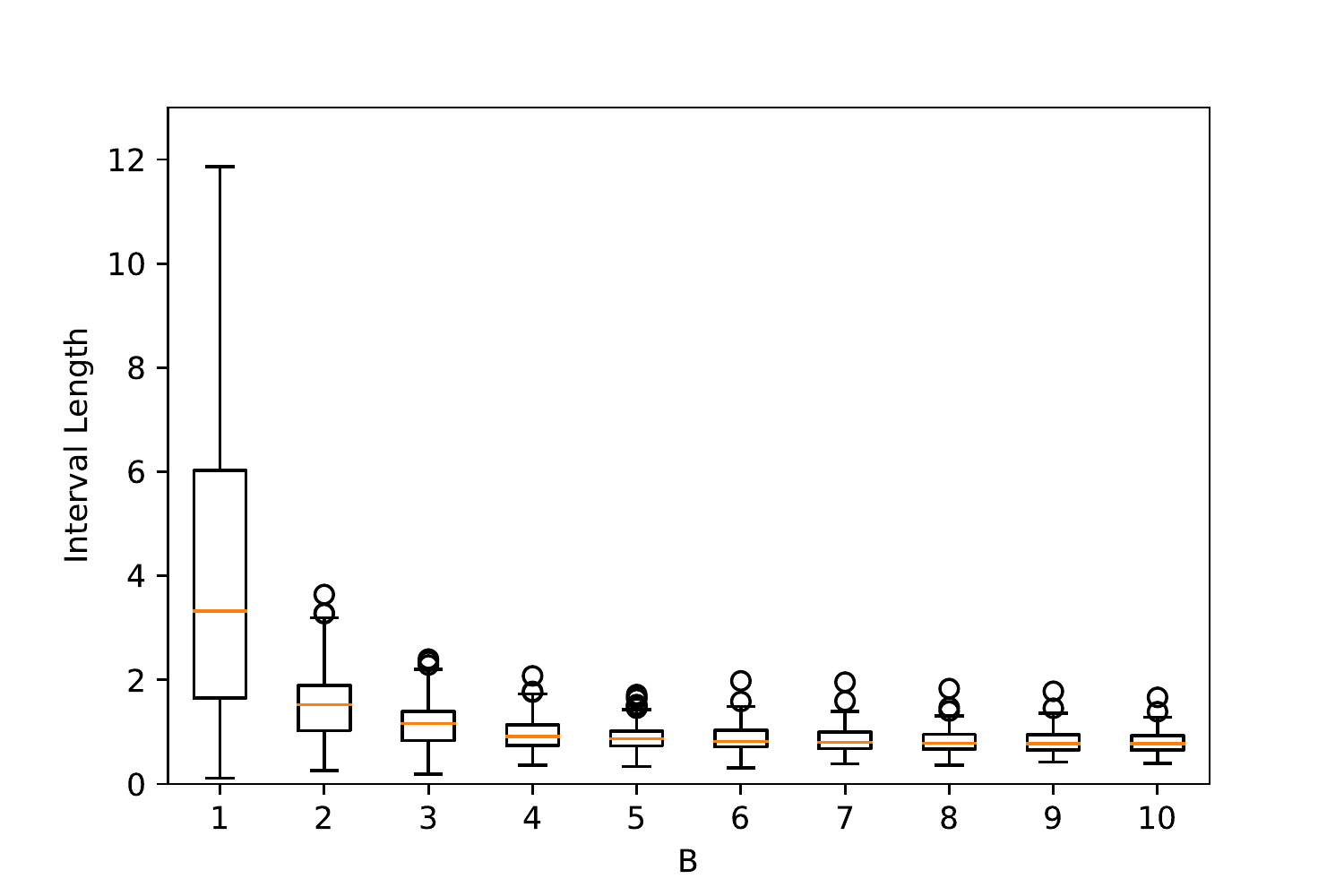} }
\subfigure
{\includegraphics[width=.49\columnwidth]{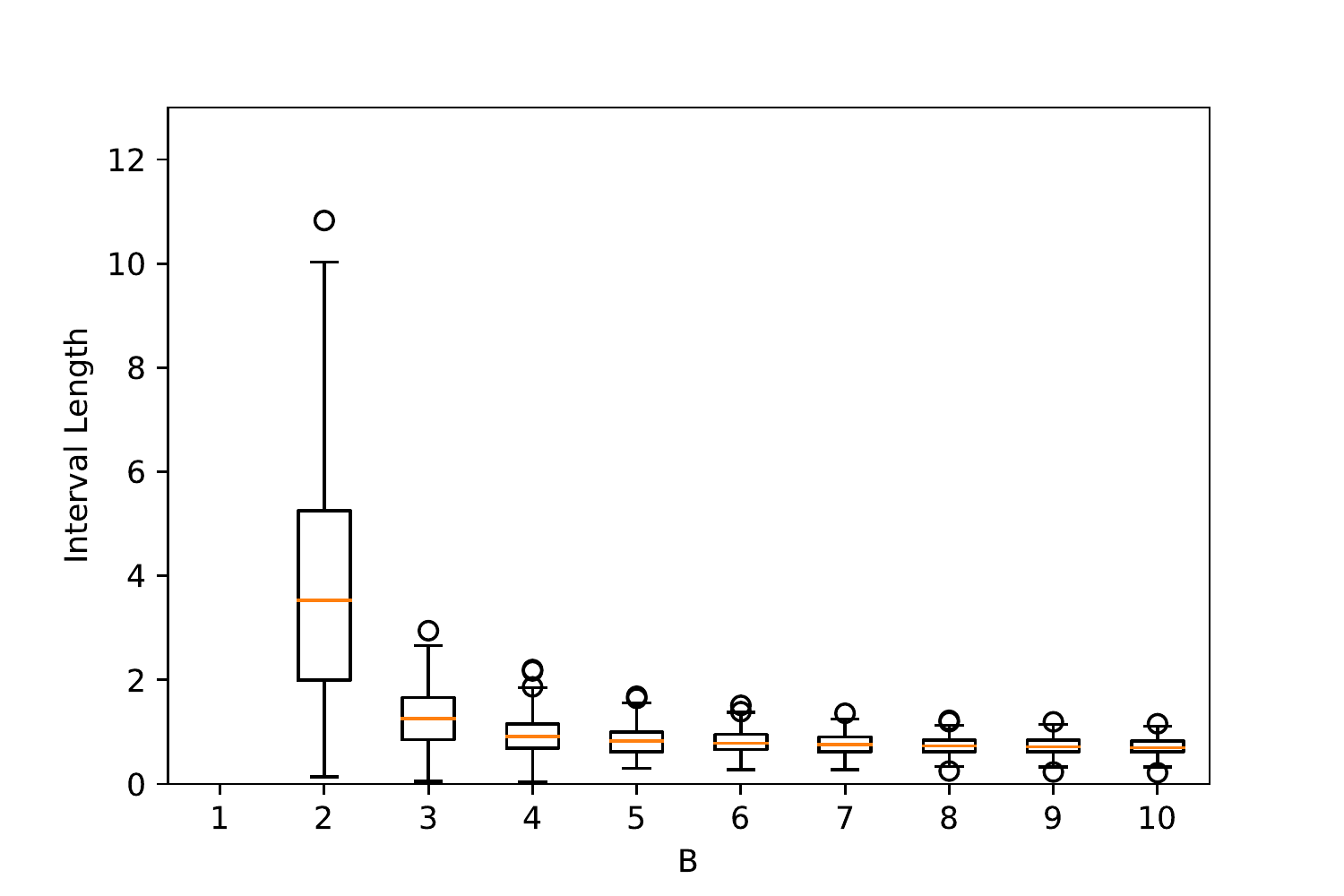}}
\caption{Box plots of $95\%$ confidence interval widths for deep ensemble prediction on synthetic data at test point that takes value $0.5$ in all dimensions, using Cheap Bootstrap centered at original estimate $\mathcal I_O$ (left graph) and Cheap Bootstrap centered at resample mean $\mathcal I_M$ (right graph).}
\label{fig:cen}
\end{center}
\vskip -0.2in
\end{figure*}

Next, we also train deep ensemble predictors for a real data set on Boston housing, available at https://www.kaggle.com/c/boston-housing. The data set has size $504$ and feature dimension $13$. Like in our synthetic data, we use $R_0=5$, and the same configurations of base neural networks and training methods. We split the data by $80\%$ to a training set and $20\%$ to a testing set, where only the training set is used to construct the predictor and also run our Cheap Bootstrap, with again $R=R_0=5$. Then we create the $95\%$ Cheap Bootstrap confidence intervals for the prediction values of the testing set. Figure \ref{fig:boston} shows the box plots of the widths of these intervals, namely $\mathcal I_O$ and $\mathcal I_M$, for $B$ ranging from $1$ to $10$. We see that the widths fall sharply when $B$ is small ($B=1,2$ in $\mathcal I_O$ and $B=2,3$ in $\mathcal I_M$) and stabilize quickly afterwards.

\begin{figure*}[tb]
\vskip 0.2in
\begin{center}
\subfigure
{\includegraphics[width=.49\columnwidth]{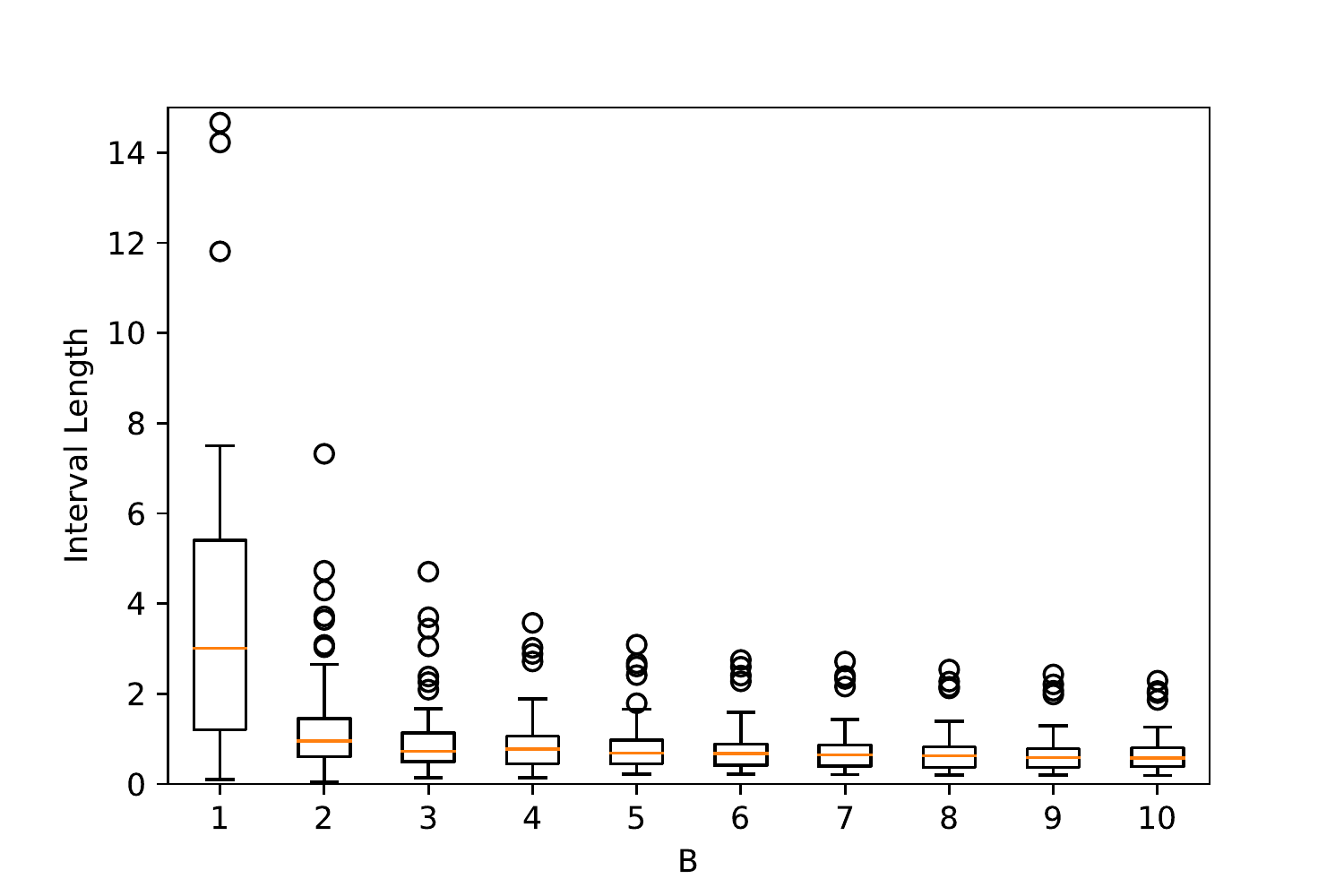} }
\subfigure
{\includegraphics[width=.49\columnwidth]{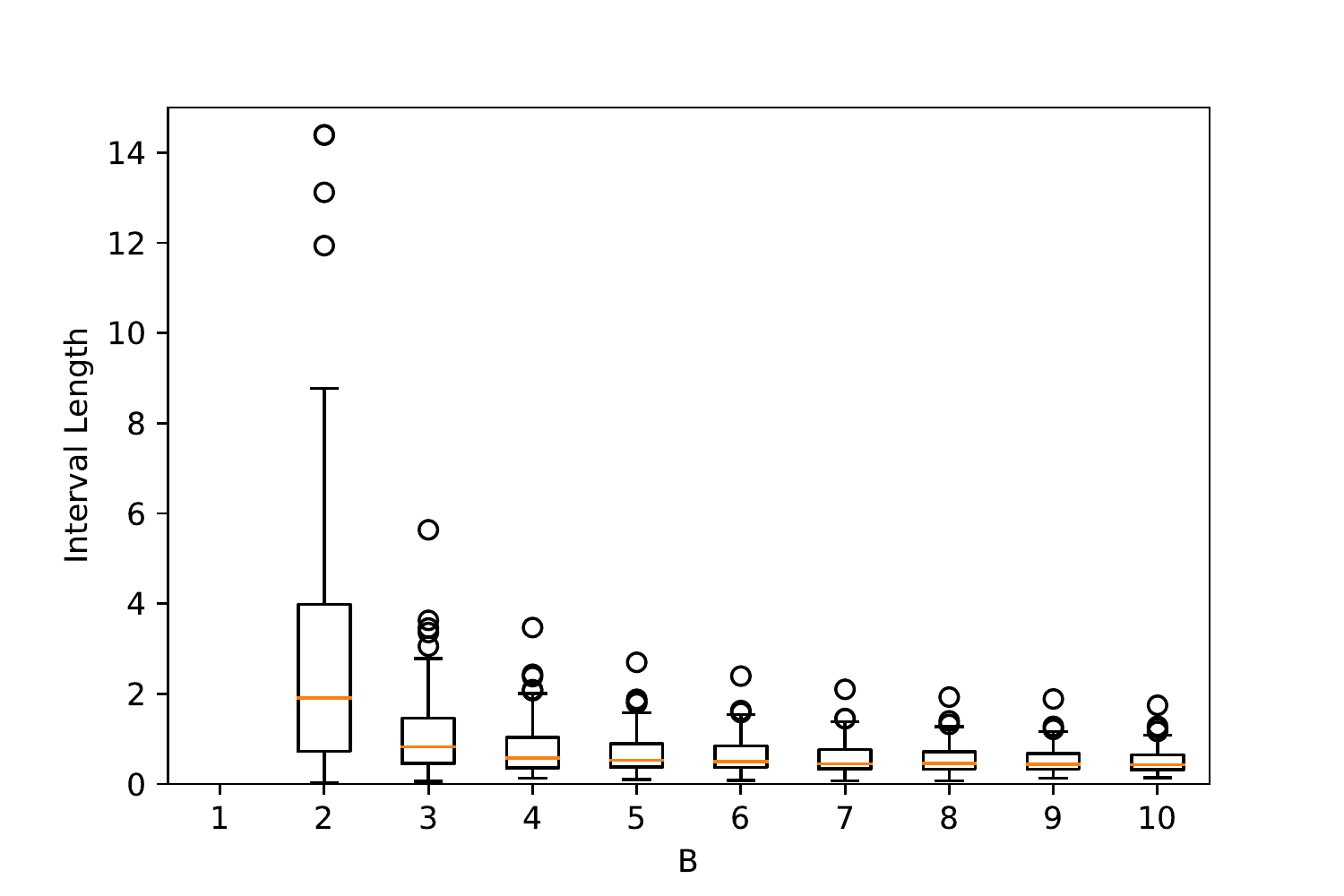}}
\caption{Box plots of $95\%$ confidence interval widths for deep ensemble prediction for the test points in the Boston housing data set, using Cheap Bootstrap centered at original estimate $\mathcal I_O$ (left graph) and Cheap Bootstrap centered at resample mean $\mathcal I_M$ (right graph).}
\label{fig:boston}
\end{center}
\vskip -0.2in
\end{figure*}

In Appendix \ref{sec:bagging}, we present an additional example on constructing confidence bounds for the optimality gaps of data-driven stochastic optimization problems. This problem, which possesses a nested sampling structure, is of interest to stochastic programming. There we would illustrate the performances of Cheap Bootstrap in constructing one-sided confidence bounds.

\section{Discussion}
Motivated by the computational demand in large-scale problems where repeated model refitting can be costly, in this paper we propose a Cheap Bootstrap method that can use very few resamples for bootstrap inference. A key element of our method is that it attains asymptotically exact coverage for any number of resamples, including as low as one. This is in contrast to conventional bootstrap approaches that require running many Monte Carlo replications. Our theory on the Cheap Bootstrap also differs from these approaches by exploiting the asymptotic independence between the original estimate and resample estimates, instead of a direct approximation of the sampling distribution via the reasample counterpart, the latter typically executed by running many Monte Carlo runs in order to obtain accurate summary statistics from the resamples. 


Besides the basic asymptotic coverage guarantee, we have studied several aspects of the Cheap Bootstrap. First is its higher-order coverage errors which match the conventional basic and percentile bootstraps. This error analysis is based on an Edgeworth expansion on a $t$-limit which, to our knowledge, has not been studied in the literature that has focused on normal limits. Second is the half-width behavior, where we show that the half widths of Cheap Bootstrap intervals are naturally wider for small number of resamples, and match the conventional approaches when resample size increases. Moreover, the decrease in the half width is sharp when the resample number is very small and flattens quickly, thus leading to a reasonable half-width performance with only few resamples. In addition, we have also investigated several generalizations of the Cheap Bootstrap, including its applications to nested sampling problems, subsampling procedures, and multivariate extensions. Our numerical results validate the performances of Cheap Bootstrap, especially its ability to conduct valid inference with an extremely small number of resamples.


This work is intended to lay the foundation of the proposed Cheap Bootstrap method. In principle, this method can be applied to many other problems than those shown in our numerical section. Essentially, for any problem where bootstrap interval is used and justified via the standard condition - conditional asymptotic normality of the resamples, one can adopt the Cheap Bootstrap in lieu of conventional bootstraps. Besides applying and testing the performances of Cheap Bootstrap across these other problems, we believe the following are also worth further investigation:
\\


\noindent\emph{Comparison with the infinitesimal jackknife and the jackknife: }While we have focused our comparisons primarily within the bootstrap family, the Cheap Bootstrap bears advantages when compared to non-bootstrap alternatives including the infinitesimal jackknife and the jackknife. These advantages inherit from the bootstrap approach in general, but strengthened further with the light resampling effort. The infinitesimal jackknife, or the delta method, uses a linear approximation to evaluate the standard error. This approach relies on influence function calculation that could be tedious analytically, and also computationally due to, for instance, the inversion of big matrices in $M$-estimation. The jackkknife, which relies on leave-one-out estimates, generally requires a number of model evaluation that is the same as the sample size. The Cheap Bootstrap thus provides potential significant computation savings compared to both methods. Note, however, that it is possible to combine the jackknife with batching to reduce computation load (see the next discussion point).

In problems facing nested sampling such as those in Section \ref{sec:double}, it is also possible to use non-bootstrap techniques such as the infinitesimal jackknife. In fact, the latter is particularly handy to derive for problems involving double resampling such as bagging predictors (e.g., \cite{wager2014confidence}). Nonetheless, it is open to our knowledge whether these alternatives can generally resolve the expensive nested sampling requirement faced by conventional bootstraps, especially in problems lacking unbiased estimators for the influence function. In these cases, the need to control the entangled noises coming from both data and computation could necessitate a large computation size, despite the apparent avoidance of nested sampling.
\\


\noindent\emph{Comparison with batching: } A technique prominently used in the simulation literature, but perhaps less common in statistics, is known as batch means or more generally standardized time series (\cite{glynn1990simulation,schmeiser1982batch,schruben1983confidence,glynn2018constructing}). This technique conducts inference by grouping data into batches and aggregating them via $t$-statistics. The batches can be disjoint or overlapping (\cite{meketon2007overlapping,song1995optimal}), similar to the blocks in subsampling (\cite{politis1994large}). While batch means uses $t$-statistics like our Cheap Bootstrap, it utilizes solely the original sample instead of resample, and is motivated to handle serially dependent observations in simulation outputs such as from Markov chain Monte Carlo (\cite{geyer1992practical,flegal2010batch,jones2006fixed}).

Batch means in principle can be used to construct confidence intervals with few model evaluations, by using a small number of batches or blocks. However, a potential downside of this approach is that in the disjoint-batches case, dividing a limited sample into batches will thin out the sample size per batch and deteriorate the accuracy of asymptotic approximation. In other words, there is a constrained tradeoff between the number of batches and sample size per batch that limits its accuracy. The Cheap Bootstrap, on the other hand, bypasses this constraint by allowing the use of any number of resamples. In fact, since the observations in different resamples overlap, the Cheap Bootstrap appears closer to overlapping batching. To this end, overlapping batching contains multiple tuning parameters to specify how the batches overlap, which could also affect the asymptotic distributions. In this regard, the Cheap Bootstrap appears easier to use as the only parameter needed is the resample budget, which is decided by the amount of computation resource.

One approach to obtain more in-depth theoretical comparisons between the Cheap Bootstrap and batching is to analyze the coefficients in the respective higher-order coverage expansions. This comprises an immediate future direction.
\\

\noindent\emph{High-dimensional problems: }As the main advantage of the Cheap Bootstrap relative to other approaches is its light computation, it would be revealing to analyze the error of the Cheap Bootstrap in high-dimensional problems where computation saving is of utmost importance. Regarding this, it is possible to obtain bounds for the coverage error of Cheap Bootstrap with explicit dependence on problem dimension, by following our analysis roadmap in Section \ref{sec:higher} but integrating with high-dimensional Berry-Esseen-type bounds for normal limits (\cite{chernozhukov2017central,fang2021high}).
\\

\noindent\emph{Higher-order coverage accuracy: }Although our investigation is orthogonal to the bootstrap literature on higher-order coverage error refinements, we speculate that our Cheap Bootstrap can be sharpened to exhibit second-order accurate intervals, by replacing the $t$-quantile with a bootstrap calibrated quantile. This approach is similar to the studentized bootstrap, but uses few outer resamples in a potential iterated bootstrap procedure. In other words, our computation effort could be larger than elementary bootstraps, as we need to run two layers of resampling, but less than a double bootstrap, as one of the layers consists of only few samples. As a related application, we may apply the Cheap Bootstrap to assess the error of the bootstrap itself, which is done conventionally via bootstrapping the bootstrap or alternately the jackknife-after-bootstrap (\cite{efron1992jackknife}).
\\

\noindent\emph{Serial dependence: }A common bootstrap approach to handle serial dependence is to use block sampling (e.g., \cite{davison1997bootstrap} \S8). The computation advantage of the Cheap Bootstrap likely continues to hold when using blocks. More specifically, instead of regenerating many series each concatenated from resampled blocks, the Cheap Bootstrap only regenerates few such series and aggregates them via a $t$-approximation. On the other hand, as mentioned above, batch means or subsampling also comprises viable approaches to handle serially dependent data, and the effectiveness of the Cheap Bootstrap compared with these approaches will need further investigation.




\section*{Acknowledgments}
I gratefully acknowledge support from the National Science Foundation under grants CAREER CMMI-1834710 and IIS-1849280.

\bibliography{bibliography}

\begin{thebibliography}{69}
\expandafter\ifx\csname natexlab\endcsname\relax\def\natexlab#1{#1}\fi
\expandafter\ifx\csname url\endcsname\relax
  \def\url#1{\texttt{#1}}\fi
\expandafter\ifx\csname urlprefix\endcsname\relax\def\urlprefix{URL: }\fi

\bibitem[{Alaa and Van Der~Schaar(2020)}]{alaa2020discriminative}
Alaa, A. and Van Der~Schaar, M. (2020) Discriminative jackknife: Quantifying
  uncertainty in deep learning via higher-order influence functions.
\newblock In \textit{International Conference on Machine Learning}, 165--174.
  PMLR.

\bibitem[{Barton(2012)}]{barton2012tutorial}
Barton, R.~R. (2012) Tutorial: Input uncertainty in output analysis.
\newblock In \textit{Proceedings of the 2012 Winter Simulation Conference
  (WSC)}, 1--12. IEEE.

\bibitem[{Barton et~al.(2014)Barton, Nelson and Xie}]{barton2014quantifying}
Barton, R.~R., Nelson, B.~L. and Xie, W. (2014) Quantifying input uncertainty
  via simulation confidence intervals.
\newblock \textit{INFORMS journal on computing}, \textbf{26}, 74--87.

\bibitem[{Beran(1987)}]{beran1987prepivoting}
Beran, R. (1987) Prepivoting to reduce level error of confidence sets.
\newblock \textit{Biometrika}, \textbf{74}, 457--468.

\bibitem[{Bickel et~al.(1997)Bickel, G{\"o}tze and van
  Zwet}]{bickel2012resampling}
Bickel, P.~J., G{\"o}tze, F. and van Zwet, W.~R. (1997) Resampling fewer than
  $n$ observations: {G}ains, losses, and remedies for losses.
\newblock \textit{Statistica Sinica}, \textbf{7}, 1--31.

\bibitem[{Birge and Louveaux(2011)}]{birge2011introduction}
Birge, J.~R. and Louveaux, F. (2011) \textit{Introduction to stochastic
  programming}.
\newblock Springer Science \& Business Media.

\bibitem[{Booth and Do(1993)}]{booth1993simple}
Booth, J.~G. and Do, K.-A. (1993) Simple and efficient methods for constructing
  bootstrap confidence intervals.
\newblock \textit{Computational Statistics}, \textbf{8}, 333--333.

\bibitem[{Booth and Hall(1994)}]{booth1994monte}
Booth, J.~G. and Hall, P. (1994) Monte {C}arlo approximation and the iterated
  bootstrap.
\newblock \textit{Biometrika}, \textbf{81}, 331--340.

\bibitem[{Breiman(1996)}]{breiman1996bagging}
Breiman, L. (1996) Bagging predictors.
\newblock \textit{Machine learning}, \textbf{24}, 123--140.

\bibitem[{B{\"u}hlmann and Yu(2002)}]{buhlmann2002analyzing}
B{\"u}hlmann, P. and Yu, B. (2002) Analyzing bagging.
\newblock \textit{The annals of Statistics}, \textbf{30}, 927--961.

\bibitem[{Cheng and Holland(2004)}]{cheng2004calculation}
Cheng, R.~C. and Holland, W. (2004) Calculation of confidence intervals for
  simulation output.
\newblock \textit{ACM Transactions on Modeling and Computer Simulation
  (TOMACS)}, \textbf{14}, 344--362.

\bibitem[{Chernozhukov et~al.(2017)Chernozhukov, Chetverikov and
  Kato}]{chernozhukov2017central}
Chernozhukov, V., Chetverikov, D. and Kato, K. (2017) Central limit theorems
  and bootstrap in high dimensions.
\newblock \textit{The Annals of Probability}, \textbf{45}, 2309--2352.

\bibitem[{Davison and Hinkley(1997)}]{davison1997bootstrap}
Davison, A.~C. and Hinkley, D.~V. (1997) \textit{Bootstrap Methods and Their
  Application}.
\newblock No.~1. Cambridge University Press.

\bibitem[{DiCiccio et~al.(1996)DiCiccio, Efron et~al.}]{diciccio1996bootstrap}
DiCiccio, T.~J., Efron, B. et~al. (1996) Bootstrap confidence intervals.
\newblock \textit{Statistical Science}, \textbf{11}, 189--228.

\bibitem[{DiCiccio et~al.(1992)DiCiccio, Martin and
  Young}]{diciccio1992analytical}
DiCiccio, T.~J., Martin, M.~A. and Young, G.~A. (1992) Analytical
  approximations for iterated bootstrap confidence intervals.
\newblock \textit{Statistics and Computing}, \textbf{2}, 161--171.

\bibitem[{Efron(1981)}]{efron1981nonparametric}
Efron, B. (1981) Nonparametric estimates of standard error: {T}he jackknife,
  the bootstrap and other methods.
\newblock \textit{Biometrika}, \textbf{68}, 589--599.

\bibitem[{Efron(1982)}]{efron1982jackknife}
--- (1982) \textit{The Jackknife, the Bootstrap and Other Resampling Plans}.
\newblock SIAM.

\bibitem[{Efron(1987)}]{efron1987better}
--- (1987) Better bootstrap confidence intervals.
\newblock \textit{Journal of the American statistical Association},
  \textbf{82}, 171--185.

\bibitem[{Efron(1992)}]{efron1992jackknife}
--- (1992) Jackknife-after-bootstrap standard errors and influence functions.
\newblock \textit{Journal of the Royal Statistical Society: Series B
  (Methodological)}, \textbf{54}, 83--111.

\bibitem[{Efron and Tibshirani(1994)}]{efron1994introduction}
Efron, B. and Tibshirani, R.~J. (1994) \textit{An Introduction to the
  Bootstrap}.
\newblock CRC Press.

\bibitem[{Fang and Koike(2021)}]{fang2021high}
Fang, X. and Koike, Y. (2021) High-dimensional central limit theorems by
  stein’s method.
\newblock \textit{The Annals of Applied Probability}, \textbf{31}, 1660--1686.

\bibitem[{Flegal et~al.(2010)Flegal, Jones et~al.}]{flegal2010batch}
Flegal, J.~M., Jones, G.~L. et~al. (2010) Batch means and spectral variance
  estimators in {M}arkov chain {M}onte {C}arlo.
\newblock \textit{The Annals of Statistics}, \textbf{38}, 1034--1070.

\bibitem[{Geyer(1992)}]{geyer1992practical}
Geyer, C.~J. (1992) Practical {M}arkov chain {M}onte {C}arlo.
\newblock \textit{Statistical Science}, \textbf{7}, 473--483.

\bibitem[{Giordano et~al.(2019)Giordano, Stephenson, Liu, Jordan and
  Broderick}]{giordano2019swiss}
Giordano, R., Stephenson, W., Liu, R., Jordan, M. and Broderick, T. (2019) A
  swiss army infinitesimal jackknife.
\newblock In \textit{The 22nd International Conference on Artificial
  Intelligence and Statistics}, 1139--1147. PMLR.

\bibitem[{Glasserman(2004)}]{glasserman2004monte}
Glasserman, P. (2004) \textit{Monte Carlo methods in financial engineering},
  vol.~53.
\newblock Springer.

\bibitem[{Glynn and Iglehart(1990)}]{glynn1990simulation}
Glynn, P.~W. and Iglehart, D.~L. (1990) Simulation output analysis using
  standardized time series.
\newblock \textit{Mathematics of Operations Research}, \textbf{15}, 1--16.

\bibitem[{Glynn and Lam(2018)}]{glynn2018constructing}
Glynn, P.~W. and Lam, H. (2018) Constructing simulation output intervals under
  input uncertainty via data sectioning.
\newblock In \textit{2018 Winter Simulation Conference (WSC)}, 1551--1562.
  IEEE.

\bibitem[{Hall(1986{\natexlab{a}})}]{hall1986bootstrap}
Hall, P. (1986{\natexlab{a}}) On the bootstrap and confidence intervals.
\newblock \textit{Annals of Statistics}, \textbf{14}, 1431--1452.

\bibitem[{Hall(1986{\natexlab{b}})}]{hall1986number}
--- (1986{\natexlab{b}}) On the number of bootstrap simulations required to
  construct a confidence interval.
\newblock \textit{The Annals of Statistics}, 1453--1462.

\bibitem[{Hall(2013)}]{hall2013bootstrap}
--- (2013) \textit{The Bootstrap and Edgeworth Expansion}.
\newblock Springer Science \& Business Media.

\bibitem[{Hall and Martin(1988)}]{hall1988bootstrap}
Hall, P. and Martin, M.~A. (1988) On bootstrap resampling and iteration.
\newblock \textit{Biometrika}, \textbf{75}, 661--671.

\bibitem[{He and Lam(2021)}]{he2021higher}
He, S. and Lam, H. (2021) Higher-order coverage errors of batching methods via
  edgeworth expansions on $ t $-statistics.
\newblock \textit{arXiv preprint arXiv:2111.06859}.

\bibitem[{Henderson(2003)}]{henderson2003input}
Henderson, S.~G. (2003) Input modeling: Input model uncertainty: Why do we care
  and what should we do about it?
\newblock In \textit{Proceedings of the 2003 Winter Simulation Conference}
  (eds. S.~Chick, P.~J. S\'{a}nchez, D.~Ferrin and D.~J. Morrice), 90--100.
  Piscataway, New Jersey: IEEE.

\bibitem[{Jones et~al.(2006)Jones, Haran, Caffo and Neath}]{jones2006fixed}
Jones, G.~L., Haran, M., Caffo, B.~S. and Neath, R. (2006) Fixed-width output
  analysis for {M}arkov chain {M}onte {C}arlo.
\newblock \textit{Journal of the American Statistical Association},
  \textbf{101}, 1537--1547.

\bibitem[{Kleiner et~al.(2012)Kleiner, Talwalkar, Sarkar and
  Jordan}]{10.5555/3042573.3042801}
Kleiner, A., Talwalkar, A., Sarkar, P. and Jordan, M.~I. (2012) The big data
  bootstrap.
\newblock In \textit{Proceedings of the 29th International Coference on
  International Conference on Machine Learning}, ICML'12, 1787–1794. Madison,
  WI, USA: Omnipress.

\bibitem[{Kleiner et~al.(2014)Kleiner, Talwalkar, Sarkar and
  Jordan}]{kleiner2014scalable}
--- (2014) A scalable bootstrap for massive data.
\newblock \textit{Journal of the Royal Statistical Society: Series B:
  Statistical Methodology}, \textbf{76}, 795--816.

\bibitem[{Kosorok(2007)}]{kosorok2007introduction}
Kosorok, M.~R. (2007) \textit{Introduction to Empirical Processes and
  Semiparametric Inference}.
\newblock Springer Science \& Business Media.

\bibitem[{Lakshminarayanan et~al.(2016)Lakshminarayanan, Pritzel and
  Blundell}]{lakshminarayanan2016simple}
Lakshminarayanan, B., Pritzel, A. and Blundell, C. (2016) Simple and scalable
  predictive uncertainty estimation using deep ensembles.
\newblock \textit{arXiv preprint arXiv:1612.01474}.

\bibitem[{Lam(2016)}]{lam2016advanced}
Lam, H. (2016) Advanced tutorial: Input uncertainty and robust analysis in
  stochastic simulation.
\newblock In \textit{Winter Simulation Conference (WSC), 2016}, 178--192. IEEE.

\bibitem[{Lam and Qian(2018{\natexlab{a}})}]{lam2018assessing}
Lam, H. and Qian, H. (2018{\natexlab{a}}) Assessing solution quality in
  stochastic optimization via bootstrap aggregating.
\newblock In \textit{2018 Winter Simulation Conference (WSC)}, 2061--2071.
  IEEE.

\bibitem[{Lam and Qian(2018{\natexlab{b}})}]{lam2018bounding}
--- (2018{\natexlab{b}}) Bounding optimality gap in stochastic optimization via
  bagging: Statistical efficiency and stability.
\newblock \textit{arXiv preprint arXiv:1810.02905}.

\bibitem[{Lam and Qian(2018{\natexlab{c}})}]{lam2018subsampling}
--- (2018{\natexlab{c}}) Subsampling variance for input uncertainty
  quantification.
\newblock In \textit{2018 Winter Simulation Conference (WSC)}, 1611--1622.
  IEEE.

\bibitem[{Lam and Qian(2021)}]{lam2021subsampling}
--- (2021) Subsampling to enhance efficiency in input uncertainty
  quantification.
\newblock \textit{Operations Research}.

\bibitem[{Law et~al.(1991)Law, Kelton and Kelton}]{law1991simulation}
Law, A.~M., Kelton, W.~D. and Kelton, W.~D. (1991) \textit{Simulation modeling
  and analysis}, vol.~2.
\newblock McGraw-Hill New York.

\bibitem[{Lee et~al.(2015)Lee, Purushwalkam, Cogswell, Crandall and
  Batra}]{lee2015m}
Lee, S., Purushwalkam, S., Cogswell, M., Crandall, D. and Batra, D. (2015) Why
  m heads are better than one: Training a diverse ensemble of deep networks.
\newblock \textit{arXiv preprint arXiv:1511.06314}.

\bibitem[{Lee and Young(1995)}]{lee1995asymptotic}
Lee, S.~M. and Young, G.~A. (1995) Asymptotic iterated bootstrap confidence
  intervals.
\newblock \textit{The Annals of Statistics}, \textbf{23}, 1301--1330.

\bibitem[{Lu et~al.(2020)Lu, Ie and Sha}]{lu2020uncertainty}
Lu, Z., Ie, E. and Sha, F. (2020) Uncertainty estimation with infinitesimal
  jackknife, its distribution and mean-field approximation.
\newblock \textit{arXiv preprint arXiv:2006.07584}.

\bibitem[{Mak et~al.(1999)Mak, Morton and Wood}]{mak1999monte}
Mak, W.-K., Morton, D.~P. and Wood, R.~K. (1999) Monte carlo bounding
  techniques for determining solution quality in stochastic programs.
\newblock \textit{Operations research letters}, \textbf{24}, 47--56.

\bibitem[{Meketon and Schmeiser(1984)}]{meketon2007overlapping}
Meketon, M.~S. and Schmeiser, B. (1984) Overlapping batch means: {S}omething
  for nothing?
\newblock In \textit{Proceedings of the Winter simulation Conference},
  227--230.

\bibitem[{Mentch and Hooker(2016)}]{mentch2016quantifying}
Mentch, L. and Hooker, G. (2016) Quantifying uncertainty in random forests via
  confidence intervals and hypothesis tests.
\newblock \textit{The Journal of Machine Learning Research}, \textbf{17},
  841--881.

\bibitem[{Nelson(2013)}]{nelson2013foundations}
Nelson, B. (2013) \textit{Foundations and Methods of Stochastic Simulation: A
  First Course}.
\newblock Springer Science \& Business Media.

\bibitem[{Politis and Romano(1994)}]{politis1994large}
Politis, D.~N. and Romano, J.~P. (1994) Large sample confidence regions based
  on subsamples under minimal assumptions.
\newblock \textit{The Annals of Statistics}, \textbf{22}, 2031--2050.

\bibitem[{Politis et~al.(1999)Politis, Romano and
  Wolf}]{politis1999subsampling}
Politis, D.~N., Romano, J.~P. and Wolf, M. (1999) \textit{Subsampling}.
\newblock Springer Science \& Business Media.

\bibitem[{Schenker(1985)}]{schenker1985qualms}
Schenker, N. (1985) Qualms about bootstrap confidence intervals.
\newblock \textit{Journal of the American Statistical Association},
  \textbf{80}, 360--361.

\bibitem[{Schmeiser(1982)}]{schmeiser1982batch}
Schmeiser, B. (1982) Batch size effects in the analysis of simulation output.
\newblock \textit{Operations Research}, \textbf{30}, 556--568.

\bibitem[{Schruben(1983)}]{schruben1983confidence}
Schruben, L. (1983) Confidence interval estimation using standardized time
  series.
\newblock \textit{Operations Research}, \textbf{31}, 1090--1108.

\bibitem[{Schulam and Saria(2019)}]{schulam2019can}
Schulam, P. and Saria, S. (2019) Can you trust this prediction? auditing
  pointwise reliability after learning.
\newblock In \textit{The 22nd International Conference on Artificial
  Intelligence and Statistics}, 1022--1031. PMLR.

\bibitem[{Sengupta et~al.(2016)Sengupta, Volgushev and
  Shao}]{sengupta2016subsampled}
Sengupta, S., Volgushev, S. and Shao, X. (2016) A subsampled double bootstrap
  for massive data.
\newblock \textit{Journal of the American Statistical Association},
  \textbf{111}, 1222--1232.

\bibitem[{Shao and Tu(2012)}]{shao2012jackknife}
Shao, J. and Tu, D. (2012) \textit{The jackknife and bootstrap}.
\newblock Springer Science \& Business Media.

\bibitem[{Shapiro et~al.(2021)Shapiro, Dentcheva and
  Ruszczynski}]{shapiro2021lectures}
Shapiro, A., Dentcheva, D. and Ruszczynski, A. (2021) \textit{Lectures on
  stochastic programming: modeling and theory}.
\newblock SIAM.

\bibitem[{Song and Nelson(2015)}]{song2015quickly}
Song, E. and Nelson, B.~L. (2015) Quickly assessing contributions to input
  uncertainty.
\newblock \textit{IIE Transactions}, \textbf{47}, 893--909.

\bibitem[{Song et~al.(2014)Song, Nelson and Pegden}]{song2014advanced}
Song, E., Nelson, B.~L. and Pegden, C.~D. (2014) Advanced tutorial: Input
  uncertainty quantification.
\newblock In \textit{Proceedings of the 2014 Winter Simulation Conference}
  (eds. A.~Tolk, S.~Diallo, I.~Ryzhov, L.~Yilmaz, S.~Buckley and J.~Miller),
  162--176. Piscataway, New Jersey: IEEE.

\bibitem[{Song and Schmeiser(1995)}]{song1995optimal}
Song, W.~T. and Schmeiser, B.~W. (1995) Optimal mean-squared-error batch sizes.
\newblock \textit{Management Science}, \textbf{41}, 110--123.

\bibitem[{Van~der Vaart(2000)}]{van2000asymptotic}
Van~der Vaart, A.~W. (2000) \textit{Asymptotic Statistics}, vol.~3.
\newblock Cambridge University Press.

\bibitem[{Van~der Vaart and Wellner(2013)}]{wellner2013weak}
Van~der Vaart, A.~W. and Wellner, J.~A. (2013) \textit{Weak Convergence and
  Empirical Processes: With Applications to Statistics}.
\newblock Springer Science \& Business Media.

\bibitem[{Wager et~al.(2014)Wager, Hastie and Efron}]{wager2014confidence}
Wager, S., Hastie, T. and Efron, B. (2014) Confidence intervals for random
  forests: The jackknife and the infinitesimal jackknife.
\newblock \textit{The Journal of Machine Learning Research}, \textbf{15},
  1625--1651.

\bibitem[{Wall et~al.(2001)Wall, Boen and Tweedie}]{wall2001effective}
Wall, M.~M., Boen, J. and Tweedie, R. (2001) An effective confidence interval
  for the mean with samples of size one and two.
\newblock \textit{The American Statistician}, \textbf{55}, 102--105.

\bibitem[{Xie et~al.(2014)Xie, Nelson and Barton}]{xie2014bayesian}
Xie, W., Nelson, B.~L. and Barton, R.~R. (2014) A bayesian framework for
  quantifying uncertainty in stochastic simulation.
\newblock \textit{Operations Research}, \textbf{62}, 1439--1452.

\bibitem[{Zouaoui and Wilson(2004)}]{zouaoui2004accounting}
Zouaoui, F. and Wilson, J.~R. (2004) Accounting for input-model and
  input-parameter uncertainties in simulation.
\newblock \textit{IIE Transactions}, \textbf{36}, 1135--1151.

\end{thebibliography}
\bibliographystyle{rss}


\newpage

\begin{appendix}

\section{Immediate Extensions from Section \ref{sec:theory}}\label{sec:extensions}

\subsection{Inference on Standard Error}\label{sec:SE}
We can conduct inference on the standard error of $\hat\psi_n$ using Cheap Bootstrap. Here by standard error we mean the standard deviation of $\hat\psi_n$, which is asymptotically $\sigma/\sqrt n$, and our inference target is $\sigma$. Using the same argument as in Section \ref{sec:exact}, we have that
\begin{equation}
\mathcal J=\left[\frac{\sqrt{Bn}S}{\sqrt{\chi^2_{1-\alpha/2,B}}},\frac{\sqrt{Bn}S}{\sqrt{\chi^2_{\alpha/2,B}}}\right]\label{CI SE}
\end{equation}
is an asymptotically exact $(1-\alpha)$-level confidence interval for $\sigma$, where $\chi^2_{\alpha/2,B}$ and $\chi^2_{1-\alpha/2,B}$ are the $\alpha/2$ and $(1-\alpha/2)$-quantiles of $\chi^2_B$. We summarize this as follows.
\begin{theorem}[Cheap Bootstrap interval for standard error]
Under Assumption \ref{bp}, we have, for any bootstrap replication size $B\geq1$, the interval $\mathcal J$ defined in \eqref{CI SE} is an asymptotically exact $(1-\alpha)$-level confidence interval for $\sigma$, i.e., $\mathbb P_n(\sigma\in\mathcal J)\to1-\alpha$ as the sample size $n\to\infty$, where $\mathbb P_n$ denotes the probability with respect to the data $X_1,\ldots,X_n$ and all randomness from the resampling.\label{thm:se}
\end{theorem}


\subsection{Multivariate Generalization}\label{sec:multi}
We present a multivariate generalization of Cheap Bootstrap. Consider now $\psi:=\psi(P)\in\mathbb R^d$. Our multivariate confidence region proceeds similarly as the univariate case, except we use Hotelling's $T^2$ instead of $t$. More specifically, we have point estimate $\hat\psi_n=\psi(\hat P_n)$, and we resample from $\{X_1,\ldots,X_n\}$ to obtain $\{X_1^{*b},\ldots,X_n^{*b}\}$ and evaluate the resample estimate $\psi_n^{*b}:=\psi(P_n^{*b})$. Our confidence region is
\begin{equation}
\mathcal R=\left\{\psi:(\hat\psi_n-\psi)^\top S^{-1}(\hat\psi_n-\psi)\leq T^2_{d,B,1-\alpha}\right\}\label{CS}
\end{equation}
where $S$ now denotes a $d\times d$ matrix given by
\begin{equation}
S=\frac{1}{B}\sum_{b=1}^B\left(\psi_n^{*b}-\hat\psi_n\right)\left(\psi_n^{*b}-\hat\psi_n\right)^\top\label{cov}
\end{equation}
and the critical value $T^2_{d,B,1-\alpha}$ is the $(1-\alpha)$-quantile of Hotelling's $T^2$ distribution with parameters $d$ and $B$.

We have the following asymptotic exact guarantee for $\mathcal R$. First we make the following multivariate analog of Assumption \ref{bp}:
\begin{assumption}
[Standard condition for multivariate bootstrap validity]
We have $\sqrt n(\hat\psi_n-\psi)\Rightarrow N(0,\Sigma)$ and $\sqrt n(\psi_n^*-\hat\psi_n)\Rightarrow N(0,\Sigma)$ conditional on the data $X_1,X_2,\ldots$ in probability  as $n\to\infty$, where $N(0,\Sigma)$ is a multivariate normal vector with mean $0\in\mathbb R^d$ and covariance matrix $\Sigma\in\mathbb R^{d\times d}$ that is positive definite.\label{bp multi}
\end{assumption}

Then we have:
\begin{theorem}[Asymptotic exactness of multivariate Cheap Bootstrap]
Under Assumption \ref{bp multi}, for any $B\geq d$, the region $\mathcal R$ in \eqref{CS} is an asymptotically exact $(1-\alpha)$-level confidence region for $\psi$, i.e.,
$\mathbb P_n(\psi\in\mathcal R)\to1-\alpha$ as $n\to\infty$, where $\mathbb P_n$ denotes the probability with respect to the data $X_1,\ldots,X_n$ and all randomness from the resampling.\label{main multi}
\end{theorem}

Note that in Theorem \ref{main multi} we require $B$ to be at least $d$, the dimension of $\psi$. In the univariate case this reduces to $B=1$.

Finally, the following is a multivariate analog of Proposition \ref{HD} that uses Hadamard differentiability with non-degenerate derivative to ensure Assumption \ref{bp multi}:
\begin{proposition}[Sufficient conditions for multivariate bootstrap validity]
Consider $\hat P_n$ and $P_n^*$ as random elements that take values in $\ell^\infty(\mathcal F)$, where $\mathcal F$ is a Donsker class with finite envelope. Suppose $\psi:\ell^\infty(\mathcal F)\to\mathbb R^d$ is Hadamard differentiable at $P$ where the derivative $\psi_P'$ satisfies that the covariance matrix of $\psi_P'(\mathbb G_P)$ is positive definite, for a tight Gaussian process $\mathbb G_P$ on $\ell^\infty(\mathcal F)$ with mean 0 and covariance $Cov(\mathbb G_P(f_1),\mathbb G_P(f_2))=Cov_P(f_1(X),f_2(X))$. Then Assumption \ref{bp multi} holds under i.i.d. data.\label{HD multi}
\end{proposition}

\section{Proofs for Section \ref{sec:theory} and Appendix \ref{sec:extensions}}\label{sec:main proof}

\begin{proof}[Proof of Proposition~\ref{HD}]
Apply Theorem \ref{thm:delta empirical} (in Appendix \ref{sec:sufficient proof}) with $d=1$ and note that $\psi_P'(\mathbb G_P)$ is normal.
\end{proof}

\begin{proof}[Proof of Theorem \ref{main Edgeworth}]
For convenience throughout the proof we use $C$ to denote a positive constant that is not necessarily the same every time it appears. We first define
$$\hat A(\mathbf x)=\frac{g(\mathbf x)-g(\overline{\mathbf X})}{h(\overline{\mathbf X})}$$
and consider $\hat A(\overline{\mathbf X}^*)$, where $\overline{\mathbf X}^*=(1/n)\sum_{i=1}^n\mathbf X_i^*$ is the mean of a resample $\{\mathbf X_1^*,\ldots,\mathbf X_n^*\}$. We can view $\hat A(\overline{\mathbf X}^*)$ as the resample counterpart of $A(\overline{\mathbf X})$.

Under the function-of-mean model described in the theorem, the pivotal statistic \eqref{pivotal basic} can be written as
$$T=\frac{g(\overline{\mathbf X})-g(\bm\mu)}{\sqrt{\frac{1}{B}\sum_{b=1}^B(g(\overline{\mathbf X}^{*b})-g(\overline{\mathbf X}))^2}}=\frac{A_s(\overline{\mathbf X})}{\sqrt{\frac{1}{B}\sum_{b=1}^B\hat A(\overline{\mathbf X}^{*b})^2}}$$
For a two-sided interval, coverage is defined by the event
$$\left|\frac{g(\overline{\mathbf X})-g(\bm\mu)}{\sqrt{\frac{1}{B}\sum_{b=1}^B(g(\overline{\mathbf X}^{*b})-g(\overline{\mathbf X}))^2}}\right|\leq t_{B,1-\alpha/2}$$
or equivalently
$$\left|\frac{\sqrt nA_s(\overline{\mathbf X})}{\sqrt{\frac{1}{B}\sum_{b=1}^B(\sqrt n\hat A(\overline{\mathbf X}^{*b}))^2}}\right|\leq t_{B,1-\alpha/2}$$
Define $Q^*$ as the conditional distribution of $\sqrt n\hat A(\overline{\mathbf X}^{*b})$ given the data $\mathcal X_n=\{\mathbf X_1,\ldots,\mathbf X_n\}$. Note that we can write the coverage probability as
\begin{equation}
E\left[\int\cdots\int_{\left|\frac{\sqrt nA_s(\overline{\mathbf X})}{\sqrt{\frac{1}{B}\sum_{b=1}^Bz_b^2}}\right|\leq t_{B,1-\alpha/2}}dQ^*(z_B)\cdots dQ^*(z_1)\right]\label{complete}
\end{equation}
where the expectation $E$ is taken with respect to the data $\mathcal X_n$.


Consider a positive number $\lambda\geq3/2$. We first show that, with probability $1-O(n^{-\lambda})$,
\begin{eqnarray}
&&\int\cdots\int_{\left|\frac{\sqrt nA_s(\overline{\mathbf X})}{\sqrt{\frac{1}{B}\sum_{b=1}^Bz_b^2}}\right|\leq t_{B,1-\alpha/2}}dQ^*(z_B)\cdots dQ^*(z_1)\nonumber\\
&&=\int\cdots\int_{\left|\frac{\sqrt nA_s(\overline{\mathbf X})}{\sqrt{\frac{1}{B}\sum_{b=1}^Bz_b^2}}\right|\leq t_{B,1-\alpha/2}}d\hat Q^*(z_B)\cdots d\hat Q^*(z_1)+R\label{replacement}
\end{eqnarray}
where
\begin{equation}
\hat Q^*(x)=\Phi(x)+\sum_{\substack{j=1,\ldots,\nu\\j\text{\ even}}} n^{-j/2}\hat p_j(x)\phi(x)\label{interim Edgeworth4}
\end{equation}
is a random signed measure (constructed from the random polynomial $\hat p_j$), and $R$ satisfies $|R|\leq Cn^{-(\nu+1)/2}$ for some constant $C$. Here, the polynomials $\hat p_j$ are the ones defined in Theorem \ref{Edgeworth bootstrap} (in Appendix \ref{sec:background Edgeworth}). To this end, denote $\mathcal E$ as the event that
\begin{equation}
\sup_{-\infty<x<\infty}\left|P(\sqrt n\hat A(\overline{\mathbf X}^*)\leq x|\mathcal X_n)-\Phi(x)-\sum_{j=1}^\nu n^{-j/2}\hat p_j(x)\phi(x)\right|\leq Cn^{-(\nu+1)/2}\label{main poly1}
\end{equation}
\begin{equation}
\max_{1\leq j\leq\nu}\sup_{-\infty<x<\infty}(1+|x|)^{-(3j-1)}|\hat p_j(x)|\leq C\label{poly1}
\end{equation}
and
\begin{equation}
\max_{1\leq j\leq\nu}\sup_{-\infty<x<\infty}(1+|x|)^{-(3j-1)}|\hat p_j'(x)|\leq C\label{poly derivative1}
\end{equation}
hold simultaneously. By Theorem \ref{Edgeworth bootstrap}, $\mathcal E$ occurs with probability $1-O(n^{-\lambda})$. Now, conditional on the data $\mathcal X_n$ and for any $z_1,\ldots,z_{B-1}\in\mathbb R$,
\begin{equation}
\int_{\left|\frac{\sqrt nA_s(\overline{\mathbf X})}{\sqrt{\frac{1}{B}\sum_{b=1}^{B}z_b^2}}\right|\leq t_{B,1-\alpha/2}}dQ^*(z_B)\label{interim Edgeworth1}
\end{equation}
is expressible as $1-Q^*(q)+Q^*(-q)$ for some $q\in\mathbb R$. Thus, using the oddness and evenness property of $p_j$ in Theorem \ref{Edgeworth original}, we have, under $\mathcal E$,
\begin{equation}
\sup_{-\infty<z_1,\ldots,z_{B-1}<\infty}\left|\int_{\left|\frac{\sqrt nA_s(\overline{\mathbf X})}{\sqrt{\frac{1}{B}\sum_{b=1}^{B}z_b^2}}\right|\leq t_{B,1-\alpha/2}}dQ^*(z_B)-\int_{\left|\frac{\sqrt nA_s(\overline{\mathbf X})}{\sqrt{\frac{1}{B}\sum_{b=1}^{B}z_b^2}}\right|\leq t_{B,1-\alpha/2}}d\hat Q^*(z_B)\right|\leq Cn^{-(\nu+1)/2}\label{interim Edgeworth2}
\end{equation}
for some $C>0$ by \eqref{main poly1}. Thus, integrating \eqref{interim Edgeworth1} with respect to $dQ^*(z_{B-1})\cdots dQ^*(z_1)$ and using \eqref{interim Edgeworth2}, we get
\begin{eqnarray}
&&\int\cdots\int_{\left|\frac{\sqrt nA_s(\overline{\mathbf X})}{\sqrt{\frac{1}{B}\sum_{b=1}^Bz_b^2}}\right|\leq t_{B,1-\alpha/2}}dQ^*(z_B)\cdots dQ^*(z_1)\nonumber\\
&=&\int\cdots\int_{\left|\frac{\sqrt nA_s(\overline{\mathbf X})}{\sqrt{\frac{1}{B}\sum_{b=1}^Bz_b^2}}\right|\leq t_{B,1-\alpha/2}}d\hat Q^*(z_B)dQ^*(z_{B-1})\cdots dQ^*(z_1)+R_B\label{interim Edgeworth}
\end{eqnarray}
where $R_B$ satisfies $|R_B|\leq Cn^{-(\nu+1)/2}$. Iterating \eqref{interim Edgeworth}, this time considering
$$\int_{\left|\frac{\sqrt nA_s(\overline{\mathbf X})}{\sqrt{\frac{1}{B}\sum_{b=1}^{B}z_b^2}}\right|\leq t_{B,1-\alpha/2}}dQ^*(z_{B-1})$$
and using Fubini's theorem with the signed measure $\hat Q^*$, we have
\begin{eqnarray*}
&&\int\cdots\int_{\left|\frac{\sqrt nA_s(\overline{\mathbf X})}{\sqrt{\frac{1}{B}\sum_{b=1}^Bz_b^2}}\right|\leq t_{B,1-\alpha/2}}d\hat Q^*(z_B)dQ^*(z_{B-1})\cdots dQ^*(z_1)\\
&=&\int\cdots\int_{\left|\frac{\sqrt nA_s(\overline{\mathbf X})}{\sqrt{\frac{1}{B}\sum_{b=1}^Bz_b^2}}\right|\leq t_{B,1-\alpha/2}}d\hat Q^*(z_B)d\hat Q^*(z_{B-1})dQ^*(z_{B-2})\cdots dQ^*(z_1)+R_{B-1}
\end{eqnarray*}
where $R_{B-1}$ satisfies
$$|R_{B-1}|\leq Cn^{-(\nu+1)/2}\int\cdots\int|(\hat Q^*)'(z_B)|dz_BdQ^*(z_{B-2})\cdots dQ^*(z_1)=Cn^{-(\nu+1)/2}\int|(\hat Q^*)'(z_B)|dz_B\leq Cn^{-(\nu+1)/2}$$
by \eqref{poly1} and \eqref{poly derivative1} (where the last $C$ is a different constant from the previous one). Continuing in this fashion, we get
\begin{eqnarray*}
&&\int\cdots\int_{\left|\frac{\sqrt nA_s(\overline{\mathbf X})}{\sqrt{\frac{1}{B}\sum_{b=1}^Bz_b^2}}\right|\leq t_{B,1-\alpha/2}}dQ^*(z_B)\cdots dQ^*(z_1)\nonumber\\
&=&\int\cdots\int_{\left|\frac{\sqrt nA_s(\overline{\mathbf X})}{\sqrt{\frac{1}{B}\sum_{b=1}^Bz_b^2}}\right|\leq t_{B,1-\alpha/2}}d\hat Q^*(z_B)\cdots d\hat Q^*(z_1)+R_B+R_{B-1}+\cdots+R_1
\end{eqnarray*}
where each $R_b$ satisfies $|R_b|\leq Cn^{-(n+1)/2}$. This gives \eqref{replacement}.

Now consider \eqref{complete}, which we can write as
\begin{eqnarray}
&&E\left[\int\cdots\int_{\left|\frac{\sqrt nA_s(\overline{\mathbf X})}{\sqrt{\frac{1}{B}\sum_{b=1}^Bz_b^2}}\right|\leq t_{B,1-\alpha/2}}dQ^*(z_B)\cdots dQ^*(z_1)\right]\nonumber\\
&=&E\left[\int\cdots\int_{\left|\frac{\sqrt nA_s(\overline{\mathbf X})}{\sqrt{\frac{1}{B}\sum_{b=1}^Bz_b^2}}\right|\leq t_{B,1-\alpha/2}}dQ^*(z_B)\cdots dQ^*(z_1);\mathcal E\right]{}\notag\\
&&{}+E\left[\int\cdots\int_{\left|\frac{\sqrt nA_s(\overline{\mathbf X})}{\sqrt{\frac{1}{B}\sum_{b=1}^Bz_b^2}}\right|\leq t_{B,1-\alpha/2}}dQ^*(z_B)\cdots dQ^*(z_1);\mathcal E^c\right]\notag\\
&=&E\left[\int\cdots\int_{\left|\frac{\sqrt nA_s(\overline{\mathbf X})}{\sqrt{\frac{1}{B}\sum_{b=1}^Bz_b^2}}\right|\leq t_{B,1-\alpha/2}}dQ^*(z_B)\cdots dQ^*(z_1);\mathcal E\right]+O(n^{-\lambda}){}\notag\\
&&{}\text{\ \ since $\int\cdots\int_{\left|\frac{\sqrt nA_s(\overline{\mathbf X})}{\sqrt{(1/B)\sum_{b=1}^Bz_b^2}}\right|\leq t_{B,1-\alpha/2}}dQ^*(z_B)\cdots dQ^*(z_1)\leq1$}\notag\\
&=&E\left[\int\cdots\int_{\left|\frac{\sqrt nA_s(\overline{\mathbf X})}{\sqrt{\frac{1}{B}\sum_{b=1}^Bz_b^2}}\right|\leq t_{B,1-\alpha/2}}d\hat Q^*(z_B)\cdots d\hat Q^*(z_1);\mathcal E\right]+O(n^{-(\nu+1)/2})+O(n^{-\lambda}){}\label{interim Edgeworth3}\\
&&{}\text{\ \ by \eqref{replacement}}\notag
\end{eqnarray}
Using \eqref{interim Edgeworth4}, we write the first term in \eqref{interim Edgeworth3} as
\begin{eqnarray}
&&E\left[\int\cdots\int_{\left|\frac{\sqrt nA_s(\overline{\mathbf X})}{\sqrt{\frac{1}{B}\sum_{b=1}^Bz_b^2}}\right|\leq t_{B,1-\alpha/2}}d\hat Q^*(z_B)\cdots d\hat Q^*(z_1);\mathcal E\right]\notag\\
&=&E\Bigg[\int\cdots\int_{\left|\frac{\sqrt nA_s(\overline{\mathbf X})}{\sqrt{\frac{1}{B}\sum_{b=1}^Bz_b^2}}\right|\leq t_{B,1-\alpha/2}}d\left(\Phi(z_B)+\sum_{\substack{j=1,\ldots,\nu\\j\text{\ even}}} n^{-j/2}\hat p_j(z_B)\phi(z_B)\right)\cdots{}\notag\\
&&{}d\left(\Phi(z_1)+\sum_{\substack{j=1,\ldots,\nu\\j\text{\ even}}} n^{-j/2}\hat p_j(z_1)\phi(z_1)\right);\mathcal E\Bigg]\notag\\
&=&E\left[\int\cdots\int_{\left|\frac{\sqrt nA_s(\overline{\mathbf X})}{\sqrt{\frac{1}{B}\sum_{b=1}^Bz_b^2}}\right|\leq t_{B,1-\alpha/2}}d\Phi(z_B)\cdots d\Phi(z_1);\mathcal E\right]{}\notag\\
&&{}+\frac{B}{n}E\left[\int\cdots\int_{\left|\frac{\sqrt nA_s(\overline{\mathbf X})}{\sqrt{\frac{1}{B}\sum_{b=1}^Bz_b^2}}\right|\leq t_{B,1-\alpha/2}}d(\hat p_2(z_B)\phi(z_B))d\Phi(z_{B-1})\cdots d\Phi(z_1);\mathcal E\right]{}\notag\\
&&{}+O\left(\frac{1}{n^2}\right)\label{interim Edgeworth5}
\end{eqnarray}
where the last equality follows by expanding out the product for $\nu\geq2$, using the symmetry among $z_1,\ldots,z_B$ to get the second term, and using \eqref{poly1} and \eqref{poly derivative1} to get the last remainder term.

Next, we can write
\begin{eqnarray}
&&E\left[\int\cdots\int_{\left|\frac{\sqrt nA_s(\overline{\mathbf X})}{\sqrt{\frac{1}{B}\sum_{b=1}^Bz_b^2}}\right|\leq t_{B,1-\alpha/2}}d(\hat p_2(z_B)\phi(z_B))d\Phi(z_{B-1})\cdots d\Phi(z_1);\mathcal E\right]\notag\\
&=&E\left[\int\cdots\int_{\left|\frac{\sqrt nA_s(\overline{\mathbf X})}{\sqrt{\frac{1}{B}\sum_{b=1}^Bz_b^2}}\right|\leq t_{B,1-\alpha/2}}d(p_2(z_B)\phi(z_B))d\Phi(z_{B-1})\cdots d\Phi(z_1);\mathcal E\right]{}\notag\\
&&{}+E\left[\int\cdots\int_{\left|\frac{\sqrt nA_s(\overline{\mathbf X})}{\sqrt{\frac{1}{B}\sum_{b=1}^Bz_b^2}}\right|\leq t_{B,1-\alpha/2}}d((\hat p_2(z_B)-p_2(z_B))\phi(z_B))d\Phi(z_{B-1})\cdots d\Phi(z_1);\mathcal E\right]\notag\\
&=&E\left[\int\cdots\int_{\left|\frac{\sqrt nA_s(\overline{\mathbf X})}{\sqrt{\frac{1}{B}\sum_{b=1}^Bz_b^2}}\right|\leq t_{B,1-\alpha/2}}d(p_2(z_B)\phi(z_B))d\Phi(z_{B-1})\cdots d\Phi(z_1);\mathcal E\right]{}\notag\\
&&{}+E\left[\text{poly}(\hat\mu_{m_1,\ldots,m_d},m_1+\cdots+m_d\leq4)-\text{poly}(\mu_{m_1,\ldots,m_d},m_1+\cdots+m_d\leq4);\mathcal E\right]\label{interim Edgeworth6}
\end{eqnarray}
where $\text{poly}(\hat\mu_{m_1,\ldots,m_d},m_1+\cdots+m_d\leq4)$ and $\text{poly}(\mu_{m_1,\ldots,m_d},m_1+\cdots+m_d\leq4)$ denote the same polynomial of bounded degree in $\hat\mu_{m_1,\ldots,m_d}$ or $\mu_{m_1,\ldots,m_d}$ for $m_1+\cdots+m_d\leq4$, where the coefficients of the polynomial consist of linear combinations of terms in the form
$$\int\cdots\int_{\left|\frac{\sqrt nA_s(\overline{\mathbf X})}{\sqrt{\frac{1}{B}\sum_{b=1}^Bz_b^2}}\right|\leq t_{B,1-\alpha/2}}d(z_B^k\phi(z_B))d\Phi(z_{B-1})\cdots d\Phi(z_1)$$
for some positive integer $k\leq5$, which is absolutely bounded by 
$$\int\cdots\int_{\mathbb R^B}|(z_B^k\phi(z_B))'|dz_B\Phi(z_{B-1})\cdots d\Phi(z_1)\leq C$$
for some constant $C>0$ independent of $k\leq5$. Now, with $E\|\mathbf X\|^l<\infty$ for sufficiently large $l$, uniform integrability gives the convergence on the power moments of sample moments, i.e.,
$$E\left|\hat\mu_{m_1,\ldots,m_d}^k-\mu_{m_1,\ldots,m_d}^k\right|=O\left(\frac{1}{\sqrt n}\right)$$
where $m_1+\cdots+m_d\leq4$ and any $k\leq5$. So we have \begin{eqnarray*}
&&\left|E\left[\text{poly}(\hat\mu_{m_1,\ldots,m_d},m_1+\cdots+m_d\leq4)-\text{poly}(\mu_{m_1,\ldots,m_d},m_1+\cdots+m_d\leq4);\mathcal E\right]\right|\\
&\leq&CE\left[\sum_{m_1+\cdots+m_d\leq4}\left|\hat\mu_{m_1,\ldots,m_d}^k-\mu_{m_1,\ldots,m_d}^k\right|\right]\\
&=&O\left(\frac{1}{\sqrt n}\right)
\end{eqnarray*}
where $C$ in the inequality could be different from the previous $C$. 

On the other hand,
\begin{eqnarray*}
&&E\left[\int\cdots\int_{\left|\frac{\sqrt nA_s(\overline{\mathbf X})}{\sqrt{\frac{1}{B}\sum_{b=1}^Bz_b^2}}\right|\leq t_{B,1-\alpha/2}}d(p_2(z_B)\phi(z_B))d\Phi(z_{B-1})\cdots d\Phi(z_1);\mathcal E\right]\\
&=&E\left[\int\cdots\int_{\left|\frac{\sqrt nA_s(\overline{\mathbf X})}{\sqrt{\frac{1}{B}\sum_{b=1}^Bz_b^2}}\right|\leq t_{B,1-\alpha/2}}d(p_2(z_B)\phi(z_B))d\Phi(z_{B-1})\cdots d\Phi(z_1)\right]+O(n^{-\lambda})
\end{eqnarray*}
since
\begin{eqnarray*}
&&E\left[\int\cdots\int_{\left|\frac{\sqrt nA_s(\overline{\mathbf X})}{\sqrt{\frac{1}{B}\sum_{b=1}^Bz_b^2}}\right|\leq t_{B,1-\alpha/2}}d(p_2(z_B)\phi(z_B))d\Phi(z_{B-1})\cdots d\Phi(z_1);\mathcal E^c\right]\\
&\leq&E\left[\int\cdots\int |(p_2(z_B)\phi(z_B))'|dz_Bd\Phi(z_{B-1})\cdots d\Phi(z_1);\mathcal E^c\right]=O(n^{-\lambda})
\end{eqnarray*}
So \eqref{interim Edgeworth6} becomes
\begin{equation}
E\left[\int\cdots\int_{\left|\frac{\sqrt nA_s(\overline{\mathbf X})}{\sqrt{\frac{1}{B}\sum_{b=1}^Bz_b^2}}\right|\leq t_{B,1-\alpha/2}}d(p_2(z_B)\phi(z_B))d\Phi(z_{B-1})\cdots d\Phi(z_1)\right]+O(n^{-\lambda})+O\left(\frac{1}{\sqrt n}\right)\label{interim Edgeworth7}
\end{equation}
Similarly,
\begin{eqnarray}
&&E\left[\int\cdots\int_{\left|\frac{\sqrt nA_s(\overline{\mathbf X})}{\sqrt{\frac{1}{B}\sum_{b=1}^Bz_b^2}}\right|\leq t_{B,1-\alpha/2}}d\Phi(z_B)\cdots d\Phi(z_1);\mathcal E\right]\notag\\
&=&E\left[\int\cdots\int_{\left|\frac{\sqrt nA_s(\overline{\mathbf X})}{\sqrt{\frac{1}{B}\sum_{b=1}^Bz_b^2}}\right|\leq t_{B,1-\alpha/2}}d\Phi(z_B)\cdots d\Phi(z_1)\right]+O(n^{-\lambda})\label{interim Edgeworth8}
\end{eqnarray}

Combining \eqref{interim Edgeworth3}, \eqref{interim Edgeworth5}, \eqref{interim Edgeworth7} and \eqref{interim Edgeworth8}, we can write \eqref{complete} as
\begin{eqnarray}
&&E\left[\int\cdots\int_{\left|\frac{\sqrt nA_s(\overline{\mathbf X})}{\sqrt{\frac{1}{B}\sum_{b=1}^Bz_b^2}}\right|\leq t_{B,1-\alpha/2}}d\Phi(z_B)\cdots d\Phi(z_1)\right]+O(n^{-\lambda}){}\notag\\
&&{}+\frac{B}{n}\left\{E\left[\int\cdots\int_{\left|\frac{\sqrt nA_s(\overline{\mathbf X})}{\sqrt{\frac{1}{B}\sum_{b=1}^Bz_b^2}}\right|\leq t_{B,1-\alpha/2}}d(p_2(z_B)\phi(z_B))d\Phi(z_{B-1})\cdots d\Phi(z_1)\right]+O(n^{-\lambda})+O\left(\frac{1}{\sqrt n}\right)\right\}{}\notag\\
&&{}+O\left(\frac{1}{n^2}\right)+O(n^{-(\nu+1)/2})+O(n^{-\lambda})\notag\\
&=&E\left[\int\cdots\int_{\left|\frac{\sqrt nA_s(\overline{\mathbf X})}{\sqrt{\frac{1}{B}\sum_{b=1}^Bz_b^2}}\right|\leq t_{B,1-\alpha/2}}d\Phi(z_B)\cdots d\Phi(z_1)\right]{}\notag\\
&&{}+\frac{B}{n}E\left[\int\cdots\int_{\left|\frac{\sqrt nA_s(\overline{\mathbf X})}{\sqrt{\frac{1}{B}\sum_{b=1}^Bz_b^2}}\right|\leq t_{B,1-\alpha/2}}d(p_2(z_B)\phi(z_B))d\Phi(z_{B-1})\cdots d\Phi(z_1)\right]{}\notag\\
&&{}+O(n^{-\lambda})+O\left(\frac{1}{n^{3/2}}\right)+O(n^{-(\nu+1)/2})\notag\\
&=&E\left[\int\cdots\int_{\left|\frac{\sqrt nA_s(\overline{\mathbf X})}{\sqrt{\frac{1}{B}\sum_{b=1}^Bz_b^2}}\right|\leq t_{B,1-\alpha/2}}d\Phi(z_B)\cdots d\Phi(z_1)\right]{}\notag\\
&&{}+\frac{B}{n}E\left[\int\cdots\int_{\left|\frac{\sqrt nA_s(\overline{\mathbf X})}{\sqrt{\frac{1}{B}\sum_{b=1}^Bz_b^2}}\right|\leq t_{B,1-\alpha/2}}d(p_2(z_B)\phi(z_B))d\Phi(z_{B-1})\cdots d\Phi(z_1)\right]+O\left(\frac{1}{n^{3/2}}\right)\label{interim Edgeworth9}
\end{eqnarray}
when $\lambda\geq3/2$ and $\nu\geq2$ so that $(\nu+1)/2\geq3/2$.

Lastly, using Theorem \ref{Edgeworth original}, and a similar argument as before with Fubini's theorem and the observation that
$$P\left(\left|\frac{\sqrt nA_s(\overline{\mathbf X})}{\sqrt{\frac{1}{B}\sum_{b=1}^Bz_b^2}}\right|\leq t_{B,1-\alpha/2}\Bigg|z_1,\ldots,z_B\right)$$
is expressible as $P(\sqrt nA_s(\overline{\mathbf X})\leq q)-P(\sqrt nA_s(\overline{\mathbf X})\leq-q)$ for some $q\in\mathbb R$, we have
\begin{eqnarray}
&&E\left[\int\cdots\int_{\left|\frac{\sqrt nA_s(\overline{\mathbf X})}{\sqrt{\frac{1}{B}\sum_{b=1}^Bz_b^2}}\right|\leq t_{B,1-\alpha/2}}d\Phi(z_B)\cdots d\Phi(z_1)\right]\notag\\
&=&\int\cdots\int_{\left|\frac{z_0}{\sqrt{\frac{1}{B}\sum_{b=1}^Bz_b^2}}\right|\leq t_{B,1-\alpha/2}}d\Phi(z_B)\cdots d\Phi(z_1)d\Phi(z_0){}\notag\\
&&{}+\frac{1}{n}\int\cdots\int_{\left|\frac{z_0}{\sqrt{\frac{1}{B}\sum_{b=1}^Bz_b^2}}\right|\leq t_{B,1-\alpha/2}}d\Phi(z_B)\cdots d\Phi(z_1)d(q_2(z_0)\phi(z_0))+o\left(\frac{1}{n}\right)\label{interim Edgeworth10}
\end{eqnarray}
when $\nu\geq2$. Similarly,
\begin{eqnarray}
&&E\left[\int\cdots\int_{\left|\frac{\sqrt nA_s(\overline{\mathbf X})}{\sqrt{\frac{1}{B}\sum_{b=1}^Bz_b^2}}\right|\leq t_{B,1-\alpha/2}}d(p_2(z_B)\phi(z_B))d\Phi(z_{B-1})\cdots d\Phi(z_1)\right]\notag\\
&=&\int\cdots\int_{\left|\frac{z_0}{\sqrt{\frac{1}{B}\sum_{b=1}^Bz_b^2}}\right|\leq t_{B,1-\alpha/2}}d(p_2(z_B)\phi(z_B))d\Phi(z_{B-1})\cdots d\Phi(z_0)+O\left(\frac{1}{n}\right)\label{interim Edgeworth11}
\end{eqnarray}
Thus, using \eqref{interim Edgeworth10} and \eqref{interim Edgeworth11}, we can write \eqref{interim Edgeworth9} as
\begin{eqnarray}
&&\int\cdots\int_{\left|\frac{z_0}{\sqrt{\frac{1}{B}\sum_{b=1}^Bz_b^2}}\right|\leq t_{B,1-\alpha/2}}d\Phi(z_B)\cdots d\Phi(z_1)d\Phi(z_0){}\notag\\
&&{}+\frac{1}{n}\Bigg\{B\int\cdots\int_{\left|\frac{z_0}{\sqrt{\frac{1}{B}\sum_{b=1}^Bz_b^2}}\right|\leq t_{B,1-\alpha/2}}d(p_2(z_B)\phi(z_B))d\Phi(z_{B-1})\cdots d\Phi(z_0){}\notag\\
&&{}+\int\cdots\int_{\left|\frac{z_0}{\sqrt{\frac{1}{B}\sum_{b=1}^Bz_b^2}}\right|\leq t_{B,1-\alpha/2}}d\Phi(z_B)\cdots d\Phi(z_1)d(q_2(z_0)\phi(z_0))\Bigg\}+o\left(\frac{1}{n}\right)\label{interim Edgeworth12}
\end{eqnarray}
which gives the first part of the theorem. Note that when $\nu\geq3$, the remainder term in \eqref{interim Edgeworth10} is refined to $o(1/n^{3/2})$ and, as a result, the remainder term in \eqref{interim Edgeworth12} is refined to $O(1/n^{3/2})$.

The second part of the theorem follows analogously by replacing the event $\left|\frac{g(\overline{\mathbf X})-g(\bm\mu)}{\sqrt{\frac{1}{B}\sum_{b=1}^B(g(\overline{\mathbf X}^{*b})-g(\overline{\mathbf X}))^2}}\right|\leq t_{B,1-\alpha/2}$
with $\frac{g(\overline{\mathbf X})-g(\bm\mu)}{\sqrt{\frac{1}{B}\sum_{b=1}^B(g(\overline{\mathbf X}^{*b})-g(\overline{\mathbf X}))^2}}\leq t_{B,1-\alpha}$
or $\frac{g(\overline{\mathbf X})-g(\bm\mu)}{\sqrt{\frac{1}{B}\sum_{b=1}^B(g(\overline{\mathbf X}^{*b})-g(\overline{\mathbf X}))^2}}\geq t_{B,1-\alpha}$. Consider now $\lambda>1/2$. Because of the aforementioned change of the considered event, now \eqref{interim Edgeworth4} is replaced by
$$\hat Q^*(x)=\Phi(x)+\sum_{j=1}^\nu n^{-j/2}\hat p_j(x)\phi(x)$$
and \eqref{interim Edgeworth5} becomes, in the upper interval case,
\begin{eqnarray*}
&=&E\left[\int\cdots\int_{\frac{\sqrt nA_s(\overline{\mathbf X})}{\sqrt{\frac{1}{B}\sum_{b=1}^Bz_b^2}}\leq t_{B,1-\alpha}}d\Phi(z_B)\cdots d\Phi(z_1);\mathcal E\right]{}\notag\\
&&{}+\frac{B}{\sqrt n}E\left[\int\cdots\int_{\frac{\sqrt nA_s(\overline{\mathbf X})}{\sqrt{\frac{1}{B}\sum_{b=1}^Bz_b^2}}\leq t_{B,1-\alpha}}d(\hat p_1(z_B)\phi(z_B))d\Phi(z_{B-1})\cdots d\Phi(z_1);\mathcal E\right]+O\left(\frac{1}{n}\right)
\end{eqnarray*}
for $\nu\geq1$ giving a modified \eqref{interim Edgeworth9} as
\begin{eqnarray}
&&E\left[\int\cdots\int_{\frac{\sqrt nA_s(\overline{\mathbf X})}{\sqrt{\frac{1}{B}\sum_{b=1}^Bz_b^2}}\leq t_{B,1-\alpha}}d\Phi(z_B)\cdots d\Phi(z_1)\right]{}\notag\\
&&{}+\frac{B}{\sqrt n}E\left[\int\cdots\int_{\frac{\sqrt nA_s(\overline{\mathbf X})}{\sqrt{\frac{1}{B}\sum_{b=1}^Bz_b^2}}\leq t_{B,1-\alpha}}d(p_1(z_B)\phi(z_B))d\Phi(z_{B-1})\cdots d\Phi(z_1)\right]+o\left(\frac{1}{\sqrt n}\right)\label{interim remark}
\end{eqnarray}
when $\lambda>1/2$ and $\nu\geq1$. Moreover, \eqref{interim Edgeworth10} becomes
\begin{eqnarray}
&&\int\cdots\int_{\frac{z_0}{\sqrt{\frac{1}{B}\sum_{b=1}^Bz_b^2}}\leq t_{B,1-\alpha}}d\Phi(z_B)\cdots d\Phi(z_1)d\Phi(z_0){}\notag\\
&&{}+\frac{1}{\sqrt n}\int\cdots\int_{\frac{z_0}{\sqrt{\frac{1}{B}\sum_{b=1}^Bz_b^2}}\leq t_{B,1-\alpha}}d\Phi(z_B)\cdots d\Phi(z_1)d(q_1(z_0)\phi(z_0))+o\left(\frac{1}{\sqrt n}\right)\label{interim remark1}
\end{eqnarray}
and \eqref{interim Edgeworth11} becomes
$$\int\cdots\int_{\frac{z_0}{\sqrt{\frac{1}{B}\sum_{b=1}^Bz_b^2}}\leq t_{B,1-\alpha}}d(p_1(z_B)\phi(z_B))d\Phi(z_{B-1})\cdots d\Phi(z_0)+O\left(\frac{1}{\sqrt n}\right)$$
giving rise to a modified \eqref{interim Edgeworth12} as
\begin{eqnarray}
&&\int\cdots\int_{\frac{z_0}{\sqrt{\frac{1}{B}\sum_{b=1}^Bz_b^2}}\leq t_{B,1-\alpha}}d\Phi(z_B)\cdots d\Phi(z_1)d\Phi(z_0){}\notag\\
&&{}+\frac{1}{\sqrt n}\Bigg\{B\int\cdots\int_{\frac{z_0}{\sqrt{\frac{1}{B}\sum_{b=1}^Bz_b^2}}\leq t_{B,1-\alpha}}d(p_1(z_B)\phi(z_B))d\Phi(z_{B-1})\cdots d\Phi(z_0){}\notag\\
&&{}+\int\cdots\int_{\frac{z_0}{\sqrt{\frac{1}{B}\sum_{b=1}^Bz_b^2}}\leq t_{B,1-\alpha}}d\Phi(z_B)\cdots d\Phi(z_1)d(q_1(z_0)\phi(z_0))\Bigg\}+o\left(\frac{1}{\sqrt n}\right)\label{interim remark2}
\end{eqnarray}
which gives the upper interval case in second part of the theorem. The lower interval case follows analogously. Moreover, when $\lambda\geq1$ and $\nu\geq2$, the remainder term in \eqref{interim remark} is refined to $O(1/n)$, the remainder term in \eqref{interim remark1} is refined to $O(1/n)$ and, as a result, the remainder term in \eqref{interim remark2} is refined to $O(1/n)$.
\end{proof}

\begin{proof}[Proof of Theorem \ref{thm:se}]
Using the same argument as in \eqref{proof weak convergence} in the proof of Theorem \ref{main}, we have
$\sqrt nS\Rightarrow\sigma\sqrt{\frac{\chi^2_B}{B}}$ and hence
$$\mathbb P_n\left(\sigma\sqrt{\frac{\chi^2_{\alpha/2,B}}{B}}\leq\sqrt nS\leq\sigma\sqrt{\frac{\chi^2_{1-\alpha/2,B}}{B}}\right)\to1-\alpha$$
The conclusion then follows.
\end{proof}

\begin{proof}[Proof of Theorem \ref{main multi}]
A straightforward modification of Proposition \ref{joint prop} from univariate to multivariate $\sqrt n(\hat\psi_n-\psi)$ and $\sqrt n(\psi_n^{*b}-\hat\psi_n)$ gives the asymptotic
\begin{equation}
\sqrt n(\hat\psi_n-\psi,\ \psi_n^{*1}-\hat\psi_n,\ldots,\ \psi_n^{*B}-\hat\psi_n)\Rightarrow(Z_0,Z_{1},\ldots,Z_{B})\label{joint convergence}
\end{equation}
where $Z_0,Z_{1},\ldots,Z_{B}\in\mathbb R^d$ are i.i.d. $N(0,\Sigma)$. By the continuous mapping theorem, we have
$$(\hat\psi_n-\psi)^\top S^{-1}(\hat\psi_n-\psi)=(\sqrt n(\hat\psi_n-\psi))^\top(n S)^{-1}(\sqrt n(\hat\psi_n-\psi))\Rightarrow T^2_{d,B}$$
where $T^2_{d,B}$ denotes Hotelling's $T^2$ variable with parameters $d$ and $B$. Hence
$$\mathbb P_n(\psi\in\mathcal R)=\mathbb P_n\left((\hat\psi_n-\psi)^\top S^{-1}(\hat\psi_n-\psi)\leq T^2_{d,B,1-\alpha}\right)\to1-\alpha$$
which concludes the theorem.
\end{proof}

\begin{proof}[Proof of Proposition \ref{HD multi}]
This is the same as Proposition \ref{HD} except we now consider general $d$ when using Theorem \ref{thm:delta empirical} in Appendix \ref{sec:sufficient proof}.
\end{proof}

\section{Further Details for Section \ref{sec:double}}
\subsection{Proofs for the Beginning of Section \ref{sec:double}}
\begin{proof}[Proof of Proposition \ref{technical moment}]
We focus on $\tau^2(\cdot)$ as the argument for $\kappa_3(\cdot)$ is the same. Suppose $\tau^2(\cdot)$ is Hadamard differentiable and satisfies the assumptions in Proposition \ref{HD}. Then, by Proposition \ref{HD}, Assumption \ref{bp} holds for $\tau^2(\cdot)$ and we have $\sqrt n(\tau^2(\hat P_n)-\tau^2(P))$ weakly converges to a tight random variable. Hence $\tau^2(\hat P_n)\stackrel{p}{\to}\tau^2(P)$ by the Slutsky theorem. Moreover, by Proposition \ref{joint prop}, we also know $\sqrt n(\tau^2(P_n^*)-\tau^2(\hat P_n))$ weakly converges to a tight variable, and hence $\tau^2(P_n^*)-\tau^2(\hat P_n)\stackrel{p}{\to}0$ by the Slutsky theorem again. Thus, $\tau^2(P_n^*)-\tau^2(P)=(\tau^2(P_n^*)-\tau^2(\hat P_n))+(\tau^2(\hat P_n)-\tau^2(P))\stackrel{p}{\to}0$ once again by the Slutsky theorem which concludes the proposition.
\end{proof}


\begin{proof}[Proof of Theorem \ref{joint1}]
We divide the proof into two steps:

\noindent\underline{Step 1.}
We first show the convergence of the following joint distribution
\begin{eqnarray}
&&\Bigg(\frac{\sqrt n(\hat\psi_n-\psi)}{\sigma},\frac{\sqrt n(\psi_n^{*1}-\hat\psi_n)}{\sigma},\ldots,\frac{\sqrt n(\psi_n^{*B}-\hat\psi_n)}{\sigma},{}\notag\\
&&\frac{\sqrt{ R_0}(\hat{\hat\psi}_{n,R_0}-\hat\psi_n)}{\tau},\frac{\sqrt{R}(\psi_{n,R}^{**1}-\psi_n^{*1})}{\tau},\ldots,\frac{\sqrt{R}(\psi_{n,R}^{**B}-\psi_n^{*B})}{\tau}\Bigg)\notag\\
&\Rightarrow&(Z_0,Z_1,\ldots,Z_B,W_0,W_1,\ldots,W_B)\label{joint elementary}
\end{eqnarray}
where $Z_0,Z_1,\ldots,Z_B,W_0,W_1,\ldots,W_B\stackrel{i.i.d.}{\sim}N(0,1)$.

For convenience, denote
\begin{align*}
Z_0^n&=\frac{\sqrt n(\hat\psi_n-\psi)}{\sigma}\\
Z_b^n&=\frac{\sqrt n(\psi_n^{*b}-\hat\psi_n)}{\sigma},\ b=1,\ldots,B\\
W_0^n&=\frac{\sqrt{R_0}(\hat{\hat\psi}_{n,R_0}-\hat\psi_n)}{\tau}\\
W_b^n&=\frac{\sqrt{R}(\psi_{n,R}^{**b}-\psi_n^{*b})}{\tau},\ b=1,\ldots,B
\end{align*}
Also, denote $\hat\tau^n=\tau(\hat P_n)$ and $\hat\kappa_3^n=\kappa_3(\hat P_n)$.

Then, we have, for any fixed real constants $z_b,w_b$ for $b=0,\ldots,B$,
\begin{eqnarray}
&&\left|P\left(Z_b^n\leq z_b,W_b^n\leq w_b,b=0,\ldots,B\right)-\prod_{b=0}^B\Phi(z_b)\Phi(w_b)\right|\notag\\
&=&\left|E\left[I(Z_0^n\leq z_0,W_0^n\leq w_0)P\left(Z_b^n\leq z_b,W_b^n\leq w_b,b=1,\ldots,B\Big|\hat P_n,\xi_{R_0}\right)\right]-\prod_{b=0}^B\Phi(z_b)\Phi(w_b)\right|\notag\\
&&\text{\ \ \ \ \ \ \  \ \ \ \ \ \ where $\xi_{R_0}$ refers to all the computation randomness in generating $\hat{\hat\psi}_{n,R_0}$,}\notag\\
&&\text{\ \ \ \ \ \ \  \ \ \ \ \ \ by noting that $Z_0^n$ and $W_0^n$ are determined solely by $\hat P_n$ and $\xi_{R_0}$}\notag\\
&=&\Bigg|E\left[I(Z_0^n\leq z_0,W_0^n\leq w_0)P\left(Z_b^n\leq z_b,W_b^n\leq w_b,b=1,\ldots,B\Big|\hat P_n,\xi_{R_0}\right)\right]{}\notag\\
&&{}-P(Z_0^n\leq z_0,W_0^n\leq w_0)\prod_{b=1}^B\Phi(z_b)\Phi(w_b)+P(Z_0^n\leq z_0,W_0^n\leq w_0)\prod_{b=1}^B\Phi(z_b)\Phi(w_b)-\prod_{b=0}^B\Phi(z_b)\Phi(w_b)\Bigg|\notag\\
&\leq&E\left[\left|P\left(Z_b^n\leq z_b,W_b^n\leq w_b,b=1,\ldots,B\Big|\hat P,\xi_{R_0}\right)-\prod_{b=1}^B\Phi(z_b)\Phi(w_b)\right|;Z_0^n\leq z_0,\ W_0^n\leq w_0\right]{}\notag\\
&&{}+\left|P(Z_0^n\leq z_0,\ W_0^n\leq w_0)-\Phi(z_0)\Phi(w_0)\right|\prod_{b=1}^B\Phi(z_b)\Phi(w_b){}\notag\\
&&{}\text{\ \ \ \ \ \ \  \ \ \ \ \ \ by the triangle inequality}\notag\\
&\leq&E\left|P\left(Z_b^n\leq z_b,W_b^n\leq w_b,b=1,\ldots,B\Big|\hat P,\xi_{R_0}\right)-\prod_{b=1}^B\Phi(z_b)\Phi(w_b)\right|{}\notag\\
&&{}+\left|P(Z_0^n\leq z_0,\ W_0^n\leq w_0)-\Phi(z_0)\Phi(w_0)\right|\prod_{b=1}^B\Phi(z_b)\Phi(w_b)\label{interim}
\end{eqnarray}
We consider the two terms in \eqref{interim} one by one, and we consider the second term first. By conditioning on $\hat P_n$ and telescoping, we have
\begin{eqnarray}
&&\left|P(Z_0^n\leq z_0,\ W_0^n\leq w_0)-\Phi(z_0)\Phi(w_0)\right|\notag\\
&\leq&E|P(W_0^n\leq w_0|\hat P_n)-\Phi(w_0)|+|P(Z_0^n\leq z_0)-\Phi(z_0)|\Phi(w_0)\label{interim1}
\end{eqnarray}
The first term in \eqref{interim1} can be bounded from above by
$$P(|\hat\tau^n-\tau|>\delta\text{\ or\ }\hat\kappa_3^n>\kappa_3+\delta)+E[|P(W_0^n\leq w_0|\hat P_n)-\Phi(w_0)|;|\hat\tau^n-\tau|\leq\delta,\hat\kappa_3^n\leq\kappa_3+\delta]$$
for some small $\delta>0$, where the second term can be written as
\begin{eqnarray*}
&&E\left[\left|P\left(\frac{W_0^n\tau}{\hat\tau^n}\leq\frac{w_0\tau}{\hat\tau^n}\Bigg|\hat P_n\right)-\Phi\left(\frac{w_0\tau}{\hat\tau^n}\right)\right|+\left|\Phi\left(\frac{w_0\tau}{\hat\tau^n}\right)-\Phi(w_0)\right|;|\hat\tau^n-\tau|\leq\delta,\hat\kappa_3^n\leq\kappa_3+\delta\right]\\
&\leq&E\left[\frac{C_1\hat\kappa_3^n}{(\hat\tau^n)^3\sqrt{ R_0}}+\left|\Phi\left(\frac{w_0\tau}{\hat\tau^n}\right)-\Phi(w_0)\right|;|\hat\tau^n-\tau|\leq\delta,\hat\kappa_3^n\leq\kappa_3+\delta\right]
\end{eqnarray*}
for some constant $C_1>0$ by the Berry-Esseen theorem, which is further bounded from above by
$$E\left[\frac{C_1(\kappa_3+\delta)}{(\tau-\delta)^3\sqrt{ R_0}}+\frac{C_2\delta}{\tau-\delta};|\hat\tau^n-\tau|\leq\delta,\hat\kappa_3^n\leq\kappa_3+\delta\right]\leq\frac{C_1(\kappa_3+\delta)}{(\tau-\delta)^3\sqrt{ R_0}}+\frac{C_2\delta}{\tau-\delta}$$
for some constant $C_2>0$, which follows from applying the mean value theorem to the function $\Phi(w_0\tau/\cdot)$ and noting that the function $x\phi(x)$ is bounded over $x\in\mathbb R$. Hence, the first term in \eqref{interim1} is bounded from above by
\begin{equation}
P(|\hat\tau^n-\tau|>\delta\text{\ or\ }\hat\kappa_3^n>\kappa_3+\delta)+\frac{C_1(\kappa_3+\delta)}{(\tau-\delta)^3\sqrt{R_0}}+\frac{C_2\delta}{\tau-\delta}\label{interim revise}
\end{equation}
Since $\hat\tau^n\stackrel{p}{\to}\tau$ and $\hat\kappa_3^n\stackrel{p}{\to}\kappa_3$
in Assumption \ref{assumption:sim}, given arbitrary $\epsilon>0$, we can choose a small enough $\delta>0$, a large enough $n$ and a large enough $R_0$ such that \eqref{interim revise} is bounded above by $\epsilon$. Thus, the first term in \eqref{interim1} converges to 0 as $n,R_0\to\infty$. The second term in \eqref{interim1} converges to 0 as $n\to\infty$ by Assumption \ref{bp}. We therefore have the second term in \eqref{interim} go to 0 as $n,R_0\to\infty$.

We handle the first term in \eqref{interim} with a similar argument. We have
\begin{eqnarray}
&&E\left|P\left(Z_b^n\leq z_b,W_b^n\leq w_b,b=1,\ldots,B\Big|\hat P_n,\xi_{R_0}\right)-\prod_{b=1}^B\Phi(z_b)\Phi(w_b)\right|\notag\\
&=&E\Bigg|E\left[I(Z_b^n\leq z_b,b=1,\ldots,B)P\left(W_b^n\leq w_b,b=1,\ldots,B\Big|P_n^{*b},b=1,\ldots,B,\hat P_n,\xi_{R_0}\right)\Big|\hat P_n,\xi_{R_0}\right]{}\notag\\
&&{}-P\left(Z_b^n\leq z_b,b=1,\ldots,B\Big|\hat P_n,\xi_{R_0}\right)\prod_{b=1}^B\Phi(w_b)+P\left(Z_b^n\leq z_b,b=1,\ldots,B\Big|\hat P_n,\xi_{R_0}\right)\prod_{b=1}^B\Phi(w_b){}\notag\\
&&{}-\prod_{b=1}^B\Phi(z_b)\Phi(w_b)\Bigg|\notag\\
&&\text{\ \ \ \ \ \ \  \ \ \ \ \ \ by noting that $Z_b^n,b=1,\ldots,B$ are determined solely by $P_n^{*b},b=1,\ldots,B$ and $\hat P_n$}\notag\\
&\leq&E\left|P\left(W_b^n\leq w_b,b=1,\ldots,B\Big|P_n^{*b},b=1,\ldots,B,\hat P_n,\xi_{R_0}\right)-\prod_{b=1}^B\Phi(w_b)\right|\notag\\
&&{}+E\left|P\left(Z_b^n\leq z_b,b=1,\ldots,B\Big|\hat P_n\right)-\prod_{b=1}^B\Phi(z_b)\right|\prod_{b=1}^B\Phi(w_b)\notag\\
&&\text{\ \ \ \ \ \ \  \ \ \ \ \ \ by the triangle and Jensen inequalities}\notag\\
&=&E\left|\prod_{b=1}^BP\left(W_b^n\leq w_b\Big|P_n^{*b}\right)-\prod_{b=1}^B\Phi(w_b)\right|+E\left|P\left(Z_b^n\leq z_b,b=1,\ldots,B\Big|\hat P_n\right)-\prod_{b=1}^B\Phi(z_b)\right|\prod_{b=1}^B\Phi(w_b)\notag\\
&&\text{\ \ \ \ \ \ \  \ \ \ \ \ \ by the conditional independence of $W_b^n,b=1,\ldots,B$ given $P_n^{*b},b=1,\ldots,B$ and that,}{}\notag\\
&&{}\text{\ \ \ \ \ \ \  \ \ \ \ \ \ given $P_n^{*b}$, $W_b^n$ is independent of $P_n^{*k}$ for $k\neq b$ and $\hat P_n$ and $\xi_{R_0}$}\notag\\
&=&E\Bigg|\prod_{b=1}^BP\left(W_b^n\leq w_b\Big|P_n^{*b}\right)-\Phi(w_1)\prod_{b=2}^BP\left(W_b^n\leq w_b\Big|P_n^{*b}\right)+\Phi(w_1)\prod_{b=2}^BP\left(W_b^n\leq w_b\Big|P_n^{*b}\right){}\notag\\
&&{}-\Phi(w_1)\Phi(w_2)\prod_{b=3}^BP\left(W_b^n\leq w_b\Big|P_n^{*b}\right)+\Phi(w_1)\Phi(w_2)\prod_{b=3}^BP\left(W_b^n\leq w_b\Big|P_n^{*b}\right)-\cdots-\prod_{b=1}^B\Phi(w_b)\Bigg|{}\notag\\
&&{}+E\left|P\left(Z_b^n\leq z_b,b=1,\ldots,B\Big|\hat P_n\right)-\prod_{b=1}^B\Phi(z_b)\right|\prod_{b=1}^B\Phi(w_b)\notag\\
&\leq&E\left[\sum_{b=1}^B\left|P\left(W_b^n\leq w_b\Big|P_n^{*b}\right)-\Phi(w_b)\right|\right]+E\left|P\left(Z_b^n\leq z_b,b=1,\ldots,B\Big|\hat P_n\right)-\prod_{b=1}^B\Phi(z_b)\right|\prod_{b=1}^B\Phi(w_b)\notag
\\
&&\text{\ \ \ \ \ \ \  \ \ \ \ \ \ by the triangle inequality}\label{interim2}
\end{eqnarray}
Note that in the first term in \eqref{interim2}, each
$$E\left|P\left(W_b^n\leq w_b\big|P^{*b}\right)-\Phi(w_b)\right|$$
converges to 0 as $n,R\to\infty$ by the same argument as for the first term in \eqref{interim1}, except that we use the bootstrapped moments $\tau(P_n^*)\stackrel{p}{\to}\tau$ and $\kappa_3(P_n^*)\stackrel{p}{\to}\kappa_3$ in Assumption \ref{assumption:sim} instead of $\hat\tau^n\stackrel{p}{\to}\tau$ and $\hat\kappa_3^n\stackrel{p}{\to}\kappa_3$. The second term in \eqref{interim2} also goes to 0 as $n\to\infty$ by Assumption \ref{bp} and the dominated convergence theorem.

Therefore, \eqref{interim} goes to 0. This proves \eqref{joint elementary}.
\\

\noindent\underline{Step 2.}
We consider the following decompositions
$$\hat{\hat\psi}_{n,R_0}-\psi=(\hat{\hat\psi}_{n,R_0}-\hat\psi_n)+(\hat\psi_n-\psi)$$
and
$$\psi_{n,R}^{**b}-\hat{\hat\psi}_{n,R_0}=(\psi^{**b}_{n,R}-\psi_n^{*b})+(\psi_n^{*b}-\hat\psi_n)+(\hat\psi_n-\hat{\hat\psi}_{n,R_0})$$
Then apply the continuous mapping theorem to get
\begin{eqnarray*}
&&\sqrt n\left(\hat{\hat\psi}_{n,R_0}-\psi,\ \psi_{n,R}^{**1}-\hat{\hat\psi}_{n,R_0},\ldots,\ \psi_{n,R}^{**B}-\hat{\hat\psi}_{n,R_0}\right)\\
&=&\Bigg(\sqrt{\frac{n}{R_0}}\sqrt{R_0}(\hat{\hat\psi}_{n,R_0}-\hat\psi_n)+\sqrt n(\hat\psi_n-\psi),
\sqrt{\frac{n}{R}}\sqrt{R}(\psi_{n,R}^{**1}-\psi_n^{*1})+\sqrt n(\psi_n^{*1}-\hat\psi_n)-\sqrt{\frac{n}{R_0}}\sqrt{ R_0}(\hat{\hat\psi}_{n,R_0}-\hat\psi_n),{}\\
&&{}\ldots,\sqrt{\frac{n}{R}}\sqrt{R}(\psi_{n,R}^{**B}-\psi_n^{*B})+\sqrt n(\psi_n^{*B}-\hat\psi_n)-\sqrt{\frac{n}{R_0}}\sqrt{ R_0}(\hat{\hat\psi}_{n,R_0}-\hat\psi_n)\Bigg)\\
&\Rightarrow&\left(\frac{\tau}{\sqrt{p_0}}W_0+\sigma Z_0,\ \frac{\tau}{\sqrt p}W_1+\sigma Z_1-\frac{\tau}{\sqrt{p_0}}W_0,\ldots,\ \frac{\tau}{\sqrt p}W_B+\sigma Z_B-\frac{\tau}{\sqrt{p_0}}W_0\right)
\end{eqnarray*}
This concludes the theorem.
\end{proof}

\subsection{Proofs and Additional Discussions for Section \ref{sec:cen}}
We first prove Theorem \ref{thm1}:

\begin{proof}[Proof of Theorem \ref{thm1}]
Consider the pivotal statistic $T_{O}=(\hat{\hat\psi}_{n,R_0}-\psi)/S_{O}$. We argue that $q_{O,1-\alpha/2}$ defined in \eqref{q center} satisfies
$$\liminf_{n\to\infty}P(\psi\in\mathcal I_{O})=\liminf_{n\to\infty}P\left(\left|T_{O}\right|\leq q_{O,1-\alpha/2}\right)\geq1-\alpha$$
which would conclude that $\mathcal I_{O}$ is an asymptotically valid $(1-\alpha)$-level confidence interval.

To this end, by Theorem \ref{joint1} and the continuous mapping theorem, as $n\to\infty$, $T_{O}$ converges weakly to
\begin{equation}
\frac{\sigma Z_0+\frac{\tau}{\sqrt{ p_0}}W_0}{\sqrt{\frac{1}{B}\sum_{b=1}^B\left(\sigma Z_b+\frac{\tau}{\sqrt{p}}W_b-\frac{\tau}{\sqrt{p_0}}W_0\right)^2}}\label{pivotal}
\end{equation}
where $Z_0,Z_1,\ldots,Z_B,W_0,W_1,\ldots,W_b\stackrel{i.i.d.}{\sim}N(0,1)$. A direct inspection on the homogenity of the expression reveals that \eqref{pivotal} only depends on $\sigma$ and $\tau$ through their ratio.  Multiplying by a factor $\sqrt{p_0}/\tau$ on both the numerator and denominator, we rewrite \eqref{pivotal} as
\begin{equation}
\frac{\theta Z_0+W_0}{\sqrt{\frac{1}{B}\sum_{b=1}^B\left(\theta Z_b+\rho W_b-W_0\right)^2}}\label{interim3}
\end{equation}
where $\theta=\sigma\sqrt{p_0}/\tau$ and $\rho=\sqrt{p_0/p}$ as defined before. We will see momentarily that \eqref{interim3} follows the same distribution as \eqref{limiting expression}. For now, recall the distribution function of \eqref{limiting expression} is $F(\cdot;\theta,\rho)$, which has the unknown $\theta$. We have
\begin{eqnarray*}
\liminf_{n\to\infty}P(\psi\in\mathcal I_{O})&=&\liminf_{n\to\infty}P\left(|T_{O}|\leq q_{O,1-\alpha/2}\right)=F\left(q_{O,1-\alpha/2};\theta,\rho\right)-F\left(-q_{O,1-\alpha/2};\theta,\rho\right){}\\
&=&{}2F\left(q_{O,1-\alpha/2};\theta,\rho\right)-1\geq2\min_{\theta\geq0} F\left(q_{O,1-\alpha/2};\theta,\rho\right)-1=1-\alpha
\end{eqnarray*}
where the third equality follows from the symmetry of \eqref{limiting expression} or \eqref{interim3}.


Finally, we see that \eqref{interim3} and \eqref{limiting expression} follow the same distribution, since by expanding the sum of squares in \eqref{interim3} we have
$$\frac{\theta Z_0+W_0}{\sqrt{\frac{1}{B}\sum_{b=1}^B\left(\theta Z_b+\rho W_b-W_0\right)^2}}\\
\stackrel{d}{=}\frac{\theta Z_0+W_0}{\sqrt{(\theta^2+\rho^2)\left(\frac{1}{B}\sum_{b=1}^B(X_b-\bar X)^2+\bar X^2\right)-2\sqrt{\theta^2+\rho^2}\bar XW_0+W_0^2}}$$
where we have written $\theta Z_b+\rho W_b=\sqrt{\theta^2+\rho^2}X_b$ with $X_1,\ldots,X_B\stackrel{i.i.d.}{\sim}N(0,1)$ which are independent of $Z_0,W_0$. Noting that $\sum_{b=1}^B(X_b-\bar X)^2\sim\chi^2_{B-1}$, $\bar X\sim N(0,1/B)$, which are independent by the property of standard normals, we get \eqref{limiting expression}.
\end{proof}

To understand how the additional intricacy from the computation noise affects the interval half-width, we consider the scenario when $B$ grows. The asymptotic distribution of $T_{O}$ given by \eqref{limiting expression} becomes
\begin{equation}
\frac{\theta V_1+V_2}{\sqrt{\theta^2+\rho^2+V_2^2}}\label{interim4}
\end{equation}
with $V_1,V_2\stackrel{i.i.d.}{\sim}N(0,1)$, since $Y/B\to1$ and $V_3^2/B,V_3V_2/\sqrt B\to0$ a.s. in \eqref{limiting expression}.
Correspondingly, $S_{O}$ defined in \eqref{S center} is distributed approximately as
\begin{equation}
\sqrt{\frac{1}{n}\left(\left(\sigma^2+\frac{\tau^2}{p}\right)\frac{Y+V_3^2}{B}-2\frac{\tau}{\sqrt{p_0}}\sqrt{\sigma^2+\frac{\tau^2}{p}}\frac{V_3V_2}{\sqrt B}+\frac{\tau^2}{p_0}V_2^2\right)}\label{interim5}
\end{equation}
which, when $B$ is large, becomes
$$\sqrt{\frac{\sigma^2+\frac{\tau^2}{p}+\frac{\tau^2}{p_0}V_2^2}{n}}$$
So the half-width of $\mathcal I_O$ behaves like
\begin{equation}
\tilde q_{O,1-\alpha/2}\sqrt{\frac{\sigma^2+\frac{\tau^2}{p}+\frac{\tau^2}{p_0}V_2^2}{n}}=\tilde q_{O,1-\alpha/2}\sqrt{\frac{\sigma^2}{n}+\frac{\tau^2}{R}+\frac{\tau^2}{R_0}V_2^2}\label{half-width center}
\end{equation}
where $\tilde q_{O,1-\alpha/2}$ is the maximum $(1-\alpha/2)$-quantile of \eqref{interim4} over all possible $\theta\geq0$. Now, supposing we know the values of $\sigma$ and $\tau$, the normality confidence interval obtained from extracting the first component in the limit in \eqref{CLT nested} is
$$\left[\hat{\hat\psi}_{n,R_0}-z_{1-\alpha/2}\sqrt{\frac{\sigma^2}{n}+\frac{\tau^2}{R_0}},\ \hat{\hat\psi}_{n,R_0}+z_{1-\alpha/2}\sqrt{\frac{\sigma^2}{n}+\frac{\tau^2}{R_0}}\right]$$
thus with a half-width
\begin{equation}
z_{1-\alpha/2}\sqrt{\frac{\sigma^2}{n}+\frac{\tau^2}{R_0}}\label{HW standard nested}
\end{equation}
The standard error in this half-width has the notable interpretation of being a combination of the data variability $\sigma^2/n$ and computation variability $\tau^2/R_0$.
Suppose in \eqref{half-width center} we use $R=R_0$, so that a resample estimate exhibits the same variability as the original estimate. Comparing \eqref{half-width center} and \eqref{HW standard nested}, we see that \eqref{half-width center} has an additional contribution coming from $V_2$. If $V_2$ is not present, then $S_{O}$ becomes $\sqrt{\sigma^2/n+\tau^2/R_0}$ and $\tilde q_{O,1-\alpha/2}$ becomes the quantile of $V_1$ which is a standard normal variable, since in this case the maximum quantile over all $\theta\geq0$ is approached by choosing $\theta\to\infty$. In other words, when $V_2$ is not present, we recover the normality interval half-width when $B$ increases. Thus, the added variability from $V_2$ can be viewed as a price we pay to handle the additional computation noise without knowledge on $\sigma$ and $\tau$.

\subsection{Proofs and Additional Discussions for Section \ref{sec:nc}}

We first prove Theorem \ref{thm2}:

\begin{proof}[Proof of Theorem \ref{thm2}]
The proof follows the roadmap of that of Theorem \ref{thm1}. Consider the pivotal statistic $T_{M}=(\hat{\hat\psi}_{n,R_0}-\psi)/S_{M}$. We first argue that $q_{M,1-\alpha/2}$ satisfies
$$\liminf_{n\to\infty}P(\psi\in\mathcal I_{M})=\liminf_{n\to\infty}P\left(\left|T_{M}\right|\leq q_{M,1-\alpha/2}\right)\geq1-\alpha$$
which would conclude that $\mathcal I_{M}$ is an asymptotically valid $(1-\alpha)$-level confidence interval.

By Theorem \ref{joint1} and the continuous mapping theorem, as $n\to\infty$, $T_{M}$ converges weakly to
\begin{equation}
\frac{\sigma Z_0+\frac{\tau}{\sqrt{ p_0}}W_0}{\sqrt{\frac{1}{B-1}\sum_{b=1}^B\left(\left(\sigma Z_b+\frac{\tau}{\sqrt{p}}W_b\right)-\left(\sigma \bar Z+\frac{\tau}{\sqrt{p}}\bar W\right)\right)^2}}\stackrel{d}{=}
\frac{\sigma Z_0+\frac{\tau}{\sqrt{ p_0}}W_0}{\sqrt{\left(\sigma^2+\frac{\tau^2}{p}\right)\frac{Y}{B-1}}}\stackrel{d}{=}\sqrt{\frac{\sigma^2+\frac{\tau^2}{p_0}}{\sigma^2+\frac{\tau^2}{p}}}t_{B-1}\label{nc expression replicate}
\end{equation}
where $Z_0,Z_1,\ldots,Z_B,W_0,W_1,\ldots,W_B\stackrel{i.i.d.}{\sim}N(0,1)$, $\bar Z=(1/B)\sum_{b=1}^BZ_b$, $\bar W=(1/B)\sum_{b=1}^BW_b$, and $Y\sim\chi^2_{B-1}$ which is independent of $Z_0$ and $W_0$. The equalities in distribution use the standard properties of normals to obtain $\chi^2_{B-1}$ and $t_{B-1}$ distributions. By multiplying both the numerator and denominator of \eqref{nc expression replicate} by $\sqrt{p_0}/\tau$, we get
\begin{equation}
\sqrt{\frac{\theta^2+1}{\theta^2+\rho^2}}t_{B-1}\label{nc expression3 replicate}
\end{equation}
where $\theta=\sigma\sqrt{p_0}/\tau$ and $\rho=\sqrt{p_0/p}$ as defined before. Denote the distribution of \eqref{nc expression3 replicate} as $\tilde F(\cdot;\theta,\rho)$. We have
$$\min_{\theta\geq0}\tilde F(q;\theta,\rho)=P\left(\max\{\rho^{-1},1\}t_{B-1}\leq q\right)$$
for any $q\geq0$, by noting that the minimum is approached by setting $\theta\to\infty$ when $\rho\geq1$ and attained at $\theta=0$ when $\rho<1$. Thus setting $q_{M,1-\alpha/2}=\max\{\rho^{-1},1\}t_{B-1,1-\alpha/2}$ gives
\begin{eqnarray*}
\liminf_{n\to\infty}P(\psi\in\mathcal I_{M})&=&\liminf_{n\to\infty}P\left(|T_{M}|\leq q_{M,1-\alpha/2}\right)=\tilde F\left(q_{M,1-\alpha/2};\theta,\rho\right)-\tilde F\left(-q_{M,1-\alpha/2};\theta,\rho\right){}\\
&=&{}2\tilde F\left(q_{M,1-\alpha/2};\theta,\rho\right)-1\geq2\min_{\theta\geq0}\tilde F\left(q_{M,1-\alpha/2};\theta,\rho\right)-1=1-\alpha
\end{eqnarray*}
where the third equality follows from the symmetry of \eqref{nc expression3 replicate}.

Finally, note that when $\rho=1$ \eqref{nc expression3 replicate} becomes $t_{B-1}$ regardless of the value of $\theta$. Thus in this case $T_{M}$ is asymptotically $t_{B-1}$, and asymptotic exactness of $\mathcal I_M$ holds.

\end{proof}

Note that, in contrast to \eqref{pivotal} where $W_0$ appears both in the numerator and denominator, in \eqref{nc expression replicate} $W_0$ only appears in the former and thus the numerator and denominator there are independent. This simplifies the asymptotic distribution to be used in $\mathcal I_M$.

We discuss the half-width efficiency of $\mathcal I_M$ and contrast with $\mathcal I_O$. First, unlike $\mathcal I_O$, note that it is possible to have asymptotic exactness for $\mathcal I_{M}$ as $n\to\infty$, in particular when we use the same computation size in the resample estimate and the original point estimate, which is a natural configuration. Moreover, from \eqref{q nc} and \eqref{nc expression replicate}, we see that as $n$ increases, the half-width of $\mathcal I_{M}$ behaves approximately as
$$\max\{\rho^{-1},1\}t_{B-1,1-\alpha/2}\sqrt{\left(\frac{\sigma^2}{n}+\frac{\tau^2}{R}\right)\frac{\chi^2_{B-1}}{B-1}}$$
When $B$ increases, this becomes
$$\max\{\rho^{-1},1\}z_{1-\alpha/2}\sqrt{\frac{\sigma^2}{n}+\frac{\tau^2}{R}}$$
so that when $\rho=1$ we get
$$z_{1-\alpha/2}\sqrt{\frac{\sigma^2}{n}+\frac{\tau^2}{R}}$$
which is the half-width of the normality confidence interval (when $R$ is set to equal $R_0$). This conformance shows the superiority of $\mathcal I_{M}$ over $\mathcal I_{O}$ when $B$ is large. Nonetheless, when $B$ is small or when the computation sizes of the resample and original estimates are different, $\mathcal I_{O}$ could possibly outperform $\mathcal I_{M}$.

Lastly, both $\mathcal I_O$ and $\mathcal I_M$ have natural one-sided analogs, where we replace $q_{O,1-\alpha/2}$ and $q_{M,1-\alpha/2}$ by $q_{O,1-\alpha}$ and $q_{M,1-\alpha}$ in one of the interval limits (and with the other side unbounded).

\section{Proofs for Section \ref{sec:subsampling}}
\begin{proof}[Proof of Theorem~\ref{main variant}]
In each of the three subsampling variants, we show that an analog of Assumption \ref{bp} holds and hence we can use the same roadmap as the proofs of Proposition \ref{joint prop} and Theorem \ref{main} to conclude our result. Note that, under the assumptions in Proposition \ref{HD}, we have immediately that $\sqrt n(\hat\psi_n-\psi)\Rightarrow N(0,\sigma^2)$ for some $\sigma^2>0$ by the functional delta method (Theorem \ref{EP delta} in Appendix \ref{sec:sufficient proof}). Now, under the additional assumption that $\mathcal F_\delta$ is measurable for every $\delta>0$, we have the following:
\\

\noindent\underline{Cheap $m$-out-of-$n$ Bootstrap:} Denote $\psi_s^{*}$ as a subsample estimate. We invoke Theorem \ref{vary size} (in Appendix \ref{sec:subsampling proofs}) to conclude that $\sqrt s(\psi_{s}^{*}-\hat\psi_n)\Rightarrow N(0,\sigma^2)$ given $X_1,X_2,\ldots$ in probability as $n\to\infty$, for any $s$ dependent on $n$ such that $s\leq n$ and $s\to\infty$. Following the same argument in the proof of Proposition \ref{joint prop}, except we replace $\sqrt n(\psi_n^{*b}-\hat\psi_n)$ by $\sqrt s(\psi_{s}^{*b}-\hat\psi_n)$, we obtain
$$\left(\sqrt n(\hat\psi_n-\psi),\sqrt s(\psi_{s}^{*1}-\hat\psi_n),\ldots,\sqrt s(\psi_{s}^{*B}-\hat\psi_n)\right)\Rightarrow(\sigma Z_0,\sigma Z_1,\ldots,\sigma Z_B)$$
where $Z_0,Z_1,\ldots,Z_B$ are i.i.d. $N(0,1)$. Therefore,
$$\frac{\hat\psi_n-\psi}{\sqrt{\frac{s}{n}}S}\Rightarrow\frac{Z_0}{\sqrt{\frac{1}{B}\sum_{b=1}^BZ_b^2}}$$
by the continuous mapping theorem. Then, following the proof of Theorem \ref{main} to note that the right hand side above is a $t_B$-variable, we get
$$P\left(-t_{1-\alpha/2,B}\leq\frac{\hat\psi_n-\psi}{\sqrt{\frac{s}{n}}S}\leq t_{1-\alpha/2,B}\right)\to1-\alpha$$
from which we conclude the result.
\\

\noindent\underline{Cheap Bag of Little Bootstraps:} Recall $\psi_s^*$ is the fixed subsample estimate, and denote $\psi_n^{**}$ as a second-layer resample estimate. Note that the subsample encoded by $P_s^*$, which is obtained by sampling without replacement from the data encoded by $\hat P_n$, is distributed i.i.d. from $P$. Thus we can involve Theorem \ref{vary size} to conclude $\sqrt n(\psi_{n}^{**}-\psi_{s}^*)\Rightarrow N(0,\sigma^2)$ in probability, given $X_1,X_2,\ldots,X_n$ and the subsampling randomness, as $n,s\to\infty$ for any $s$ dependent on $n$ such that $s\leq n$. The rest is identical to the proofs of Proposition \ref{joint prop} and Theorem \ref{main}, except we replace $\psi_n^{*b}-\hat\psi_n$ by $\psi_{n}^{**b}-\psi_{s}^*$ and, instead of conditioning on $\hat P_n$ in the series of inequalities in the proof of Proposition \ref{joint prop}, we condition on both $\hat P_n$ and the randomness in the subsampling that obtains $\psi_s^*$.
\\

\noindent\underline{Cheap Subsampled Double Bootstrap:} Denote $\psi_s^*$ as a first-layer subsample estimate and $\psi_n^{**}$ as the derived second-layer resample estimate. We invoke Theorem \ref{thm SDB} (in Appendix \ref{sec:subsampling proofs}) to conclude that $\sqrt n(\psi_{n}^{**}-\psi_{s}^*)\Rightarrow N(0,\sigma^2)$ given $X_1,X_2,\ldots$ in probability as $n\to\infty$, for any $s$ dependent on $n$ such that $s\leq n$ and $s\to\infty$. The rest is identical to the proofs of Proposition \ref{joint prop} and Theorem \ref{main}, except we replace $\psi_n^{*b}-\hat\psi_n$ by $\psi_{n}^{**b}-\psi_{s}^{*b}$.
\end{proof}

\section{Additional Numerical Results}\label{sec:add numerics}
\subsection{Logistic Regression}\label{sec:logistic}
We present another example on a logistic regression model $Y\sim Bernoulli(p)$ where $p=1/(1+\exp(-(\beta_1X_{1}+\cdots+\beta_dX_{d})))$. We set $d=10$ and use data $(X_{1,i},\ldots,X_{d,i},Y_i)$ of size $10^5$ to fit the model and estimate the coefficients $\beta_j$'s. The ground truth is set as $X_j\sim t_3$ and the coefficients $(\beta_1,\ldots,\beta_{10})=(1.9, 1.7, 1.3, 1.8, 1.1, 1.2, 1.9, 2.2, 1.5, 2.0)$, a set of numbers arbitrarily chosen from the interval $[1,3]$. Our setup is similar to the linear regression example in Section \ref{sec:linear}, where we test all methods to compute $95\%$ confidence intervals on the first coefficient $\beta_1$, and use subsample size $n^{0.6}=1000$ for $m$-out-of-$n$ Bootstrap, Bag of Little Bootstraps and Subsampled Double Bootstrap. We again use $50$ resamples in total for each method to depict the trend.

Figures \ref{fig:cov log}, \ref{fig:len log} and \ref{fig:sd log} show the coverage probabilities, mean interval widths, and standard deviations of interval widths respectively for Standard and Cheap bootstrap methods. The comparisons are largely similar to the linear regression example. The Cheap bootstrap methods all attain close to the target $95\%$ coverage at $B=1$, with $96\%$,  $96\%$, $97\%$ and $97\%$ for Cheap Bootstrap, $m$-out-of-$n$, Bag of Little Bootstrap and Subsampled Double Bootstrap respectively. On the other hand, the Standard bootstrap methods all fail at $B=1$ and require much larger $B$ to approach the target coverage. Both the means and standard deviations of interval widths for Cheap bootstrap methods decrease sharply from $B=1$ (e.g., mean $0.38$ and standard deviation $0.28$ for Cheap Bootstrap) to $2$ (mean $0.14$ and standard deviation $0.07$), and continue to drop further at $B=3$ (mean $0.11$ and standard deviation $0.04$) and beyond at a continuously slower rate. On the other hand, the mean interval widths of Standard bootstrap methods are initially small and exhibit increasing trends, whereas the standard deviations appear roughly constant against $B$. Thus, similar to the linear regression example, here Cheap bootstrap methods again consistently attain accurate coverage regardless of $B$ and their interval widths drop fast to levels comparable to large $B$. On the other hand, Standard methods under-cover when $B$ is small and converge to the nominal level at much slower rates.

\begin{figure*}[tb]
\vskip 0.2in
\begin{center}
\subfigure
{\includegraphics[width=.24\columnwidth,height=.15\textheight]{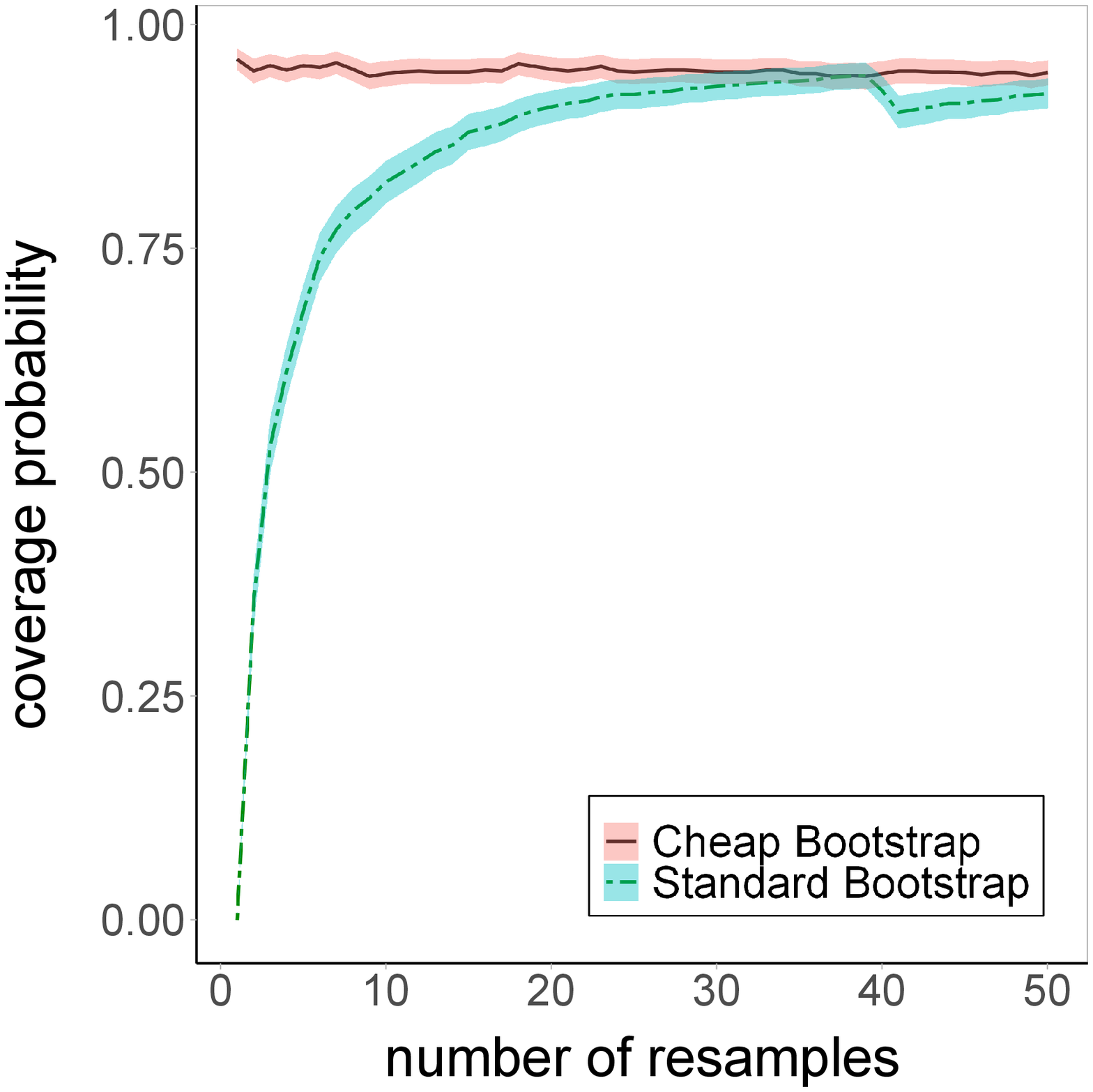} }
\subfigure
{\includegraphics[width=.24\columnwidth,height=.15\textheight]{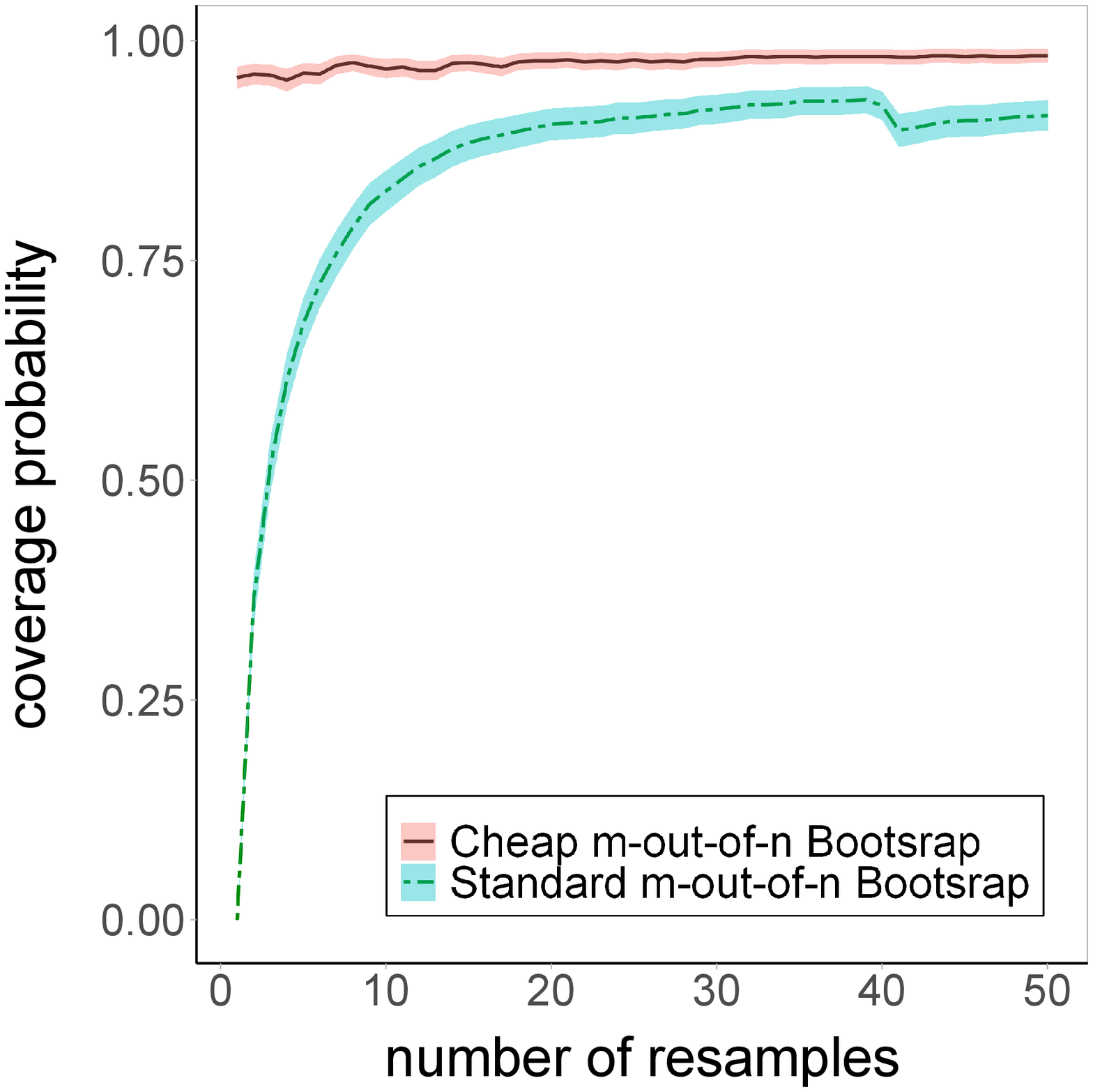}}
\subfigure
{\includegraphics[width=.24\columnwidth,height=.15\textheight]{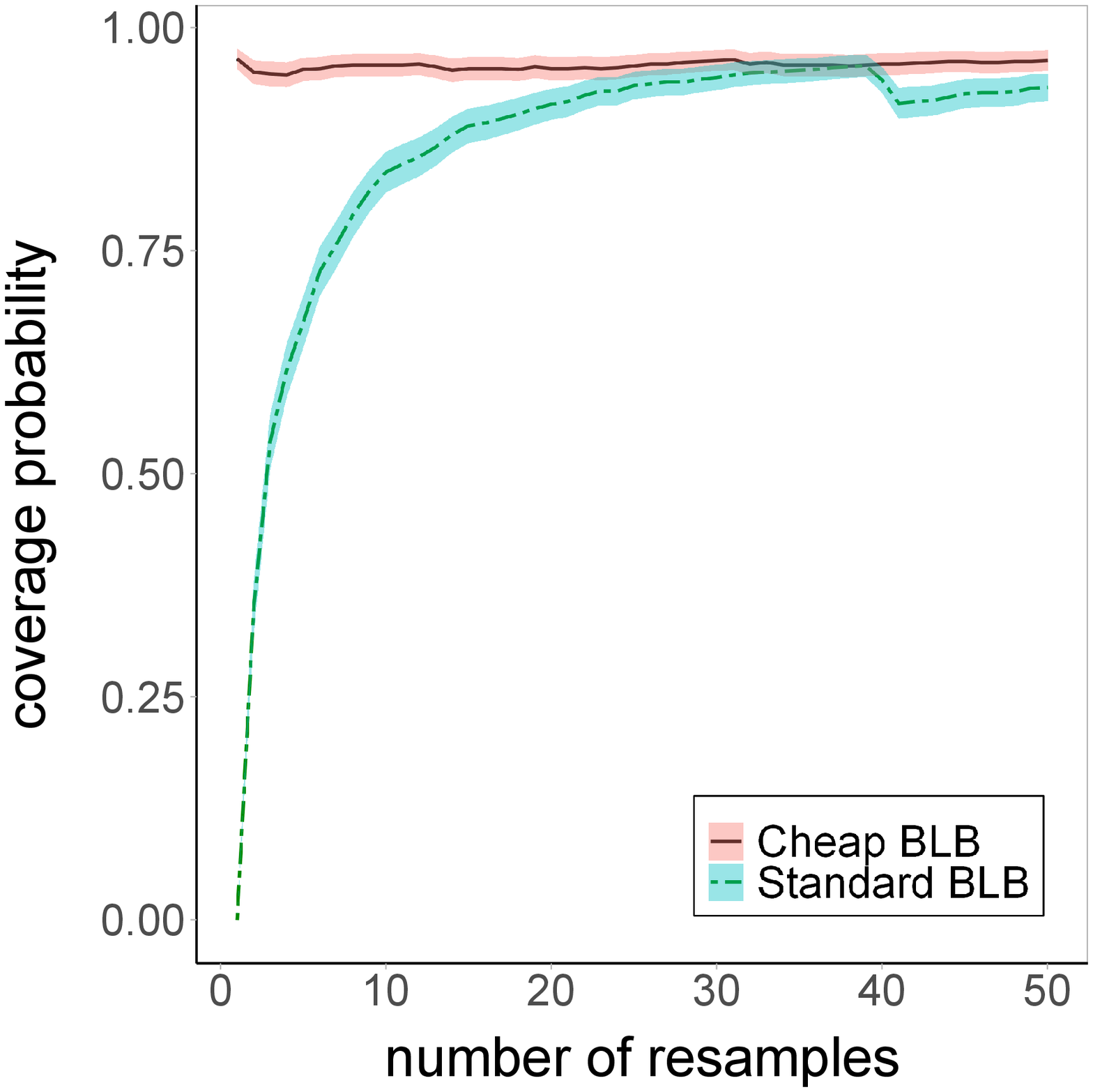}}
\subfigure
{\includegraphics[width=.24\columnwidth,height=.15\textheight]{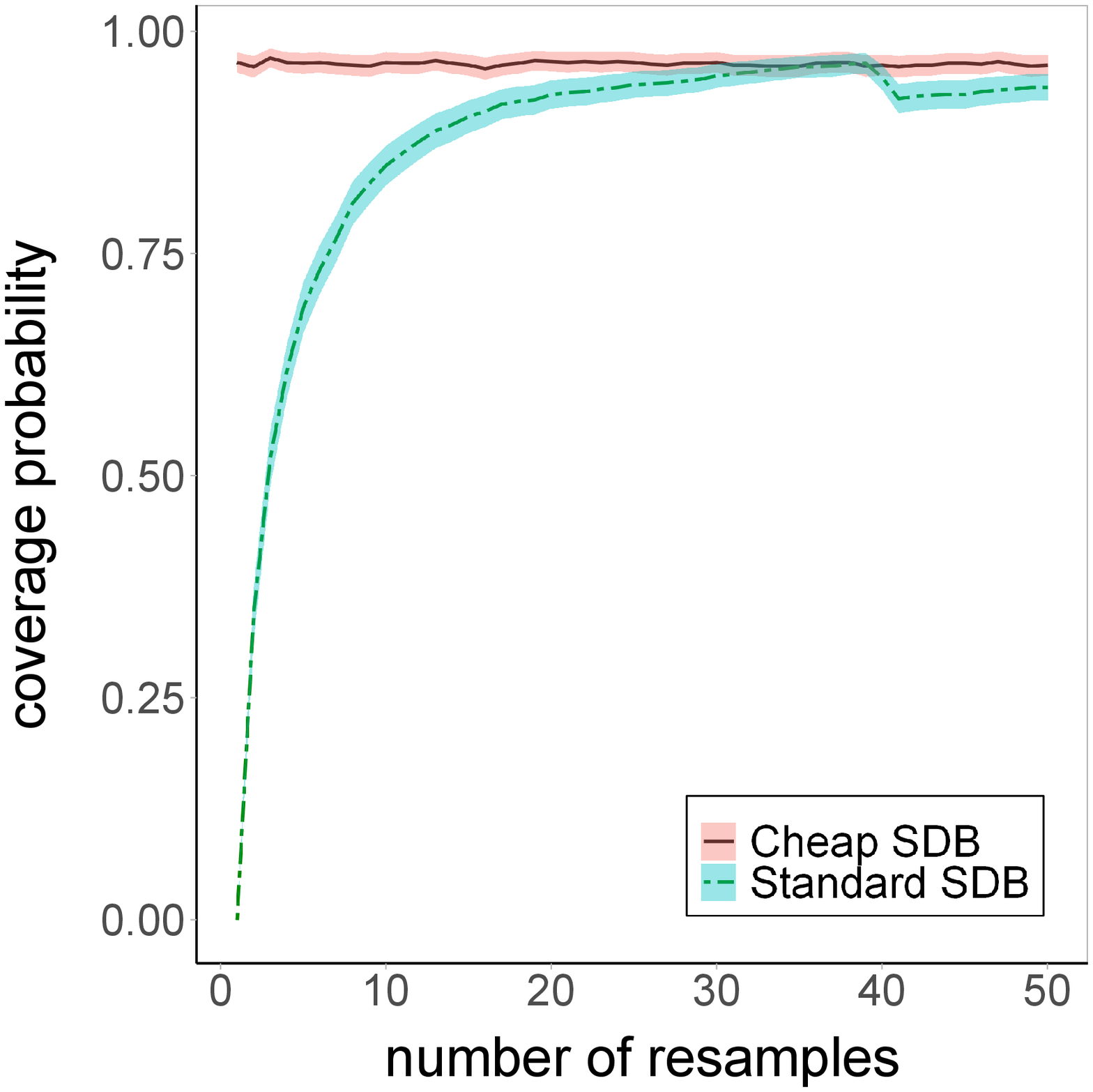} }
\caption{Confidence interval coverage probabilities of Standard versus Cheap Bootstrap methods in logistic regression. Nominal confidence level $=95\%$ and sample size $n=10^5$. Shaded areas depict the associated confidence intervals of the coverage probability estimates from 1000 experimental repetitions.}
\label{fig:cov log}
\end{center}
\vskip -0.2in
\end{figure*}

\begin{figure*}[tb]
\vskip 0.2in
\begin{center}
\subfigure
{\includegraphics[width=.24\columnwidth,height=.15\textheight]{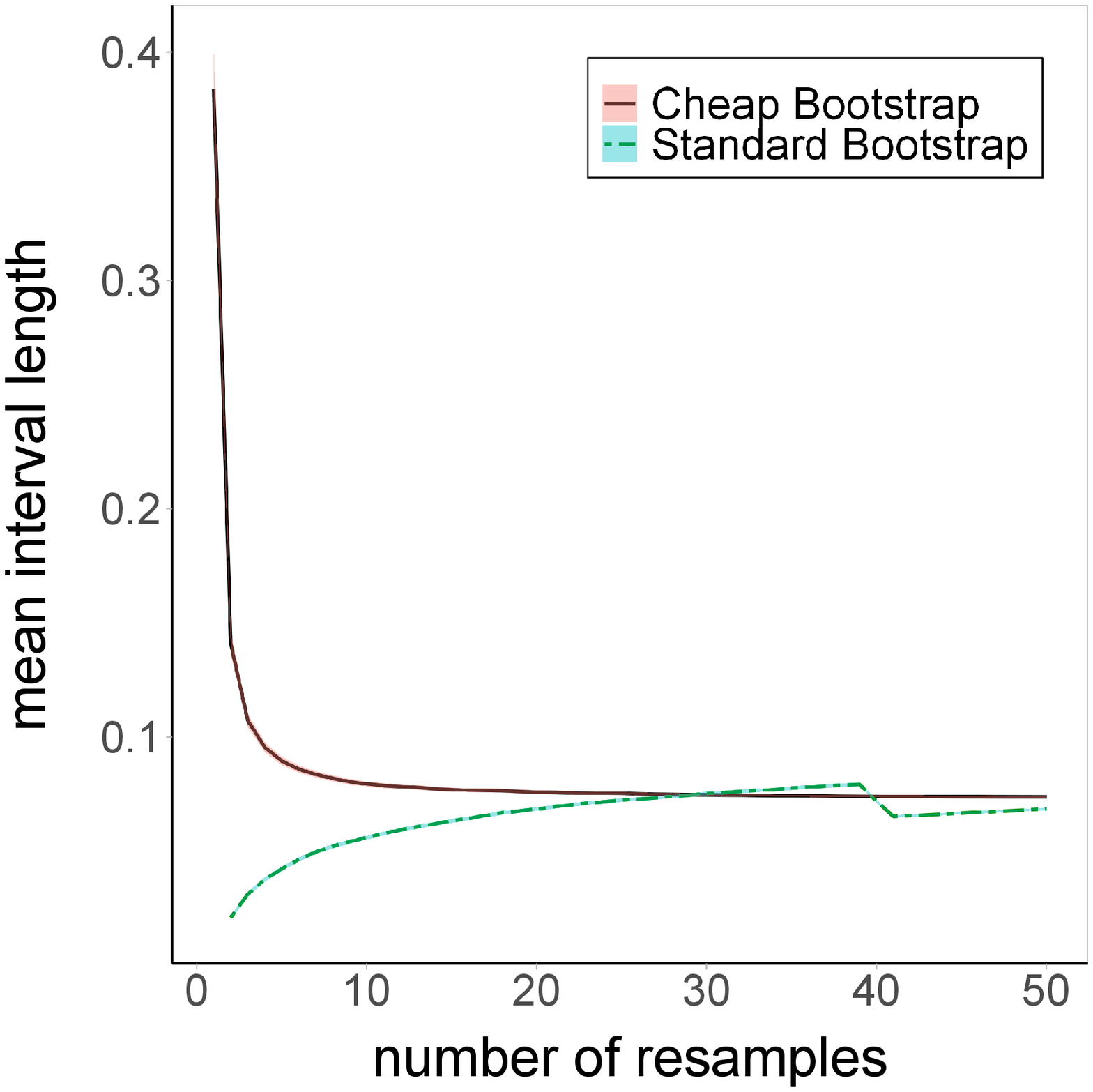}}
\subfigure
{\includegraphics[width=.24\columnwidth,height=.15\textheight]{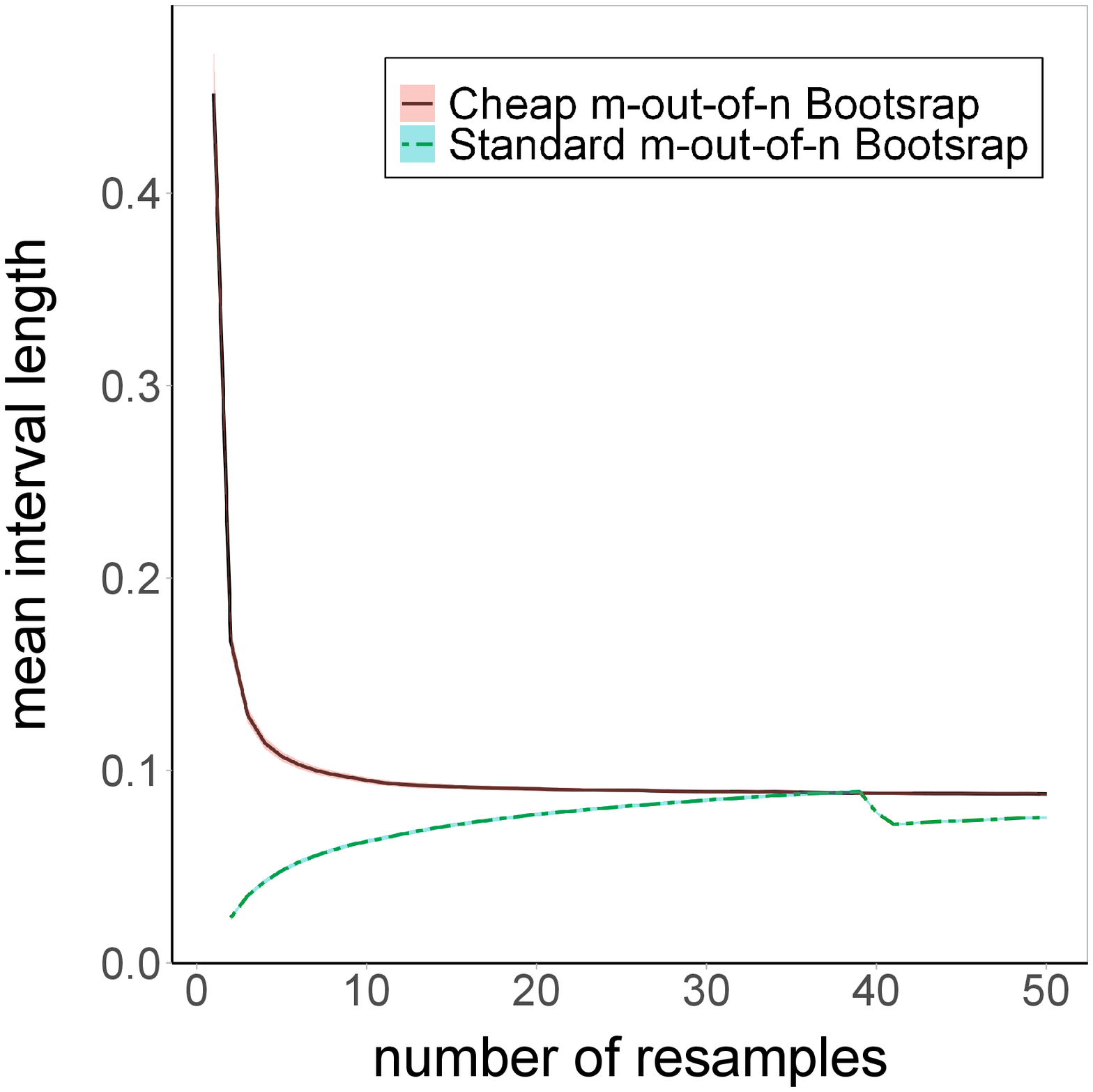}}
\subfigure
{\includegraphics[width=.24\columnwidth,height=.15\textheight]{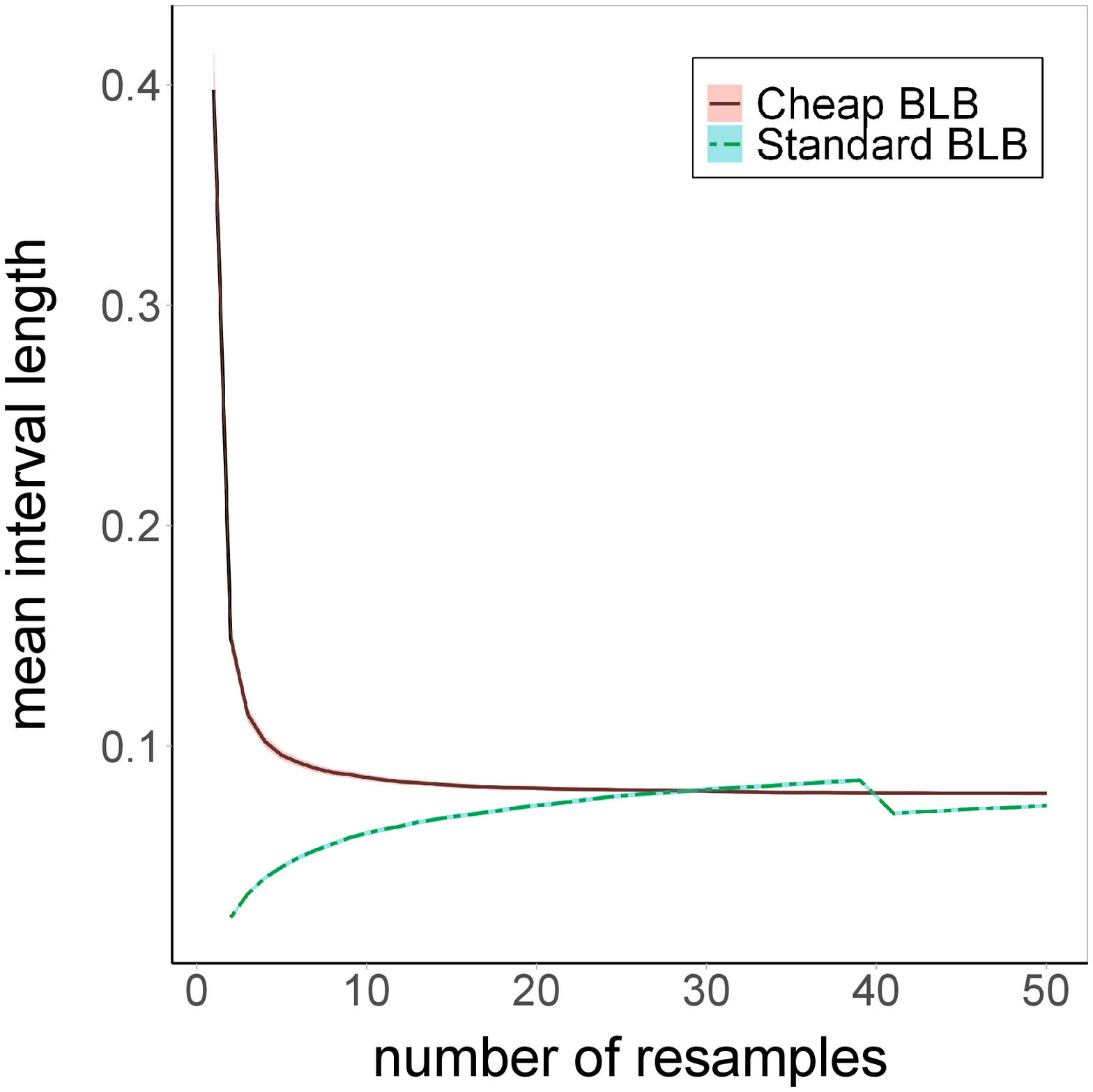}}
\subfigure
{\includegraphics[width=.24\columnwidth,height=.15\textheight]{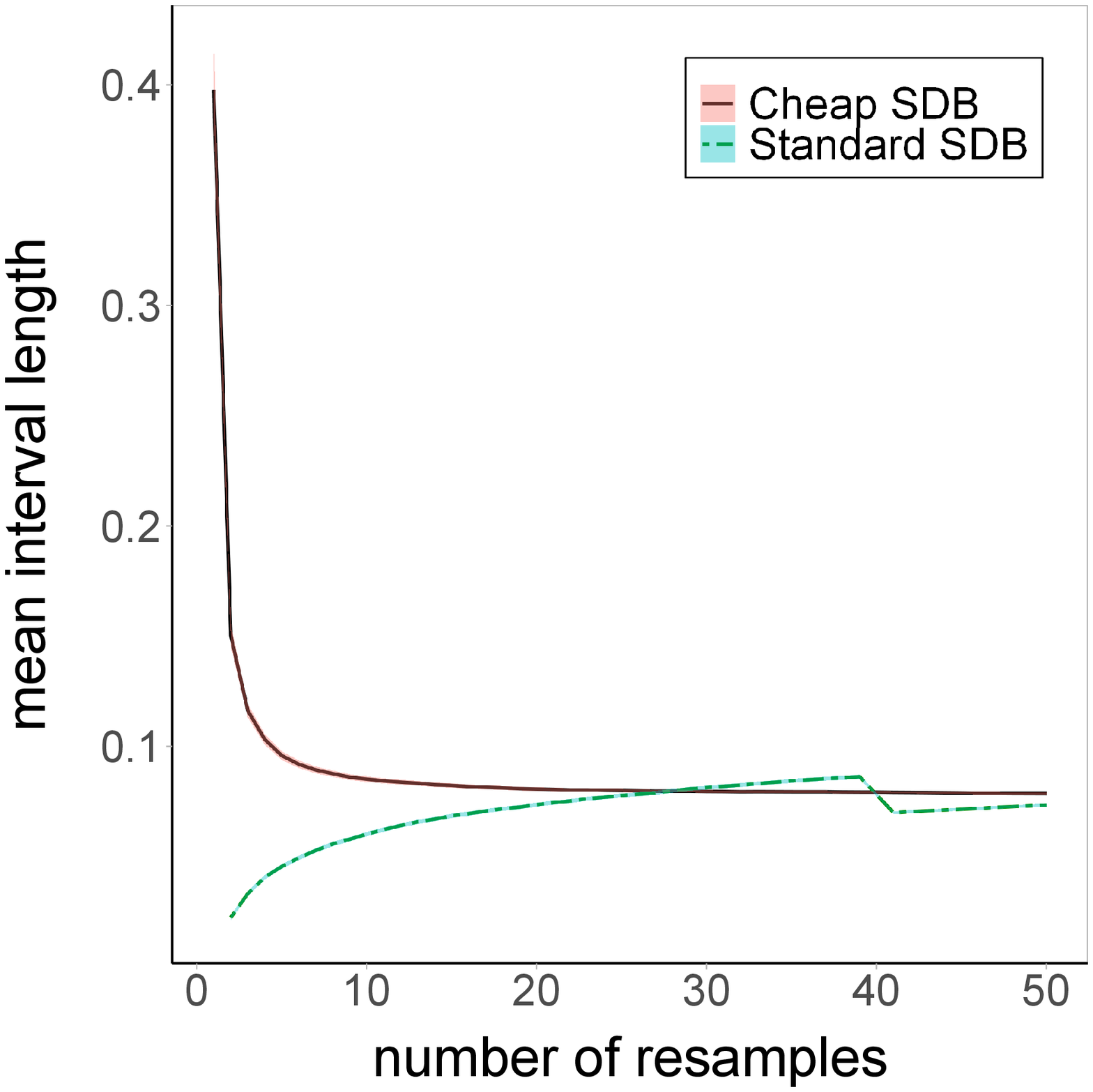} }
\caption{Mean confidence interval widths of Standard versus Cheap Bootstrap methods in logistic regression. Nominal confidence level $=95\%$ and sample size $n=10^5$. Shaded areas depict the associated confidence intervals of the mean width estimates from 1000 experimental repetitions.}
\label{fig:len log}
\end{center}
\vskip -0.2in
\end{figure*}

\begin{figure*}[tb]
\vskip 0.2in
\begin{center}
\subfigure
{\includegraphics[width=.24\columnwidth,height=.15\textheight]{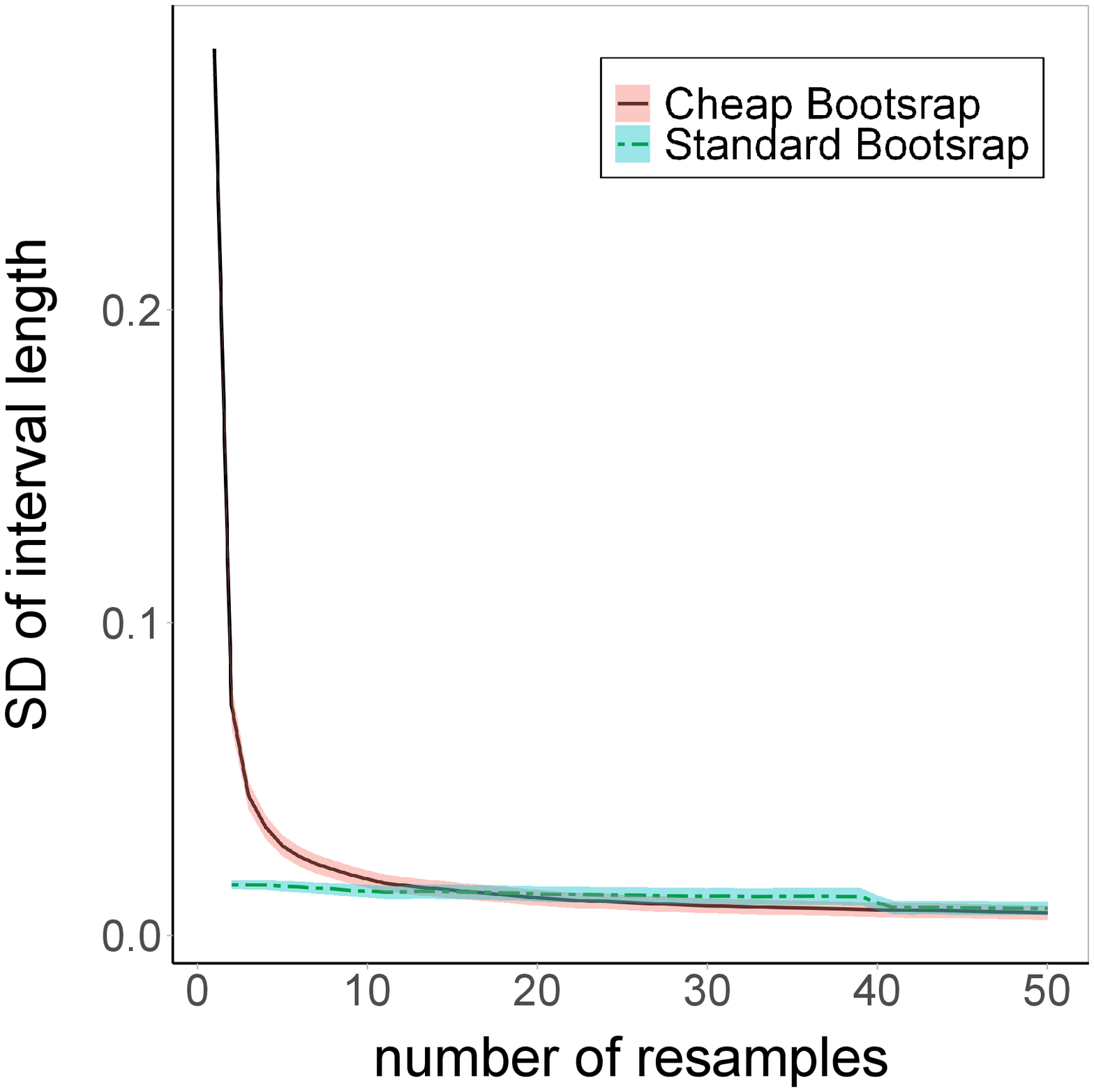} }
\subfigure
{\includegraphics[width=.24\columnwidth,height=.15\textheight]{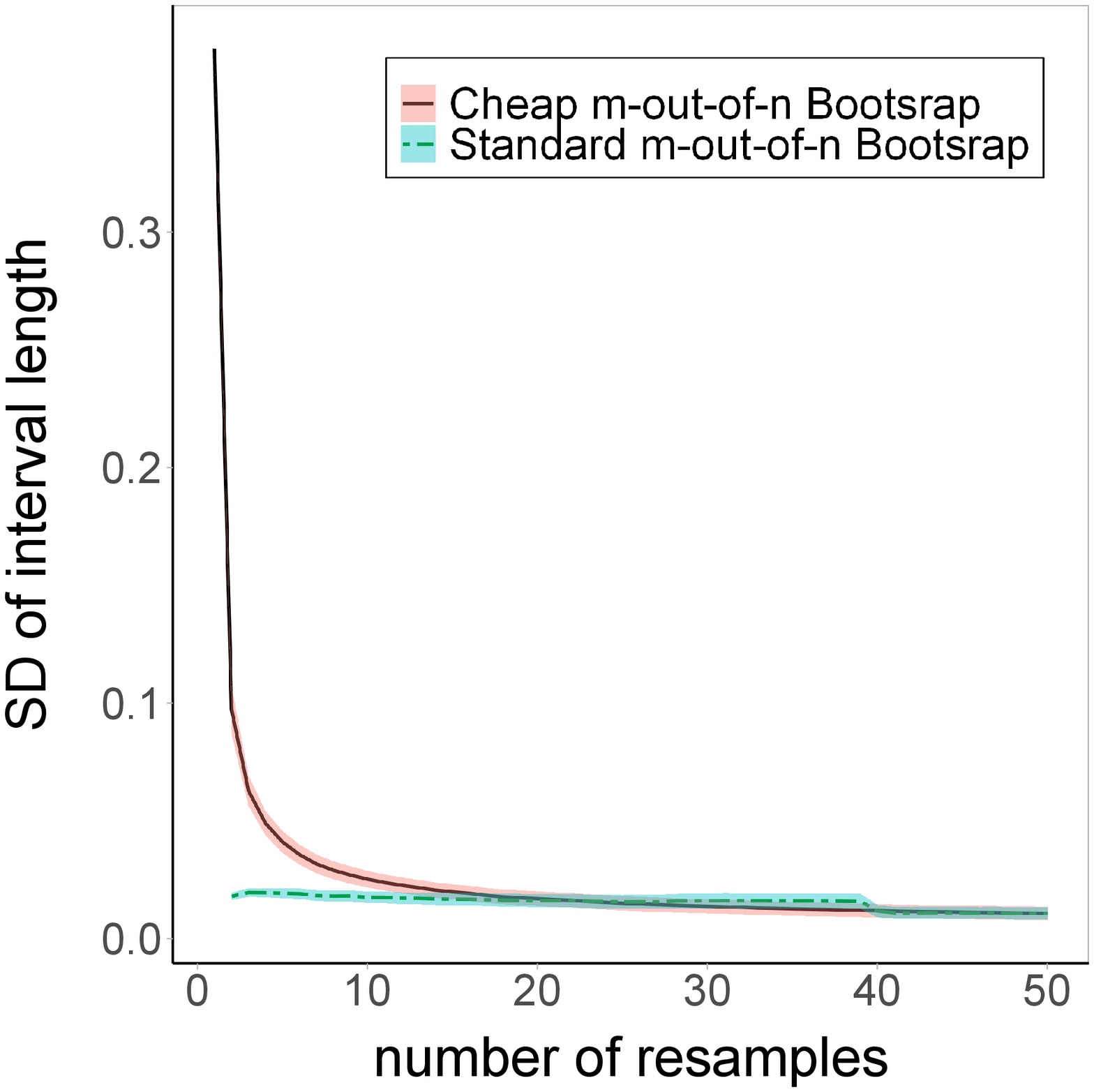}}
\subfigure
{\includegraphics[width=.24\columnwidth,height=.15\textheight]{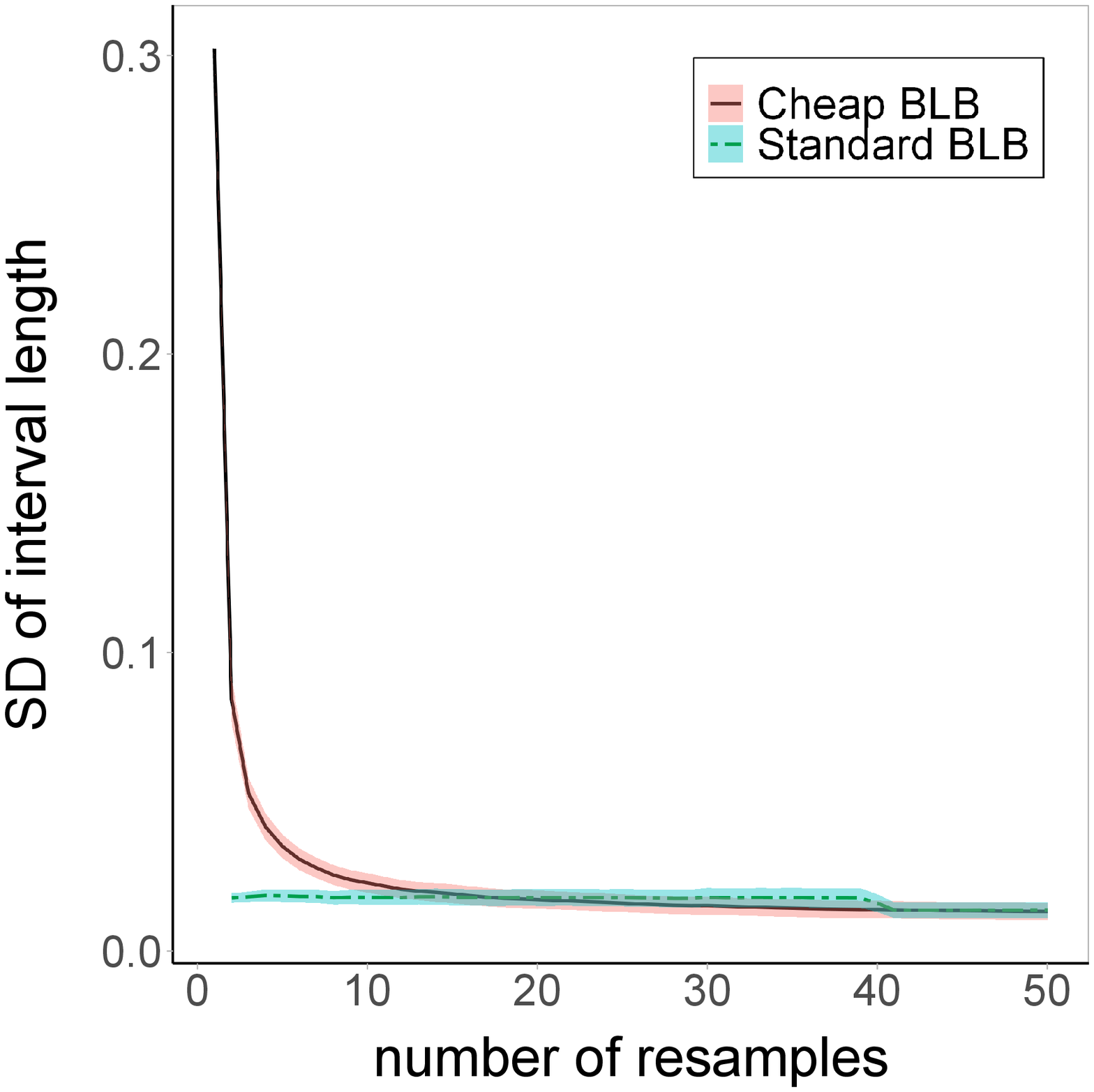}}
\subfigure
{\includegraphics[width=.24\columnwidth,height=.15\textheight]{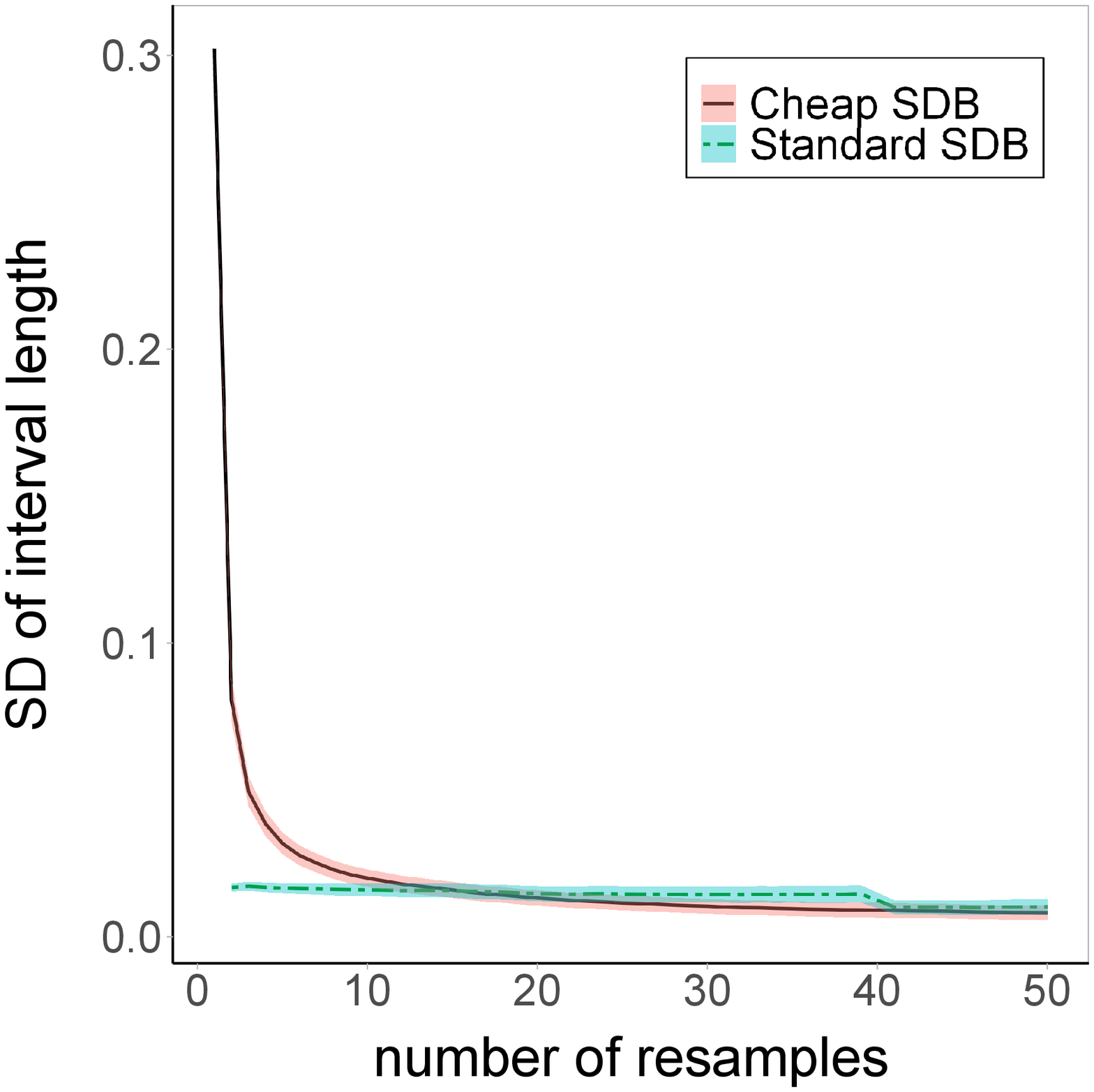} }
\caption{Standard deviations of confidence interval widths of Standard versus Cheap Bootstrap methods in logistic regression. Nominal confidence level $=95\%$ and sample size $n=10^5$. Shaded areas depict the associated confidence intervals of the standard deviation estimates from 1000 experimental repetitions.}
\label{fig:sd log}
\end{center}
\vskip -0.2in
\end{figure*}

\subsection{Critical Values for Cheap Bootstrap in Nested Sampling Problems}\label{sec:q compute}
Table \ref{table:nested} shows the approximated $q_{O,1-\alpha/2}$ for $B$ ranging from $1$ to $20$ and $\alpha=0.05$ (and also $\alpha=0.1$ for later use in the example in Appendix \ref{sec:bagging}) when $R_0=R$, i.e., $\rho=1$, where we also show $q_{M,1-\alpha/2}=t_{B-1,1-\alpha/2}$ for comparison. Note that when $B=1$, only $q_{O,1-\alpha/2}$ is well-defined but not $q_{M,1-\alpha/2}$. We also see that while $q_{O,1-\alpha/2}$ is smaller than $q_{M,1-\alpha/2}$ when $B$ is small, it appears that they are very similar as $B$ reaches $20$.

\begin{table}[!ht]
\caption{Values of $q_{O,1-\alpha/2}$ and $q_{M,1-\alpha/2}$ when $R_0=R$, at $\alpha=0.05$ and $\alpha=0.1$.}
\centering
\begin{tabular}{c|cc|cc}
\multirow{2}{*}{$B$}&\multicolumn{2}{|c}{$\alpha=0.05$}&\multicolumn{2}{|c}{$\alpha=0.1$}\\
 &$q_{O,1-\alpha/2}$&$q_{M,1-\alpha/2}$&$q_{O,1-\alpha/2}$&$q_{M,1-\alpha/2}$\\
   \hline
 1 &   12.75 & NA&6.32&NA\\
  2&     4.32 &  12.71&2.92&6.31\\
  3 &    3.19 &   4.30&2.36&2.92\\
  4 &   2.78  &  3.18&2.14&2.35\\
  5 &   2.57 &   2.78&2.02&2.13\\
  6 &    2.45  &  2.57&1.95&2.02\\
  7 &    2.37 &   2.45&1.90&1.94\\
  8 &    2.31 &   2.36&1.86&1.89\\
  9 &    2.27 &   2.31&1.84&1.86\\
 10 &    2.23 &   2.26&1.82&1.83\\
 11 &    2.21 &   2.23&1.80&1.81\\
 12 &    2.19 &   2.20&1.79&1.80\\
 13 &    2.16 &   2.18&1.78&1.78\\
 14 &    2.15 &   2.16&1.77&1.77\\
 15 &    2.14 &   2.14&1.76&1.76\\
 16 &    2.12 &   2.13&1.75&1.75\\
 17 &    2.11 &   2.12&1.74&1.75\\
 18 &    2.11 &   2.11&1.74&1.74\\
 19 &    2.10  &   2.10&1.73&1.73\\
 20 &    2.09  &   2.09&1.73&1.73
\end{tabular}
\label{table:nested}
\end{table}

\subsection{Bagging Estimation for Bounding Optimality Gap}\label{sec:bagging}
We apply the Cheap Bootstrap on a bagging method to construct upper confidence bounds for data-driven stochastic optimization problems. More specifically, consider an expected-value optimization problem $\min_{\theta\in\Theta}\{H(\theta):=E[h(\theta,X)]\}$ where the distribution $P$ governing $X$ in the expectation $E[\cdot]$ is unknown and is only informed from data. The cost function $h(\cdot,\cdot)$ is known or evaluatable, and $\Theta$ denotes the feasible region for the decision variable $\theta$. Such an optimization formulation appears broadly in multiple disciplines such as revenue management, portfolio selection, among others (e.g., \cite{shapiro2021lectures,birge2011introduction}).

Suppose we have a given solution, say $\hat\theta$ (which is presumably obtained from a data-driven procedure such as solving an empirical optimization or sample average approximation \cite{shapiro2021lectures}). To assess the quality of this given solution, we can construct an upper confidence bound for its optimality gap $G=H(\hat\theta)-H^*$, where $H^*=\min_{\theta\in\Theta}H(\theta)$ denotes the unknown true optimal value. To this end, a general upper bound for $G$ is given by
$E[\tilde H(X_1,X_2,\ldots,X_k)]$
where $X_1,X_2,\ldots,X_k\stackrel{i.i.d.}{\sim}P$ and $E[\cdot]$ denotes the corresponding expectation. The quantity
$$\tilde H(X_1,X_2,\ldots,X_k)=\max_{\theta\in\Theta}\frac{1}{k}\sum_{i=1}^k(h(\hat\theta,X_i)-h(\theta,X_i))$$
is a sample average approximation on cost function $h(\hat\theta,X)-h(\theta,X)$, which is equivalent to the difference between an empirical estimate of $E[h(\hat\theta,X)]$ and a sample average approximation on $h(\theta,X)$ itself. That $E[\tilde H(X_1,X_2,\ldots,X_k)]$ is an upper bound for $G$ can be argued via the Jensen inequality and is a well-established optimistic bound in stochastic optimization (e.g., \cite{mak1999monte}, \cite{glasserman2004monte} \S8).

To implement $E[\tilde H(X_1,X_2,\ldots,X_k)]$ using data, \cite{lam2018assessing,lam2018bounding} suggest to use bagging. Namely, given data $X_1,\ldots,X_n$, where $n$ denotes the sample size (that could be different from $k$), we repeatedly resample data set $\{X_1^{*b},\ldots,X_k^{*b}\}$ and solve a sample average approximation $\tilde H(X_1^{*b},X_2^{*b},\ldots,X_k^{*b})$, for $b=1,\ldots,B$. Then we output their average
$$\frac{1}{B}\sum_{b=1}^B\tilde H(X_1^{*b},X_2^{*b},\ldots,X_k^{*b})$$
to give a point estimate for $E[\tilde H(X_1,X_2,\ldots,X_k)]$. Cast in our framework in Section \ref{sec:double}, we view $\psi(P)=E[\tilde H(X_1,X_2,\ldots,X_k)]$ where $X_1,\ldots,X_k\stackrel{i.i.d.}{\sim}P$ and each noisy computation run as $\hat\psi_r(Q)=\tilde H(X_1^*,X_2^*,\ldots,X_k^*)$ where $X_1^*,,\ldots,X_k^*\stackrel{i.i.d.}{\sim}Q$. Thus we can use the Cheap Bootstrap centered at original estimate and centered at resample mean in Section \ref{sec:double} to construct valid confidence bounds. Note that each computation run here involves solving a sample average approximation problem which could be costly.

To test the performance of our Cheap Bootstrap, we consider more specifically the following expected-value optimization problem 
$$\begin{array}{ll}
\min_\theta&E[X^\top\theta]\\
\text{subject to}&A\theta\leq b\\
&\theta_i\in\{0,1\}\text{\ for\ }i=1,2,\ldots,10
\end{array}$$
This is a $10$-dimensional stochastic integer program, with binary decision variables $\theta_i,i=1,\ldots,10$. The true distribution of $X\in\mathbb R^{10}$ is $N(\mu,\Sigma)$, where $\mu=(-1,-7/9,-5/9,\ldots,7/9,1)^\top$ and $\Sigma$ is an arbitrarily generated covariance matrix. We set $b=(-1,2)^\top$ and
\begin{equation*}
   A=\left[\begin{matrix} 
      -1 &-1&-1&-1&-1&-1&-1&-1&-1&-1  \\
      0&0 &0&0&0&0&1&1&1&1 \\
   \end{matrix}\right].
\end{equation*}
This example is also used in \cite{lam2018bounding}. We use the data size $n=100$, and the size in sample average approximation $k=50$. Throughout our set of experiments, we set $\hat\theta=(1,1,0,1,0,0,0,0,0,0)^\top$ which is obtained by solving a sample average approximation from an independent data set of size $30$. We construct our one-sided Cheap Bootstrap intervals using $q_{O,1-\alpha}$ and $q_{M,1-\alpha}$ at $\alpha=0.05$ (which corresponds to the case $\alpha=0.1$ in Table \ref{table:nested}), for $B$ ranging from $1$ to $10$. We repeat our experiment $1000$ times to obtain summary statistics on our generated bounds, including the one-sided empirical coverage, i.e., the proportion of experiments where our upper confidence bound is at least the ground true value, and the mean and standard deviation of the upper confidence bound. 

Table \ref{table:bagging} shows the results. We see that the empirical coverages range from $96\%$ to $98\%$ in all considered $B$ for Cheap Bootstrap centered at original estimate and centered at resample mean (note that the latter is defined only starting from $B=2$). The coverages are higher than the nominal value $95\%$ likely because the true value targeted by the confidence bound, namely $E[\tilde H(X_1,X_2,\ldots,X_k)]$, is itself an upper bound on the true optimality gap $G$. The trends of the generated confidence bounds appear consistent with our previous examples. For centered at original estimate, the mean falls quickly from $2.83$ at $B=1$ to $1.91$ at $B=2$, further to $1.72$ at $B=3$ and much more steadily to $1.58$ at $B=10$, while the standard deviation falls from $1.88$ at $B=1$ to $0.91$ at $B=2$, to $0.73$ at $B=3$, and steadily to $0.62$ at $B=10$. Similarly, for centered at resample mean, the mean falls quickly from $2.79$ at $B=2$ to $1.84$ at $B=3$, and then steadily to $1.56$ at $B=10$, while the standard deviation falls from $1.71$ at $B=2$ to $0.81$ at $B=3$ and steadily to $0.62$ at $B=10$. The performance of centered at original estimate appears better than centered at resample mean for small $B$ but their performances are similar as $B$ reaches close to $10$.

\begin{table}[ht]
\caption{Performances of upper confidence bounds using Cheap Bootstrap centered at original estimate and centered at resample mean, at nominal confidence level $95\%$ for bagging estimation of optimality gap.}
\centering
\begin{tabular}{c|cc|cc}
\multirow{4}{*}{$B$}&\multicolumn{2}{c}{Centered at original estimate}&\multicolumn{2}{|c}{Centered at resample mean}\\
   &  Empirical coverage & Bound mean & Empirical coverage & Bound mean \\
  &  (margin of error) & (st. dev.) & (margin of error) & (st. dev.) \\
   \hline
   1& 0.96 \ \ (0.01) & 2.83\ \ (1.88) & NA & NA \\
2& 0.97 \ \ (0.01) & 1.91\ \ (0.91) & 0.98 \ \ (0.01) & 2.79\ \ (1.71) \\
3& 0.97 \ \ (0.01) & 1.72\ \ (0.73) & 0.98 \ \ (0.01) & 1.84\ \ (0.81) \\
4& 0.96 \ \ (0.01) & 1.67\ \ (0.69) & 0.97 \ \ (0.01) & 1.72\ \ (0.72) \\
5& 0.97 \ \ (0.01) & 1.64\ \ (0.67) & 0.97 \ \ (0.01) & 1.66\ \ (0.68) \\
6& 0.97 \ \ (0.01) & 1.61\ \ (0.65) & 0.96 \ \ (0.01) & 1.62\ \ (0.66) \\
7& 0.97 \ \ (0.01) & 1.60\ \ (0.65) & 0.97 \ \ (0.01) & 1.60\ \ (0.65) \\
8& 0.97 \ \ (0.01) & 1.58\ \ (0.63) & 0.97 \ \ (0.01) & 1.58\ \ (0.63) \\
9& 0.97 \ \ (0.01) & 1.58\ \ (0.63) & 0.97 \ \ (0.01) & 1.57\ \ (0.63) \\
10& 0.97 \ \ (0.01) & 1.58\ \ (0.62) & 0.97 \ \ (0.01) & 1.56\ \ (0.62)
\end{tabular}
\label{table:bagging}
\end{table}

\section{Useful Technical Backgrounds}
For self-containedness purpose, we provide some existing results needed for our technical developments.
\subsection{Hadamard Differentiability}\label{sec:HD}
We give some details on Hadamard differentiability used in Proposition \ref{HD} and Theorem \ref{main variant}. Consider a functional $\psi(\cdot):\mathcal P\to\mathbb R^d$, where $\mathcal P$ denotes the space of probability distribution on the domain $\mathcal X$. We call $\psi(\cdot)$ Hadamard differentiable at $P$ with derivative $\psi_P'(H)$, tangential to some subset $\mathcal Q$ of $\mathcal P$, if there exists a continuous, linear map $\psi_P':\mathcal Q\to\mathbb R^d$ such that
$$\left\|\frac{\psi(P+tH_t)-\psi(P)}{t}-\psi_P'(H)\right\|\to0$$
as $t\searrow0$ for every sequence $H_t$ such that $P+tH_t\in\mathcal P$  for any small $t>0$ and converging to $H\in\mathcal Q$ (\cite{van2000asymptotic} \S20.2).

\subsection{Bootstrap Empirical Processes}\label{sec:sufficient proof}
We say that a sequence of random elements $G_n$ in a normed space $\mathbb D$, with norm denoted $\|\cdot\|$, converges in distribution to a tight limit $G$ in $\mathbb D$ if
\begin{equation}
\sup_{h\in BL_1(\mathbb D)}|E^*h(G_n)-Eh(G)|\to0\label{useful interim}
\end{equation}
where $BL_1(\mathbb D)$ is the set of all functions $h:\mathbb D\to[-1,1]$ that are uniformly Lipschitz, i.e., $|h(z_1)-h(z_2)|\leq \|z_1-z_2\|$ for every pair $z_1,z_2\in\mathbb D$, and $E^*[\cdot]$ denotes the outer expectation.

For a class of functions $\mathcal F$ from $\mathcal X$ to $\mathbb R$, define
$$\ell^\infty(\mathcal F):=\left\{z:\|z\|_{\mathcal F}:=\sup_{f\in\mathcal F}|z(f)|<\infty\right\}$$
where $z$ is a map from $\mathcal F$ to $\mathbb R$. Consider the empirical process $\mathbb G_n=\sqrt n(\hat P_n-P)$ as a random element that takes value in $\ell^\infty(\mathcal F)$, and consider the bootstrap empirical process $\mathbb G_n^*=\sqrt n(P_n^*-\hat P_n)$. Supposing $\mathcal F$ is Donsker with a finite envelope function, it is well established that $\mathbb G_n\Rightarrow\mathbb G_P$ in $\ell^\infty(\mathcal F)$ where $\mathbb G_P$ is a tight Gaussian process (with mean 0 and covariance $Cov(\mathbb G_P(f_1),\mathbb G_P(f_2))=Cov_P(f_1(X),f_2(X))$ where $Cov_P$ denotes the covariance taken with respect to $P$). Denote $E_M[\cdot]$ as the expectation conditional on $X_1,\ldots,X_n$. We recall the following result for the bootstrap empirical process.

\begin{theorem}[Adapted from \cite{van2000asymptotic} Theorem 23.7]
For every Donsker class $\mathcal F$ of measurable functions with a finite envelope function,
$$\sup_{h\in BL_1(\ell^\infty(\mathcal F))}|E_Mh(\mathbb G_n^*)-Eh(\mathbb G_P)|\stackrel{p}{\to}0$$
Furthermore, the sequence $\mathbb G_n^*$ is asymptotically measurable.\label{BEP}
\end{theorem}

Next we recall two theorems:

\begin{theorem}[Adapted from \cite{van2000asymptotic} Theorem 20.8]
Let $\mathbb D$ be a normed space and $\phi:\mathbb D_\phi\subset\mathbb D\to\mathbb R^d$ be Hadamard differentiable at $\theta$ tangential to some subspace $\mathbb D_0$. Let $\hat\theta_n$ be random maps with values in $\mathbb D_\phi$ such that $\sqrt n(\hat\theta_n-\theta)\Rightarrow\mathbb T$ where $\mathbb T$ takes values in $\mathbb D_0$. Then $\sqrt n(\phi(\hat\theta_n)-\phi(\theta))\Rightarrow\phi_\theta'(\mathbb T)$.\label{EP delta}
\end{theorem}

\begin{theorem}[Adapted from \cite{van2000asymptotic} Theorem 23.9]
Let $\mathbb D$ be a normed space and $\phi:\mathbb D_\phi\subset\mathbb D\to\mathbb R^d$ be Hadamard differentiable at $\theta$ tangential to some subspace $\mathbb D_0$. Let $\hat\theta_n$ and $\theta_n^*$ be random maps with values in $\mathbb D_\phi$ such that $\sqrt n(\hat\theta_n-\theta)\Rightarrow\mathbb T$ and $\sup_{h\in BL_1(\mathbb D)}|E_Mh(\sqrt n(\theta_n^*-\hat\theta_n))-Eh(\mathbb T)|\stackrel{p}{\to}0$, in which $\sqrt n(\theta_n^*-\hat\theta_n)$ is asymptotically measurable and $\mathbb T$ is tight and takes values in $\mathbb D_0$. Then $\sqrt n(\phi(\theta_n^*)-\phi(\hat\theta_n))\Rightarrow\phi_\theta'(\mathbb T)$ conditionally given $X_1,X_2,\ldots$ in probability.\label{BEP delta}
\end{theorem}


In general, we say $\sqrt n(\theta_n^*-\hat\theta_n)$ weakly converges to $\mathbb T$ conditionally given $X_1,X_2,\ldots$ in probability if the condition $\sup_{h\in BL_1(\mathbb D)}|E_Mh(\sqrt n(\theta_n^*-\hat\theta_n))-Eh(\mathbb T)|\stackrel{p}{\to}0$ in Theorem \ref{BEP delta} holds (\cite{van2000asymptotic} equation (23.8)). By \cite{kosorok2007introduction} Lemma 10.11, in the case $\mathbb D=\mathbb R^d$, this condition implies $P(\sqrt n(\theta_n^*-\hat\theta_n)\leq x|\hat P_n)\stackrel{p}{\to}F(x)$ for all $x\in\mathbb R^d$
if the distribution function $F(\cdot)$ of $\mathbb T$ is continuous.

The following theorem, which is an immediate consequence of the above results, is used to justify Propositions \ref{HD} and \ref{HD multi}.

\begin{theorem}[Delta method for empirical bootstrap]
Consider $\hat P_n$ and $P_n^*$ as random elements that take values in $\ell^\infty(\mathcal F)$, where $\mathcal F$ is a Donsker class with a finite envelope. Suppose $\phi:\ell^\infty(\mathcal F)\to\mathbb R^d$ is Hadamard differentiable at $P$ (tangential to $\ell^\infty(\mathcal F)$). Then $\sqrt n(\phi(\hat P_n)-\phi(P))\Rightarrow\phi_P'(\mathbb G_P)$, and also $\sqrt n(\phi(P_n^*)-\phi(\hat P_n))\Rightarrow\phi_P'(\mathbb G_P)$ given $X_1,X_2,\ldots$ in probability.\label{thm:delta empirical}
\end{theorem}

\begin{proof}[Proof of Theorem \ref{thm:delta empirical}]
Setting $\hat\theta_n=\hat P_n$ and $\mathbb D=\ell^\infty(\mathcal F)$, Theorem \ref{EP delta} implies $\sqrt n(\phi(\hat P_n)-\phi(P))\Rightarrow\phi_P'(\mathbb G_P)$. Moreover, Theorem \ref{BEP} gives the conditions needed in Theorem \ref{BEP delta} to conclude that $\sqrt n(\phi(P_n^*)-\phi(\hat P_n))\Rightarrow\phi_P'(\mathbb G_P)$ given $X_1,X_2,\ldots$ in probability.
\end{proof}



\subsection{Edgeworth Expansions}\label{sec:background Edgeworth}
We have the following higher-order expansion of coverage probability for a general function-of-mean model:
\begin{theorem}[\cite{hall2013bootstrap} Theorem 2.2]
Assume that for an integer $\nu\geq1$, the function $\tilde A$ has $\nu+2$ continuous derivatives in a neighborhood of $\bm\mu$, that $\tilde A(\bm\mu)=0$, that $E\|\mathbf X\|^{\nu+2}<\infty$, that the characteristic function $\chi$ of $\mathbf X$ satisfies Cramer's condition $\limsup_{\|\mathbf t\|\to\infty}|\chi(\mathbf t)|<1$, and that the asymptotic variance of $\sqrt n\tilde A(\overline{\mathbf X})$ equals 1. Then
$$P(\sqrt n\tilde A(\overline{\mathbf X})\leq x)=\Phi(x)+\sum_{j=1}^\nu n^{-j/2}\pi_j(x)\phi(x)+o(n^{-\nu/2})$$
uniformly in $x$, where $\pi_j$ is a polynomial of degree $3j-1$, odd for even $j$ and even for odd $j$, with coefficients depending on moments of $\mathbf X$ up to order $j+2$ polynomially and also $A$.\label{Edgeworth original}
\end{theorem}

In Theorem \ref{Edgeworth original}, $\tilde A$ is generally defined and the condition $\tilde A(\bm\mu)=0$ is satisfied by $A$ and $A_s$ defined in \eqref{A def} and \eqref{A def2}. When $\tilde A=A$, we denote its $\pi_j$ as $p_j$. When $\tilde A=A_s$, we denote its $\pi_j$ as $q_j$.

Now denote $\hat\pi_j(\cdot)$ as $\pi_j(\cdot)$ but with all moments in its coefficients replaced by the sample moments, i.e., denoting $\mathbf X=(X(1),\ldots,X(d))$, the moment
$$\mu_{m_1,\ldots,m_d}=E[{X(1)}^{m_1}\cdots{X(d)}^{m_d}]$$
is replaced by
$$\hat\mu_{m_1,\ldots,m_d}=\frac{1}{n}\sum_{i=1}^n{X_i(1)}^{m_1}\cdots{X_i(d)}^{m_d}$$
for sample $\mathbf X_i=(X_i(1),\ldots,X_i(d)),i=1,\ldots,n$. Specifically, in the case of $A$, we denote $\hat p_j$ as the $p_j$ with moments replaced by sample moments. Moreover, we also define
$$\hat A(\mathbf x)=\frac{g(\mathbf x)-g(\overline{\mathbf X})}{h(\overline{\mathbf X})}$$
Then, $\hat A(\overline{\mathbf X}^*)$, where $\overline{\mathbf X}^*=(1/n)\sum_{i=1}^n\mathbf X_i^*$ for a resample $\mathbf X_i^*,i=1,\ldots,n$, is the resample counterpart of $A(\overline{\mathbf X})$. We have the following expansion and bounds for the resample counterpart:

\begin{theorem}[Adapted from \cite{hall2013bootstrap} Theorem 5.1]
Let $\lambda>0$ be given, and let $l=l(\lambda)$ be a sufficiently large positive number. Assume that $g$ and $h$ each have $\nu+3$ bounded derivatives in a neighborhood of $\bm\mu$, that $E\|\mathbf X\|^l<\infty$, and that the characteristic function $\chi$ of $\mathbf X$ satisfies Cramer's condition $\limsup_{\|\mathbf t\|\to\infty}|\chi(\mathbf t)|<1$. Then there exists a constant $C>0$ such that
$$P\left(\sup_{-\infty<x<\infty}\left|P(\sqrt n\hat A(\overline{\mathbf X}^*)\leq x|\mathcal X_n)-\Phi(x)-\sum_{j=1}^\nu n^{-j/2}\hat p_j(x)\phi(x)\right|>Cn^{-(\nu+1)/2}\right)=O(n^{-\lambda})$$
$$P\left(\max_{1\leq j\leq\nu}\sup_{-\infty<x<\infty}(1+|x|)^{-(3j-1)}|\hat p_j(x)|>C\right)=O(n^{-\lambda})$$
and
\begin{equation}
P\left(\max_{1\leq j\leq\nu}\sup_{-\infty<x<\infty}(1+|x|)^{-(3j-1)}|\hat p_j'(x)|>C\right)=O(n^{-\lambda})\label{poly derivative}
\end{equation}
where $\mathcal X_n=\{\mathbf X_1,\ldots,\mathbf X_n\}$ denotes the data.
\label{Edgeworth bootstrap}
\end{theorem}
Theorem \ref{Edgeworth bootstrap} is slightly strengthened from \cite{hall2013bootstrap} in that we put \eqref{poly derivative} as an additional conclusion. The proof in \cite{hall2013bootstrap} works for any polynomial $p_j$ of degree $3j-1$ and thus also $p_j'$, and the additional implication is immediate (and useful for our proof of Theorem \ref{main Edgeworth}).

\subsection{Subsampling}\label{sec:subsampling proofs}
Consider $\mathbb G_{n,k}^*=\sqrt k(P_k^*-\hat P_n)$ where $P_k^*$ is the bootstrap empirical distribution constructed using $k$ resampled values from $X_1,\ldots,X_n$ by sampling with replacement. Let $\mathcal F_\delta=\{f-g:f,g\in\mathcal F,\ \rho_P(f-g)<\delta\}$, where $\rho_P(f-g):=(Var_P(f(X)-g(X)))^{1/2}$ is the canonical metric. The following result is useful for analyzing Cheap $m$-out-of-$n$ Bootstrap and Cheap Bag of Little Bootstraps:
\begin{theorem}[Adapted from \cite{wellner2013weak} Theorem 3.6.3]
Let $\mathcal F$ be a Donsker class of measurable functions such that $\mathcal F_\delta$ is measurable for every $\delta>0$. Then
$$\sup_{h\in BL_1(\ell^\infty(\mathcal F))}|E_Mh(\mathbb G_{n,k_n}^*)-Eh(\mathbb G_P)|\stackrel{p}{\to}0$$
as $n\to\infty$, for any sequence $k_n\to\infty$, where $E_M[\cdot]$ denotes the expectation conditional on the data $X_1,X_2,\ldots,X_n$.\label{vary size}
\end{theorem}

The following is useful for analyzing Cheap Subsampled Double Bootstrap:
\begin{theorem}[Adapted from \cite{sengupta2016subsampled} Theorem 1]
Let $\mathcal F$ be a Donsker class of measurable functions such that $\mathcal F_\delta$ is measurable for every $\delta>0$. Then the Subsampled Double Bootstrap process defined by
 $\mathbb G_{SDB,n,s}^*=\sqrt n(P_n^{**}-P_s^*)$ satisfies
 $$\sup_{h\in BL_1(\ell^\infty(\mathcal F))}|E_Mh(\mathbb G_{SDB,n,s}^*)-Eh(\mathbb G_P)|\stackrel{p}{\to}0$$
as $\min(n,s)\to\infty$, where $E_M[\cdot]$ is with respect to the randomness of both the first and second-layer resampling in Subsampled Double Bootstrap, conditional on the data $X_1,X_2,\ldots,X_n$. \label{thm SDB}
\end{theorem}

\end{appendix}

\end{document}